\documentclass[fleqn]{CSML}

\pdfoutput=1

\usepackage{lastpage}
\newcommand{\dOi}{13(3:7)2017}
\lmcsheading{}{1--\pageref{LastPage}}{}{}%
{Oct.~03, 2016}{Jul.~28, 2017}{}
\usepackage[hang]{footmisc}
\makeatletter
\ifFN@para
\else
  \long\def\@makefntext#1{%
    \ifFN@hangfoot
      \bgroup
      \setbox\@tempboxa\hbox{%
        \ifdim\footnotemargin>0pt
          \hb@xt@\footnotemargin{\@makefnmark\hss}%
        \else
          \@makefnmark\hskip-\footnotemargin      %%Changed here
        \fi
      }%
      \leftmargin\wd\@tempboxa
      \rightmargin\z@
      \linewidth \columnwidth
      \advance \linewidth -\leftmargin
      \parshape \@ne \leftmargin \linewidth
      \footnotesize
      \@setpar{{\@@par}}%
      \leavevmode
      \llap{\box\@tempboxa}%
      \parskip\hangfootparskip\relax
      \parindent\hangfootparindent\relax
    \else
      \parindent1em
      \noindent
      \ifdim\footnotemargin>\z@
        \hb@xt@ \footnotemargin{\hss\@makefnmark}%
      \else
        \ifdim\footnotemargin=\z@
          \llap{\@makefnmark}%
        \else
          \llap{\hb@xt@ -\footnotemargin{\@makefnmark\hss}}%
        \fi
      \fi
    \fi
    \footnotelayout#1%
    \ifFN@hangfoot
      \par\egroup
    \fi
  }
\fi
\makeatother

\newtheorem*{claim}{Claim}

\newcommand{\orddownarrow}{\mathord{\downarrow}}

\newcommand{\textqt}[1]{``#1''}

\usepackage{amsmath}
\usepackage{amsfonts}
\usepackage{amssymb}
\usepackage{graphicx}
\usepackage{hyperref}
\hypersetup{hidelinks}

\keywords{Rank-width, quasi-tree, join-tree, ordered tree, algebra,
regular term, monadic second-order logic.}

\ACMCCS{[{\bf Theory of computation}]: Logic; Decision
problems ; [{\bf Mathematics of
computing}]: Trees; Model theory;}
\amsclass{05C05, 68R10}

\begin{document}

\title[Generalized trees]{Algebraic and logical descriptions\\
of generalized trees}
\author[B.~Courcelle]{Bruno Courcelle}
\address{LaBRI, CNRS, 351 Cours de la Lib\'{e}ration, 33405 Talence, France}
\email{courcell@labri.fr}
\thanks{This work has been supported by the French National Research Agency
(ANR) within the IdEx Bordeaux program \textqt{Investments for the future}, CPU,
ANR-10-IDEX-03-02.}

\begin{abstract}
\noindent \emph{Quasi-trees} generalize trees in that the unique ``path''
between two nodes may be infinite and have any countable order type. They
are used to define the rank-width of a countable graph in such a way that it
is equal to the least upper-bound of the rank-widths of its finite induced
subgraphs. $\emph{Join-trees}$ are the corresponding directed trees. They
are useful to define the modular decomposition of a countable graph. We also
consider \emph{ordered join-trees}, that generalize rooted trees equipped
with a linear order on the set of sons of each node. We define algebras
with finitely many operations that generate (via infinite terms) these
generalized trees. We prove that the associated regular objects (those
defined by regular terms) are exactly the ones that are the unique models
of monadic second-order sentences. These results use and generalize a
similar result by W.\ Thomas for countable linear orders.
\end{abstract}

\maketitle

%required

\section*{Introduction}

We define and study countable generalized trees, called \emph{quasi-trees,}
such that the unique \textqt{path} between two nodes may be infinite and have any
order type, in particular that of rational numbers. Our motivation comes
from the notion of \emph{rank-width}, a complexity measure of finite graphs
investigated first in \cite{Oum} and \cite{OumSey}. Rank-width is based on
graph decompositions formalized with finite undirected trees of maximal
degree at most 3. In order to extend it to countable graphs in such a way
that \emph{the compactness property} holds, \emph{i.e.}, that the rank-width
of a countable graph is the least upper bound of those of its finite induced
subgraphs, we base decompositions on \emph{quasi-trees}\footnote{Compactness does not hold if one uses trees. For a comparison, the natural
extension of tree-width to countable graphs has the compactness property 
\cite{KriTho} without needing quasi-trees.} \cite{Cou14}. Quasi-trees arise
as least upper bounds of increasing sequences of finite trees, $%
H_{1}\subseteq H_{2}\subseteq \dots\subseteq H_{n}\subseteq \dots,$ where $%
H_{i+1}$ is obtained from $H_{i}$ by the addition of a new node, either
linked to an existing one by a new edge or inserted on an existing edge. If
one inserts infinitely many nodes on an edge of some $H_{i}$, then, the
least upper bound is not a tree but a quasi-tree.

\emph{Join-trees} can be seen as directed quasi-trees. A join-tree is a
partial order $(N,\leq )$ such that every two elements have a least upper
bound (called their \emph{join}) and each set $\{y\mid y\geq x\}$ is
linearly ordered. The modular decomposition of a countable graph is based on
an \emph{ordered join-tree} \cite{CouDel}.

Our objective is to obtain finitary descriptions (usable in algorithms) of
generalized trees that are of the following three types: join-trees,
ordered join-trees and quasi-trees. For this purpose, we will define, for
each type of generalized tree, an algebra based on finitely many operations
such that the finite and infinite terms over these operations define all
generalized trees of this type. The \emph{regular generalized trees} are
those defined by \emph{regular terms}, \emph{i.e.} that have finitely many
different subterms, equivalently, that are the unique solutions of certain
finite equation systems. We will prove that a generalized tree is regular if
and only if it is \emph{monadic second-order definable}, \emph{i.e.}, is the
unique finite or countable model (up to isomorphism) of a monadic
second-order sentence.

As a special case, we have linear orders. A countable linear order whose
elements are labelled by letters from a finite alphabet is called an \emph{%
arrangement}. The linear order of a \emph{regular arrangement} is the
left-right order of the leaves of the tree representing a \emph{regular term}%
, equivalently, the lexicographic ordering of the words of a regular
language. Regular arrangements were first defined and studied in \cite%
{Cou78} and \cite{Hei}, and their monadic second-order definability was
proved in \cite{Tho86}. We will use the latter result for proving its
extension to our generalized trees.

The study of regular linear orders has been continued by Bloom and \'{E}sik
in \cite{BE03,BE05}. They have also studied the \emph{algebraic linear
orders}, defined similarly from \emph{algebraic trees} (infinite terms that
are solutions of certain first-order equation systems, cf. \cite{Cou83}) or
equivalently, as lexicographic orderings of the words of deterministic
context-free languages \cite{BE10,BE11}.

In Sections 1 and 2, we review definitions and basic results. In Section 3,
we first study binary join-trees and then, we extend the definitions and
results concerning them to all join-trees. In Section 4, we study ordered
join-trees, and, in Section 5, we study quasi-trees. An introductory
article on these results is \cite{Cou15}.

\section{Orders, trees and terms}

%\bigskip

All sets, trees and logical structures are finite or countably infinite. We
denote by $X\uplus Y$ the union of sets $X$ and $Y$ if they are disjoint.\
Isomorphism of ordered sets, trees and other logical structures is denoted
by $\simeq $. The restriction of a relation $R$ or a function $f$ defined on
a set $V$ to a subset $W$ of $V$ is denoted by $R\upharpoonright W$ or $%
f\upharpoonright W$ respectively.

For partial orders $\leq,\preceq,\sqsubseteq$, \dots we denote respectively by 
$<,\prec,\sqsubset$, \dots the corresponding strict orders and $X<Y$ means
that $x<y$ for every $x\in X$ and $y\in Y$.

Let $(V,\leq)$ be a partial order. The least upper bound of $x$ and $y$ is
denoted by $x\sqcup y$ if it exists and is called their \emph{join}. The
notation $x\bot y$ means that $x$ and $y$ are incomparable. A \emph{\mbox{line\/}%
\footnote{In \cite{Cou14} we call \emph{line} a linearly ordered subset, without
imposing the convexity property.} }is a subset $Y$ of $V$ that is linearly
ordered and satisfies the following \emph{convexity property}: if $x,z\in Y$%
, $y\in V $ and $x\leq y\leq z$, then $y\in Y$. Particular notations for
convex sets (not necessarly linearly ordered) are $[x,y]$ denoting $\{z\mid
x\leq z\leq y\}$, $\mathopen{]}x,y]$ denoting $\{z\mid x<z\leq y\}$, $\mathopen{]}-\infty,x]$
denoting $\{y\mid y\leq x\}$ (even if $V$ is finite), $\mathopen{]}x,+\infty\mathclose{[}$
denoting $\{y\mid x<y\}$ etc. If $X\subseteq V$, then $\mathord\downarrow(X)$ is the
union of the sets $\mathopen]-\infty,x]$ for $x$ in $X$.

The first infinite ordinal and the linear order $(\mathbb{N},\leq )$ are
denoted by $\omega $.

Let $A$ be a finite set that is linearly ordered by $\leq $, and $A^{\ast }$
be the set of finite words over $A$; the empty word is $\varepsilon $. This
set is linearly ordered by the lexicographic order $\leq _{lex}$ defined by 
$u\leq _{lex}v$ if and only if $v=uw$ or, $u=wax$ and $u=wby$ for some $%
w,x,y $ in $A^{\ast }$ and $a,b$ in $A$ such that $a<b$. Every finite or
countable linear order is isomorphic to $(L,\leq _{lex})$ for some set $%
L\subseteq \{0,1\}^{\ast }$ that is \emph{prefix-free}, which means that, if 
$u,uv\in L$ where $v\in \{0,1\}^{\ast },$ then $v=\varepsilon $ (Theorem 1.7
of \cite{Cou78}). The case where $L$ is regular has been studied in \cite%
{BE03,BE05,Cou78,Hei,Tho86}.

\subsection{Trees}

A \emph{tree} is a possibly empty, finite or countable, undirected graph
that is connected and has no cycles. Hence, it has neither loops nor
parallel edges (it has no two edges with same ends). The set of nodes of a
tree $T$ is $N_{T}$.

A \emph{rooted tree} is a nonempty tree equipped with a distinguished node
called its \emph{root}. The \emph{level }of a node $x$ is the number of
edges of the path between it and the root and $Sons(x)$ denotes the set of
its sons. We define on $N_{T}$ the partial order $\leq _{T}$ such that $%
x\leq _{T}y$  if and only if $y$ is on the unique path between $x$ and the
root. The least upper bound of $x$ and $y$, denoted by $x\sqcup _{T}y,$ is
their least common ancestor. We will specify a rooted tree $T$ by $%
(N_{T},\leq _{T})$ and we will omit the index $T$ when the considered tree
is clear. For a node $x$ of $T$, the \emph{subtree issued from} $x,$ denoted
by $T/x,$ is defined as $(N_{T/x},\leq _{T}\upharpoonright N_{T/x})$ where $%
N_{T/x}:= \mathopen]-\infty, x]$.

A partial order $(N,\leq)$ is $(N_{T},\leq_{T})$ for some rooted tree $T$ if
and only if it has a largest element $max$ and for each $x\in N$, the set $%
[x,max]$ is finite and linearly ordered. These conditions imply that any two
nodes have a join.

An \emph{ordered tree} is a rooted tree such that each set $Sons(x)$ is
linearly ordered by an order $\sqsubseteq_{x}$.

\subsection{Finite and infinite terms}

%\bigskip

Let $F$ be a finite set of operations, each $f$ in $F$ being given with an
arity $\rho (f)$. We call $(F,\rho )$ a \emph{signature}. The maximal
arity of a symbol is denoted by $\rho (F)$. A \emph{term over} $F$ is finite
or infinite. We denote by $T^{\infty }(F)$ the set of all terms over $F$\
and by $T(F)$ the set of finite ones. A typical example of an infinite
term, easily describable linearly, is, with $f$ binary and $a$ and $b$
nullary, the term $t_{\infty }:=f(a,f(b,f(a,f(b,f(\dots))))))$ that is the
unique solution in $T^{\infty }(F)$ of the equation $t=f(a,f(b,t))$.

Positions in terms are designated by Dewey words\footnote{For the term $t=f(a,g(b,c))$ taken as an example, $\varepsilon $ denotes
the occurrence of $f$ and the unique occurrences of $a,g,b$ and $c$ are
denoted respectively by the words 1,2,21 and 22.} over $\{1,\dots,\rho (F)\}$
considered as an alphabet. The set $Pos(t)$ of positions of a term $t$ is
ordered by $\leq _{t}$, the reversal of the prefix order on words. A term $%
t$ can be seen as a labelled, ordered and rooted tree whose set of nodes is $%
Pos(t)$. We have $Pos(t_{\infty })=2^{\ast }\uplus 2^{\ast }1$, where $%
2^{\ast }$ is the set of occurrences of $f$, $(22)^{\ast }1$ is the set of
occurrences of $a$ and $(22)^{\ast }21$ is that of $b$.

There is a canonical $F$-algebra structure on $T^{\infty }(F)$, of which $%
T(F)$ is a subalgebra. If $\mathbb{M}=\left\langle M,(f_{\mathbb{M}%
})_{f\in F}\right\rangle $ is an $F$-algebra, a \emph{value mapping} is a
homomorphism $h:T^{\infty }(F)\rightarrow \mathbb{M}.$ Its restriction to
finite terms is uniquely defined. In some cases, we will use algebras with
two sorts. The corresponding modifications of the definitions are
straightforward, see \cite{Gog} for details.

\subsection*{The partial order on terms.}

Let $F$ contain a special nullary symbol $\Omega$ intended to be the least
term. We define on $T(F)$ a partial order $\underline{\ll}$ by the following induction: $\Omega  \underline{\ll } t$ for any $t\in T(F)$, and $f(t_{1},\dots,t_{k}) \underline{\ll } g(t_{1}^{\prime
},\dots,t_{k^{\prime }}^{\prime })$ if and only if $k=k^{\prime }$, $f=g$ and 
$t_{i}$ $\underline{\ll }$ $t_{i}^{\prime }$ for $i=1,..,k$.

For terms in $T^{\infty}(F)$, the definition (subsuming the previous one) is:    $t$ $\underline{\ll }$ $t^{\prime }$ if and only if $Pos(t)\subseteq
Pos(t^{\prime })$ and every occurrence in $t$ of a symbol in $F-\{\Omega \}$
is an occurrence in $t^{\prime }$ of the same symbol (an occurrence in $t$
of $\Omega $ is an occurrence in $t^{\prime }$ of any symbol).

Every increasing sequence of terms has a least upper bound. More details on
this order can be found in \cite{Cou83,Gog}.

If $\mathbb{M}=\left\langle M,(f_{\mathbb{M}})_{f\in F}\right\rangle $ is a
partially ordered $F$-algebra, whose order is $\omega $-complete (increasing
sequences have least upper bounds) and whose operations are $\omega $%
-continous (they preserve such least upper bounds), then, one can define the
value in $\mathbb{M}$ of an infinite term as the least upper bound of the
values of the finite smaller terms \cite{Cou83,Gog}. However, this approach
fails for our algebras of generalized trees, because no appropriate partial
order can be defined, as proved in Section 6 of \cite{Cou78}. Instead of
orders, this article uses category theory. This categorical setting could
be used here but direct constructions of generalized trees from terms (in
Definitions \ref{D:3.15}, \ref{D:3.28} and \ref{D:4.9}) are simpler and better formalizable in
logic.

\subsection*{Regular terms}

A term $t\in T^{\infty}(F)$ as \emph{regular} if there is a mapping $h$ from 
$Pos(t)$ into a finite set $Q$ and a mapping $\tau:Q\rightarrow F\times
Seq(Q)$ (where $Seq(Q)$ denotes the set of finite sequences of elements of $%
Q $) such that, if $u$ is an occurrence of a symbol $f$ of arity $k$, then $\tau
(h(u))=(f,(h(u_{1}),\dots,h(u_{k})))$ where $(u_{1},\dots,u_{k})$ is the
sequence of sons of $u$.

Intuitively, $\tau$ is the transition function of a top-down deterministic
automaton with set of states $Q$; $h(\varepsilon)$ is its initial (root)
state and $h$ defines its unique run. This is equivalent to requiring that $%
t$ has finitely many different subterms, or is a component of a finite
system of equations that has a unique solution in $T^{\infty}(F)$. (The set 
$Q$ can be taken as the set of unknowns of such a system, see \cite{Cou83}.)

The above term $t_{\infty}$ is regular with $Q:=\{1,2,3,4\}$, $%
\tau(1)=(f,(2,3))$, $\tau(2)=(a,()),\tau(3)$ $=(f,(4,1)),\tau(4)=(b,()).$

%\bigskip

We associate with a term $t$ the relational structure $\left\lfloor
t\right\rfloor :=(Pos(t),\leq _{t},(br_{i})_{1\leq i\leq \rho (F)},$ $%
(lab_{f})_{f\in F})$ where $br_{i}(u)$ is true if and only if $u$ is the $%
i $-th son of his father and $lab_{f}(u)$ is true if and only if $f$ occurs
at position $u$. A term $t$ can be reconstructed in a unique way from any
relational structure isomorphic to $\left\lfloor t\right\rfloor .$

A term $t$ is regular if and only if $\left\lfloor t\right\rfloor $ is MS
definable, \emph{i.e.}, is, up to isomorphism, the unique model of a monadic
second-order sentence. This result is due to Rabin \cite{Rab}, see Thomas 
\cite{Tho90}.

\subsection{Arrangements and labelled sets}

We review a notion introduced in \cite{Cou78} and further studied in \cite%
{Hei,Tho86}. Let $X$ be a set. A linear order $(V,\leq )$ equipped with a
labelling mapping $lab:V\rightarrow X$ is called an \emph{arrangement} \emph{over} $X$. It is \emph{simple} if $lab$ is injective. We denote by $%
\mathcal{A}(X)$ the set of arrangements over $X$. We will generalize
arrangements to tree structures.

An arrangement over a finite set $X$ considered as an alphabet can be
considered as a generalized word. A linear order $(V,\leq )$ is identified
with the simple arrangement $(V,\leq ,Id_{V})$ such that $Id_{V}(v):=v$ for
each $v\in V$. In the sequel, $Id$ denotes the identity function on any set.

An \emph{isomorphism of arrangements} $i:(V,\leq,lab)\rightarrow(V^{\prime
},\leq^{\prime},lab^{\prime})$ is an order preserving bijection $%
i:V\rightarrow V^{\prime}$ such that $lab^{\prime}\circ i=lab.$ Isomorphism
is denoted by $\simeq$.

If $w=(V,\leq ,lab)\in \mathcal{A}(X)$ and $r:X\rightarrow Y$, then, $%
\overline{r}(w):=(V,\leq ,r\circ lab)$ is an arrangement over $Y$. If $r$
maps $V$ into $Y$, then $\overline{r}((V,\leq ))$ is the arrangement $%
(V,\leq ,r)$ over $Y$ since we identify $(V,\leq )$ to the simple
arrangement $(V,\leq ,Id)$.

%\bigskip

The concatenation of linear orders yields a concatenation of arrangements
denoted by $\bullet $. We denote by $\Omega $ the empty arrangement and by $%
a$ the one reduced to a single occurrence of $a\in X$. Clearly, $w\bullet
\Omega =\Omega \bullet w=w$ for every $w\in \mathcal{A}(X)$. The infinite
word $w=a^{\omega }$ is the arrangement over $\{a\}$ with underlying order $%
\omega $; it is described by the equation $w=a\bullet w$. Similarly, the
arrangement $w=a^{\eta }$  over $\{a\}$ with underlying linear order $(%
\mathbb{Q},\leq )$ (that of rational numbers) is described by the equation $%
w=w\bullet (a\bullet w)$.

Let $X$ be a set of first-order variables (they are nullary symbols) and $%
t\in T^{\infty }(\{\bullet ,\Omega \}\cup X).$ Hence, $Pos(t)\subseteq
\{1,2\}^{\ast }$. The \emph{value }of $t$ is the arrangement $val(t):=(%
\mathrm{Occ}(t,X),\leq _{lex},lab)$ where $\mathrm{Occ}(t,X)$ is the set of
positions of the elements of $X$ and $lab(u)$ is the symbol of $X$ occurring
at position $u$. We say that $t$ \emph{denotes }$w$ if $w$ is isomorphic to 
$val(t),$ and that $w$ is the \emph{frontier} of the syntactic tree of $t$ 
\cite{Cou78}.

For an example, $t_{\bullet }:=\bullet (a,\bullet (b,\bullet (a,\bullet
(b,\bullet (\dots\dots\dots))))))$ denotes the infinite word $abab\dots$. Its
value is defined from $\mathrm{Occ}(t_{\bullet },\{a,b\})=2^{\ast }1$,
lexicographically ordered (we have $1<21<221<\dots)$ by taking $lab(2^{i}1):=a$
if $i$ is even and $lab(2^{i}1):=b$ if $i$ is odd. The arrangements $%
a^{\omega }$ and $a^{\eta }$ are denoted respectively by $t_{1}$ and $t_{2}$
that are the unique solutions in $T^{\infty }(\{\bullet ,\Omega ,a\})$ of
the equations $t_{1}=a\bullet t_{1}$ and $t_{2}=t_{2}\bullet (a\bullet
t_{2}) $.

%\bigskip

An arrangement is \emph{regular} if it is denoted by a regular term. The
term $t_{\bullet }$ is regular. The arrangements $a^{\omega }$ and $%
a^{\eta }$ are regular\footnote{The subalgebra of regular arrangements is characterized by equational axioms
in \cite{BE05}.}.

An arrangement is regular if and only if it is a component of the \emph{%
\mbox{initial\/}\footnote{in the sense of category theory, see \cite{Cou78}.} solution of a regular
system of equations} over $F$, or also, the value of a \emph{regular
expression} in the sense of \cite{Hei}. We will use the result of \cite%
{Tho86} that an arrangement over a finite alphabet is regular if and only if
is \emph{monadic second-order \mbox{definable\/}\footnote{The article \cite{Car} establishes that a set of arrangements is
recognizable if and only if it is monadic second-order definable.}}. (We
review monadic second-order logic in the next section). For this result, we
represent an arrangement $w=(V,\leq ,lab)$ over a finite set $X$ by the
relational structure $\left\lfloor w\right\rfloor :=(V,\leq ,(lab_{a})_{a\in
X})$ where $lab_{a}(u)$ is true if and only if $lab(u)=a$.

If $r$ maps $X$ to $Y$ and $w\in \mathcal{A}(X)$ is regular, then $%
\overline{r}(w)$ is regular. This is clear from the definitions because the
substitution of $r(x)$ for $x\in X$ in a regular term in $T^{\infty
}(\{\bullet ,\Omega \}\cup X)$ yields a regular term \cite{Cou83}.

%\bigskip

An $X$-\emph{labelled set} is a pair $m=(V,lab)$ where $lab:V\rightarrow X$%
, or, equivalently, a relational structure $(V,(lab_{a})_{a\in X})$ where
each element of $V$ belongs to a unique set $lab_{a}$. We denote by $set(w)$
the $X$-labelled set obtained by forgetting the linear order of an
arrangement $w$ over $X$. Up to isomorphism, an $X$-labelled set $m$ is
defined by the cardinalities in $\mathbb{N}\cup \{\omega \}$ of the sets $%
lab_{a}$, hence is a finite or countable \emph{multiset of elements of} $X$,
in other words, a mapping that indicates for each $a\in X$ the number, in $%
\mathbb{N}\cup \{\omega \},$ of its occurrences in $m$.

If $X$ is finite, each $X$-labelled set is MS$_{\mathit{fin}}$-definable, 
\emph{i.e.}, is the unique, finite or countably infinite model up to
isomorphism of a sentence of \emph{monadic second-order logic} extended with
a set predicate $Fin(U)$ expressing that a set $U$ is finite. It is also 
\emph{regular}, \emph{i.e.}, is $set(val(t))$ for some regular term in $%
T^{\infty }(\{\bullet ,\Omega \}\cup X).$ The notion of regularity is thus
trivial for $X$-labelled sets when $X$ is finite.

\section{Monadic second-order logic and related notions.}

\emph{Monadic second-order \mbox{logic\/}\footnote{\emph{MS} will abreviate\emph{\
monadic second-order} in the sequel.}} extends first-order logic by the use
of \emph{set variables} $X,Y,Z$ \dots denoting subsets of the domain of the
considered logical structure, and the atomic formulas $x\in X$ expressing
membership of $x$ in $X$. We call \emph{first-order} a formula where set
variables are not quantified. For example, a first-order formula can express
that $X\subseteq Y$. A \emph{sentence} is a formula without free variables.

Let $\mathcal{R}$ be a finite set of relation symbols, each symbol $R$
being given with an arity $\rho (R)$. We call it a \emph{relational signature%
}. For every set of variables $\mathcal{W}$, we denote by $MS(\mathcal{R},%
\mathcal{W})$ the set of MS\ formulas written with $\mathcal{R}$ and free
variables in $\mathcal{W}$. An $\mathcal{R}$-\emph{structure} is a tuple $%
S=(D_{S},(R_{S})_{R\in \mathcal{R}})$ where $D_{S}$ is a finite or countably
infinite set, called its \emph{domain,} and each $R_{S}$ is a relation on $%
D_{S}$ of arity $\rho (R)$. A property $P$ of $\mathcal{R}$-structures is 
\emph{monadic second-order expressible} if it is equivalent to the validity,
in every $\mathcal{R}$-structure $S$, of a monadic second-order sentence $%
\varphi $, which we denote by $S\models \varphi $.

For example, a graph $G$ without parallel edges can be identified with the
\{$edg$\}-structure $(V_{G},edg_{G})$ where $V_{G}$ is its vertex set and $%
edg_{G}(x,y)$ means that there is an edge from $x$ to $y,$ or between $x$
and $y$ if $G$ is undirected. To take an example, 3-colorability is
expressed by the MS sentence:
\begin{multline*}
\exists X,Y\big[X\cap Y=\emptyset\wedge\lnot\exists u,v(edg(u,v)\wedge
[(u\in X\wedge v\in X)\vee \\
(u\in Y\wedge v\in Y)\wedge(u\notin X\cup Y\wedge v\notin
X\cup Y)])\big].
\end{multline*}

Many properties of partial orders $(N,\leq )$ can also be expressed by MS\
formulas. Here are examples that will be useful in our proofs.

\begin{enumerate}[label=(\alph*)]
\item The formula $Lin(X)$ defined as $\forall x,y.[(x\in X\wedge y\in
X)\Longrightarrow (x\leq y\vee y\leq x)]$ expresses that a subset $X$ of $N$,
partially ordered by $\leq$, is linearly ordered.

\item The formula $Lin(X)\wedge \exists a,b.[\min (X,a)\wedge \max
(X,b)\wedge \theta (X,a,b)]$ expresses that $X$ is linearly ordered and
finite, where $\min (X,a)$ and $\max (X,b)$ are first-order formulas
expressing respectively that $X$ has a least element $a$ and a largest one $%
b$, and $\theta (X,a,b)$ is an MS formula expressing, (1) that each element $x$ of $X$ except $b$ has a successor $c$ in $X$ (i.e., $c$
is the least element of $\{y\in X\mid y>x\}$), and (2), that  $(a,b)\in Suc^{\ast },$ where $Suc$ is the above defined
successor relation (depending on $X$) and $Suc^{\ast }$ is its reflexive
and transitive closure.
\end{enumerate}

\noindent Property (b) is expressed by the MS\ formula:
\[
\forall U\left[U\subseteq X\wedge a\in U\wedge\forall x,y((x\in U\wedge(x,y)\in
Suc)\Longrightarrow y\in U) \Longrightarrow b\in U\right].
\]
First-order formulas expressing $U\subseteq X$, $(x,y)\in Suc$ and Property (a) are easy to build. 
Without a linear order, the finiteness of a set $X$\
is not MS expressible. It is thus useful, in some cases, to enrich MS\
logic with a \emph{finiteness predicate} $Fin(X)$ expressing that the set $X$
is finite. We denote by MS$_{\mathit{fin}}$ the corresponding extension of
MS logic.

%\bigskip

If $S$ is a relational structure $(N,\leq ,(br_{i})_{1\leq i\leq \rho
(F)},(lab_{f})_{f\in F})$ isomorphic to the structure $\left\lfloor
t\right\rfloor $ representing a term $t\in T^{\infty }(F)$, then a linear
order $\sqsubseteq $ on $N$ is definable by a first-order formula as follows:
\begin{multline*}
x\sqsubseteq y:\Longleftrightarrow x\leq y\vee \text{($x\bot y$ and $x$ is below
the $i$-th son of } x\sqcup y \\
\text{ and $y$ is below the $j$-th son of $x\sqcup y$ where $ i<j$)}.
\end{multline*}
The definability of linear orders by MS formulas is studied in \cite{BluCou}.

%\bigskip

\emph{Monadic second-order transductions} are transformations of logical
structures specified by MS or MS$_{\mathit{fin}}$ formulas. We will use them
in the proofs of Theorems \ref{T:3.21}, \ref{T:3.30}, \ref{T:4.11}, \ref{T:5.7} and \ref{T:5.10}. For
these proofs, we will only need very simple transductions, said to be \emph{%
noncopying and parameterless} in \cite{CouEng}. We call them \emph{%
MS-transductions}.

Let $\mathcal{R}$ and $\mathcal{R}^{\prime}$ be two relational signatures.
A \emph{definition scheme} of type $\mathcal{R}\rightarrow\mathcal{R}%
^{\prime}$ is a tuple of formulas of the form $\mathcal{D}%
=\langle\chi,\delta,(\theta _{R})_{R\in\mathcal{R}^{\prime}}\rangle$ such
that $\chi\in MS(\mathcal{R})$, $\delta\in MS(\mathcal{R},\{x\})$ and $%
\theta_{R}\in MS(\mathcal{R},\{x_{1},\dots,x_{\rho(R)}\})$ for each $R$ in $%
\mathcal{R}^{\prime}$. We define $\widehat{\mathcal{D}}(S):=S^{%
\prime}=(D_{S^{\prime}},(R_{S^{\prime}})_{R\in\mathcal{R}^{\prime}})$  as
follows:

\begin{itemize}
\item $S^{\prime}$ is defined if and only if $S\models\chi,$

\item $D_{S^{\prime }}$ is the set of elements $d$ of $D_{S}$ such that $S\models \delta (d),$

\item $R_{S^{\prime }}$ is the set of tuples $(d_{1},\dots ,d_{\rho(R)}) $ of elements of $D_{S}$ such that 
$S\models \theta _{R}(d_{1},\dots ,d_{\rho (R)})$.
\end{itemize}
Our main tool is the following (well-known) result:

\begin{thm}\label{T:2.1} 
Let $\mathcal{D}$ be a definition scheme as above
and $\varphi \in MS_{\mathit{fin}}(\mathcal{R}^{\prime },\mathcal{W})$.\
There exists a formula $\varphi ^{\mathcal{D}}\in MS_{\mathit{fin}}(\mathcal{%
R},\mathcal{W})$ such that, for every $\mathcal{R}$-structure $S$, for every 
$\mathcal{W}$-assignment $\nu $ in $D_{S}$, we have $(S,\nu )\models \varphi
^{\mathcal{D}}$ if and only if 
\begin{enumerate}
\item
$S\models\chi$ (so that $\widehat{\mathcal{D}}(S)=S^{\prime}$ is
well-defined), 

\item
$\nu $ is an $\mathcal{W}$-assignment in $D_{S^{\prime }}$ (\emph{i.e.}, $\nu (x)\in D_{S^{\prime }}$ and $\nu (X)\subseteq D_{S^{\prime}}$ for $x,X\in \mathcal{W}$) and 

\item $(S^{\prime },\nu )\models \varphi .$
\end{enumerate}

\end{thm}

\begin{proof} 
The proof is given in \cite{CouEng} (\emph{Backwards
Translation Theorem}, Theorem 7.10) for finite structures, hence the
finiteness predicate $\mathit{Fin}(X)$ is of no use. However, it works as
well for infinite structures and formulas written with the predicate $\mathit{Fin}(X)$ that translates back to itself (under the assumption that $%
\nu (X)\subseteq D_{S^{\prime }}$).

The formula $\varphi ^{\mathcal{D}}$ is the conjunction of $\chi $, a
formula expressing Property (2) and a formula $\varphi ^{\prime }$ obtained from $%
\varphi $ by replacing each atomic formula $R(x_{1},\dots,x_{r})$ by $\theta
_{R}(x_{1},\dots ,x_{\rho (R)}),$ that is, by its definition given by $%
\mathcal{D}$. If $\varphi $ is a sentence, then $\mathcal{W}=\emptyset $
and Property (2) is trivially true. 
\end{proof}

%\bigskip

It follows that, if the monadic theory of a class of structures $\mathcal{S}$
is decidable (which means that one can decide whether a given sentence is
true in all structures of $\mathcal{S}$) and $\mathcal{S}^{\prime }=\widehat{%
\mathcal{D}}(\mathcal{S})$ for some definition scheme $\mathcal{D}$, then
the monadic theory of $\mathcal{S}^{\prime }$ is decidable, because $%
S^{\prime }\models \varphi $ for all $S^{\prime }$ in $\mathcal{S}^{\prime
} $ if and only if $S\models \chi \Longrightarrow \varphi ^{\mathcal{D}}$
for all $S$ in $\mathcal{S}$.

\section{Join-trees}

Join-trees have been used in \cite{CouDel} for defining the modular
decomposition of countable graphs.

\subsection{Join-trees, join-forests and their structurings}

Join-trees are defined as particular partial-orders. Finite nonempty
join-trees correspond to finite rooted trees.

\begin{defi}[Join-tree]\label{D:3.1} \ %
\begin{enumerate}[label=(\alph*),midpenalty=99,beginpenalty=0]
\item A \emph{\mbox{join-tree\/}\footnote{The article \cite{LS} defines a \emph{tree} as a partial order of any
cardinality that satisfies Condition (2). A join-forest is a tree in that
sense.}} is a pair $J=(N,\leq )$ such that:
%\begin{quote}
\begin{enumerate}[label=(\arabic*)]
\item

$N$ is a possibly empty, finite or countable set called the set of \emph{%
nodes},

\item

$\leq$ is a partial order on $N$ such that, for every node $x,$
the set $[x,+\infty\lbrack$ (the set of nodes $y\geq x$) is linearly ordered,

\item

every two nodes $x$ and $y$ have a join $x\sqcup y$.
\end{enumerate}
%\end{quote}

A minimal node is a \emph{leaf}. If $N$ has a largest element, we call it
the \emph{root} of $J$. The set of strict upper bounds of a nonempty set $%
X\subseteq N$ is a line $L$; if $L$ has a smallest element, we denote it by $%
\widehat{X}$ and we say that $\widehat{X}$ is \emph{the top} of $X$. Note
that $\widehat{X}\notin X.$

%\bigskip

\item A \emph{join-forest} is a pair $J=(N,\leq)$ that satisfies Conditions
(1), (2) and the following weakening of (3):
\begin{enumerate}
\item[(3')] if two nodes have an upper bound, they have a join.
\end{enumerate}

%\bigskip

The relation that two nodes have a join is an equivalence. Let $N_{s}$ for $%
s\in S$ be its equivalence classes and $J_{s}:=(N_{s},\leq\upharpoonright
N_{s}),$ more simply denoted by $(N_{s},\leq)$ by leaving implicit the
restriction to $N_{s}$. Then each $J_{s}$ is a join-tree, and $J$ is the
union of these pairwise disjoint join-trees, called its \emph{components.}

%\bigskip

\item A join-forest $J=(N,\leq )$ is \emph{included in} a join-forest $%
J^{\prime }=(N^{\prime },\leq ^{\prime }),$ denoted by $J\subseteq J^{\prime
},$ if $N\subseteq N^{\prime }$, $\leq $ is $\leq ^{\prime }\upharpoonright
N $ and $\sqcup $ is $\sqcup ^{\prime }\upharpoonright N$; if $J$ and $%
J^{\prime }$ are join-trees, we say also that $J$ is a \emph{subjoin-tree of}
$J^{\prime }$.

\end{enumerate}
\end{defi}

\begin{defi}[Direction and degree]\label{D:3.2} 

Let $J=(N,\leq )$ be a join-forest, and $x$ be one of its nodes. Let $\sim $%
 be the equivalence relation on $\mathopen]-\infty ,x[$ such that $z\sim y$ if and
only if $z\sqcup y<x$. Each equivalence class $C$ is called a\emph{direction of}
$J$ \emph{relative to} $x,$ and we have $\widehat{C}=x$. The
set of directions relative to $x$ is denoted by $Dir(x)$ and the \emph{degree%
} of $x$ is the number of its directions. The leaves are the nodes of degree   0.
A join-tree is \emph{binary} if its nodes have degree at most 2. We call it
a \emph{BJ-tree} for short. 
\end{defi}

For concatenating vertically two join-trees, we need that every join-tree
has a distinguished \textqt{branch}, a line, that we call an axis. As we want to
construct join-trees with operations including concatenation, all subtrees
must be of the same type, hence must have axes. It follows that a join-tree
will be partionned into lines, one of them being its axis, the others being
the axes of its subtrees. We call such a partition a structuring.

\begin{defi}[Structured join-trees and join-forests]\label{D:3.3}
  \leavevmode
\begin{enumerate}[label=(\alph*)]
\item Let $J=(N,\leq)$ be a join-tree. A \emph{structuring} of $J$ is a set $%
\mathcal{U}$ of nonempty lines forming a partition of $N$ that satisfies
some conditions, stated with the following notation~: if $x\in N$, then $%
U(x) $ denotes the line of $\mathcal{U}$ containing $x$, $U_{-}(x):=U(x)\cap
\mathopen]-\infty,x\mathclose[$ and  $U_{+}(x):=U(x)\cap\lbrack x,+\infty\mathclose\lbrack$. (The set $%
[x,+\infty\mathclose\lbrack$  has no top but it can have a greatest element). The
conditions are~:

\begin{enumerate}[label=(\arabic*)]
\item
exactly one line $A$ of $\mathcal{U}$ is \emph{upwards closed} (\emph{i.e.%
}, $[x,+\infty\mathclose{\lbrack} \subseteq A$ if $x\in A$), hence, has no strict
upper bound and no top; we call it the \emph{axis}; each other line $U$ has
a top $\widehat{U}$,

\item
 for each $x$ in $N$, the sequence $y_{0},y_{1},y_{2},\dots$ such
that $y_{0}=x,$ $y_{i+1}=\widehat{U(y_{i})}$ is finite; its last element is $%
y_{k}\in A$ ($y_{k+1}$ is undefined) and we call $k$ the $depth$ of $x$.
\end{enumerate}

The nodes on the axis are those at depth 0. The lines $[y_{i},y_{i+1}[$\
for $i\in \lbrack 0,k-1]$ and $[y_{k},+\infty \lbrack $ are convex subsets
of pairwise distinct lines of $\mathcal{U}$. We have
\[
[x,+\infty \mathclose{\lbrack}
=[y_{0},y_{1}\mathclose{[}\uplus \lbrack y_{1},y_{2}\mathclose{[}\uplus \dots\uplus \lbrack
y_{k},+\infty \mathclose\lbrack,
\]
where $[y_{i},y_{i+1}\mathclose[=U_{+}(y_{i})$ for each $i<k$%
, $[y_{k},+\infty \mathclose\lbrack =U_{+}(y_{k})$ $\subseteq A$ and the depth of $%
y_{i}$ is $k-i$.

% \bigskip

We call such a triple $(N,\leq,\mathcal{U})$ a \emph{structured join-tree},
an \emph{SJ-tree} for short.
%
%\bigskip
%
Every linear order is an SJ-tree whose elements are all of depth 0.

\bigskip

\noindent   \textbf{Remark}. If $x<A$ for some $x$, then $A$ has a smallest element,
which is the node $y_{k}$ of Condition 2) (because if $z\in A$ is smaller
than $y_{k}$, then $x<z$, which contradicts the observation that $%
[y_{k-1},y_{k}\mathclose[\subseteq U(y_{k-1})$ and $U(y_{k-1})\cap A=\emptyset ).$

\bigskip

\item   Let $J=(N,\leq)$ be a join-forest whose components are $J_{s}$, $s\in S$%
. A \emph{structuring} of $J$ is a set $\mathcal{U}$ of nonempty lines
forming a partition of $N$ such that, if $\mathcal{U}_{s}$ is the set of
lines of $\mathcal{U}$ included in $N_{s}$ (every line of $\mathcal{U}$ is
included in some $N_{s}$), then each triple $(N_{s},\leq,\mathcal{U}_{s})$
is a structuring of $J_{s}$.
\end{enumerate}

\end{defi}

\begin{exa}\label{E:3.4} 
Figure 1 shows a structuring
 $\{U_{0},\dots,U_{5}\}$ of a binary join-tree. The axis is $U_{0}$. The
directions relative to $x_{2}$ are $U_{0-}(x_{2})\cup U_{1}$  and $%
U_{2}\cup U_{3}$. The maximal depth of a node is 2.
\end{exa}

\begin{figure}[ptb]
\begin{center}
\includegraphics[
%natheight=6.238800in,
%natwidth=5.350600in,
%height=3.4584in,
width=2.9706in
]%
{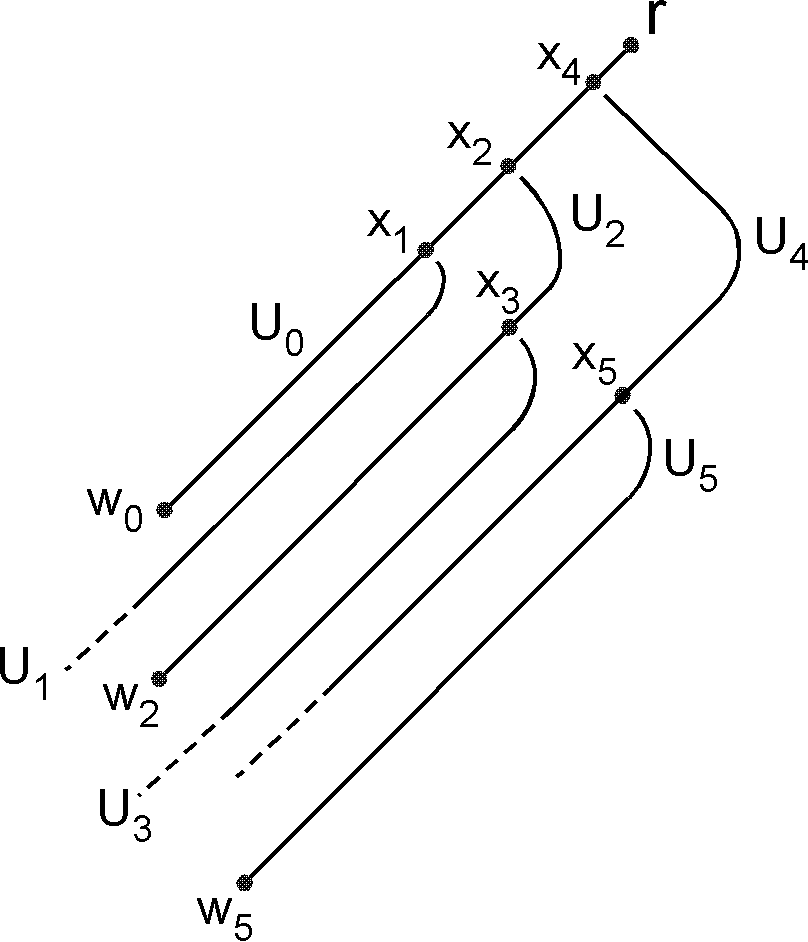}%
\caption{A structured binary join-tree.}%
\end{center}
\end{figure}

\begin{prop}\label{P:3.5} 
Every join-tree and, more generally,
every join-forest has a structuring.
\end{prop}

\begin{proof}
Let $J=(N,\leq)$ be a join-tree. Let us choose an
enumeration of $N$ and a maximal line $B_{0}$ ; it is upwards closed. For
each $i>0$, we choose a maximal line $B_{i}$ containing the first node not
in $B_{i-1}\cup\dots\cup B_{0}$. We define $U_{0}:=B_{0}$ and, for $i>0$, $%
U_{i}:=B_{i}-(U_{i-1}\uplus\dots\uplus U_{0})=B_{i}-(B_{i-1}\cup\dots\cup B_{0})$. We define $\mathcal{U}$ as the set of lines $U_{i}$. It is a structuring
of $J$. The axis is $U_{0}$. If $J$ is a join-forest, it has a structuring that is the union of
structurings of its components.
\end{proof}
\noindent   \textbf{Remark.} Since each line $B_{i}$ is maximal, if $U_{i}$ has a smallest
element, this element is a node of degree 0 in $J$.

In view of our use of monadic second-order logic, we give a
description of SJ-trees by relational structures.

\begin{defi}[SJ-trees as relational structures]\label{D:3.6}

If $J=(N,\leq ,\mathcal{U})$ is an SJ-tree, we define $S(J)$ as the
relational structure $(N,\leq ,N_{0},N_{1})$ such that $N_{0}$ is the set
of nodes at even depth and $N_{1}:=N-N_{0}$ is the set of those at odd
depth. ($N_{0}$ and $N_{1}$ are sets but we will also consider them as unary
relations).
\end{defi}

\begin{prop}\label{P:3.7} 
  \leavevmode
  \begin{enumerate}
  \item There is an MS\ formula $\varphi (N_{0},N_{1})$ expressing that a
    relational structure $(N,\leq ,N_{0},N_{1}) $ is $S(J)$ for some SJ-tree
    $J=(N,\leq ,\mathcal{U})$.

  \item There exist MS\ formulas $\theta _{Ax}(X,N_{0},N_{1})$  and $\theta
    (u,U,N_{0},N_{1})$  expressing, respectively, in a structure $(N,\leq
    ,N_{0},N_{1})=S((N,\leq ,\mathcal{U})),$ that $X$ is the axis and that $U\in
    \mathcal{U}\wedge u=\widehat{U}$.
  \end{enumerate}
\end{prop}

\begin{proof}
Let $J=(N,\leq )$ be a join-tree and $X\subseteq N$. We
say that $X$ is \emph{laminar} if, for all $x,y,z\in X$, if $[x,z]\cup
\lbrack y,z]\subseteq X$ (where $x<z$ and $y<z),$  then $[x,z]\subseteq
\lbrack y,z]$ or $[y,z]\subseteq \lbrack x,z]$ (the intervals $[x,z]$ and $%
[y,z]$ are relative to $J$). This condition implies that the lines of $J$
that are included in $X$ and are maximal with this condition form a
partition of $X$ whose parts will be called its \emph{components}.

It is clear from the definitions that, if $J=(N,\leq ,\mathcal{U})$ is an
SJ-tree and $S(J)=(N,\leq ,$ $N_{0},N_{1})$, then the sets $N_{0}$ and $%
N_{1}$ are laminar, $\mathcal{U}$ is the set of their components and the
axis $A$ is a component of $N_{0}$.

\begin{enumerate}
\item That a partial order $(N,\leq )$ is a join-tree is first-order
expressible. The formula $\varphi (N_{0},N_{1})$ will include this
condition.

Let $J=(N,\leq )$ be a join-tree, $N$ the union of two disjoint laminar sets 
$N_{0}$ and $N_{1},$ and $\mathcal{U}$ the set of their components. Then, $%
J=(N,\leq ,\mathcal{U})$ is an SJ-tree such that $S(J)=(N,\leq ,N_{0},N_{1})$
if and only if:
\begin{enumerate}[label=(\roman*)]

\item every component of $N_{1}$ has a top in $N_{0}$ and every component of $N_{0}$ except one has a top in $N_{1}$,

\item  for each $U$ in $\mathcal{U}$, the sequence $U_{0},U_{1},\dots$  of lines
of $\mathcal{U}$ such that $U_{0}=U$, $\widehat{U_{0}}\in U_{1},\dots,%
\widehat {U_{i}}\in U_{i+1}$ terminates at some $U_{k}$ that has no top,
hence is included in $N_{0}$.
\end{enumerate}

\noindent These conditions are necessary. As they rephrase Definition
\ref{D:3.3}, they are also sufficient. The integer $k$ in Condition (ii) is the
common depth of all nodes in $U$.

That a set $X$ is laminar is first-order expressible, and one can build an
MS formula $\psi (U,X)$ expressing that $U$ is a component of $X$ assumed to
be laminar. This formula can be used to express that $N$ is the union of
two disjoint laminar sets $N_{0}$ and $N_{1}$ that satisfy Conditions (i) and
(ii). For expressing Condition (ii), we define, for each $U$ in $\mathcal{U}$, a set of nodes $W$ as follows: it is the least set such that $\widehat{U}%
\in W,$ and, for each $w\in W$, the top of $U(w)$ belongs to $W$ if it is
defined (where $U(w)$ is the unique set in $\mathcal{U}$ that contains $w$%
). The set $W$ is linearly ordered (it consists of $\widehat{U_{0}}<\dots<%
\widehat{U_{i}}\dots$) and Condition (ii) says that it must be finite. To
write the formula, we use the observation made in Section 2 that the
finiteness of a linearly ordered set is MS expressible.

\item The construction of $\varphi $ actually uses the MS formulas $\theta
_{Ax}(X,N_{0},N_{1})$ and $\theta (u,U,N_{0},N_{1})$.  \qedhere
\end{enumerate}
\end{proof}

\subsection{Description schemes of structured binary join-trees}

In order to introduce technicalities step by step, we first consider binary
join-trees. They are actually sufficient for defining the rank-width of a
countable graph. See Section 5.

\begin{defi}[Structured binary join-trees]\label{D:3.8}

Let $J=(N,\leq )$ be a binary join-tree, a \emph{BJ-tree} in short. A \emph{%
structuring} of $J$ is a set $\mathcal{U}$ of lines satisfying the
conditions of Definition \ref{D:3.3} and, furthermore:
\begin{enumerate}[label=(\roman*)]
\item 
if the axis $A$ has a smallest element, then its degree is 0 or 1,

\item
each $x\in N$ is the top of at most one set $U\in \mathcal{U}$, denoted by $U^{x}$, and $U^{x}:=\emptyset $ if $x$ is the top of no $U\in 
\mathcal{U}$.
\end{enumerate}

We call $(N,\leq,\mathcal{U})$ a \emph{structured binary join-tree}, an 
\emph{SBJ-tree} in short.
\end{defi}

\begin{prop}\label{P:3.9}  
  \leavevmode
  \begin{enumerate}
  \item Every BJ-tree $J$ has a structuring.
  \item The class of stuctures $S(J)$ for SBJ-trees $J$ is monadic second-order
    definable.
  \end{enumerate}
\end{prop}

\begin{proof} 
  \leavevmode
  \begin{enumerate}
  \item We use the construction of Proposition \ref{P:3.5} for $J=(N,\leq )$.\
    The remark following it implies that, if the axis $A=U_{0}$ has a smallest
    element, this element has degree 0. It implies also that, if $%
    \widehat{U_{i}}=x$, then $x$ cannot have degree 0 in the BJ-tree $J_{i-1}$
    induced by $U_{i-1}\uplus \dots\uplus U_{0}$ because each line $B_{i}$ is
    chosen maximal; furthermore, it cannot have degree 2 or more in $J_{i-1}$
    because $J$ is binary. Hence it has degree 1 in $J_{i-1}$. It follows that
    $x$ is the top of no line $U_{j}$ for $j<i$. Hence (ii) holds and the
    construction yields an SBJ-tree $(N,\leq ,\mathcal{U})$.

  \item The formula $\varphi$ of Proposition \ref{P:3.7} can be modified
    so as to express that ${(N,\leq,} N_{0},N_{1})$ is $S(J)$ for some SBJ-tree
    $J.$ \qedhere
  \end{enumerate}
\end{proof}

\begin{defi}[Description schemes for SBJ-trees]\label{D:3.10}
\leavevmode
\begin{enumerate}[label=(\alph*)]
\item A \emph{description scheme for an SBJ-tree}, in short an \emph{%
    SBJ-scheme,} is a triple $\Delta =(Q,w_{Ax},(w_{q})_{q\in Q})$ such that $Q$
  is a set called the set of \emph{states}, $w_{Ax}\in \mathcal{A}(Q)$ (is an
  arrangement over $Q$) and $w_{q}\in \mathcal{A}(Q)$ for each $q$.
  It is \emph{regular} if $Q$ is finite and the arrangements $w_{Ax}$ and $%
  w_{q}$ are regular.

\item We recall that a linear order $(V,\leq)$ is identified with the
  arrangement $(V,\leq,Id)$. If $W\subseteq V$ and $r:V\rightarrow Q$, then $%
  \overline{r}((W,\leq))$ is the arrangement $(W,\leq\upharpoonright
  W,r\upharpoonright W)\in\mathcal{A}(Q)$ that we will denote more simply by $%
  (W,\leq,r)$ leaving implicit the restrictions of $\leq$ and $r$ to $W$.

  An SBJ-scheme $\Delta $ \emph{describes} an SBJ-tree $J=(N,\leq ,\mathcal{U})$
  whose axis is $A$ if there exists a mapping $r:N\rightarrow Q$ that we call a
  \emph{run}, such that:
  \[
    \overline{r}((A,\leq))\simeq w_{Ax}\text{ and }\overline{r}((U^{x},\leq))\simeq
    w_{r(x)}\text{ for every }x\in N.
  \]
  We will also say that $\Delta $ \emph{describes} the BJ-tree $\mathit{fgs}%
  (J):=(N,\leq )$, where $\mathit{fgs}$ makes an SBJ-tree into a BJ-tree by
  \emph{forgetting its structuring}.
  The mapping $r$ need not be surjective, this means that some elements of $Q$
  and the corresponding arrangements may be useless, and thus can be removed
  from $\Delta $.
\end{enumerate}
\end{defi}

\begin{figure}
[ptb]
\begin{center}
\includegraphics[
%natheight=5.250300in,
%natwidth=5.239000in,
%height=2.3454in,
width=2.3402in
]%
{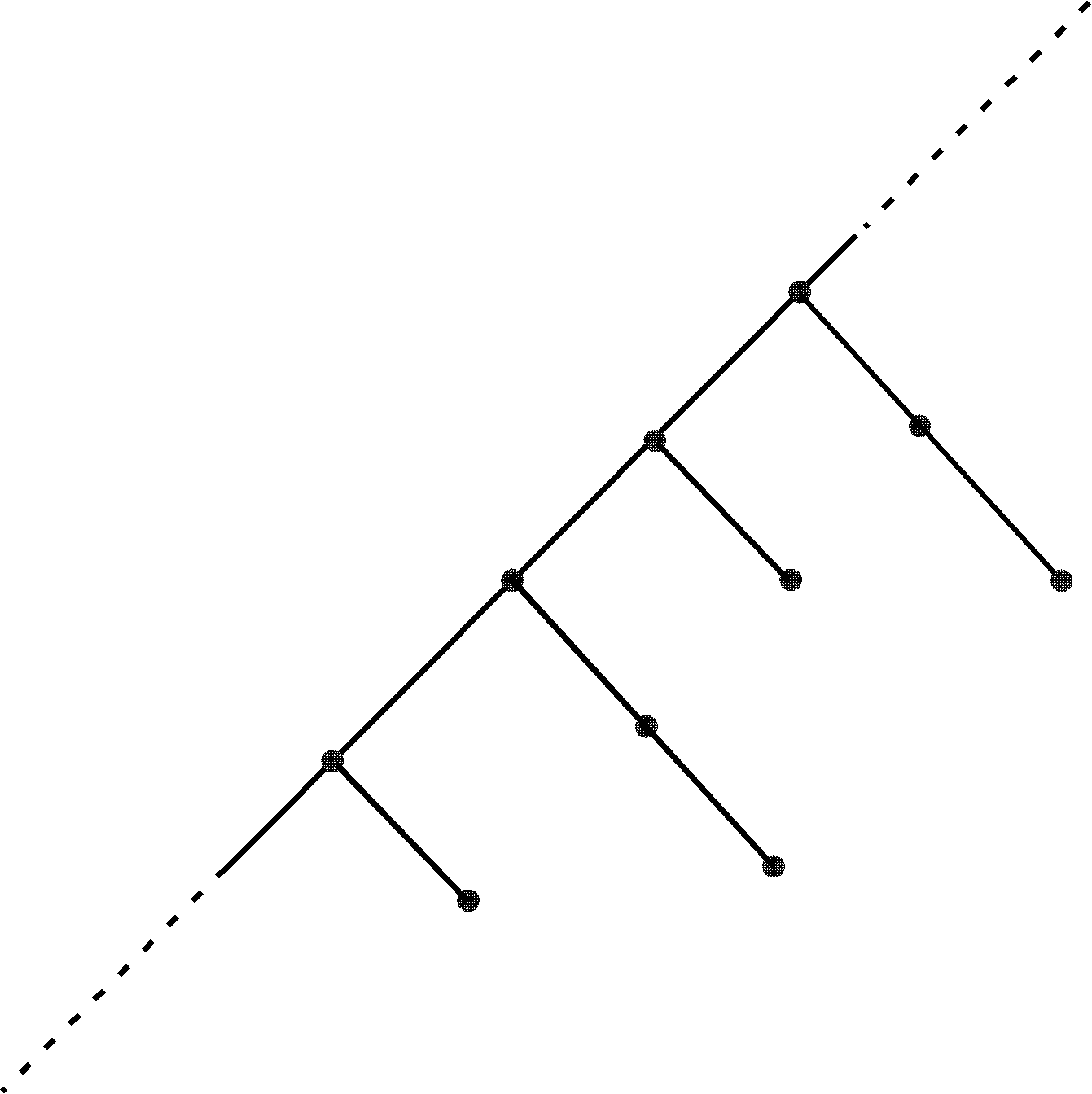}%
\caption{A binary join-tree.}%
\end{center}
\end{figure}

For an example, let $\Delta =(Q,w_{Ax},(w_{q})_{q\in Q})$ be the SBJ-scheme
such that $Q=\{a,b,c\}$, $w_{Ax}:=(\mathbb{Z},\leq ,\ell )$ where $\ell
(i)=a $ if $i$ is even and $\ell (i)=b$ if $i$ is odd, $w_{a}:=\{c%
\},w_{b}:=cc$ (two nodes labelled by $c$) and $w_{c}=\Omega .$ It describes
the BJ-tree of Figure 2. 

\begin{prop}\label{P:3.11}\leavevmode
 \begin{enumerate}
 \item Every SBJ--tree is descibed by some
SBJ-scheme.

\item Every SBJ-scheme $\Delta $ describes a unique SBJ-tree, where unicity is
up to isomorphism.
\end{enumerate}
\end{prop}

\begin{proof} \leavevmode
  \begin{enumerate}
\item Each SBJ-tree $J=(N,\leq ,\mathcal{U})$ has a \emph{%
standard} description scheme \[\Delta (J):=(N, (A,\leq), ((U^{x},\leq))_{x\in N}).\] The run is the identity mapping $r:N\rightarrow N$ showing
that $\Delta (J)$ describes $J$.

\item We will denote by $\mathit{Unf}(\Delta )$ the SBJ-tree described by $%
\Delta $ and call it called the \emph{unfolding of} $\Delta $ (see the
remark following the proof about terminology).
Let $\Delta =(Q,w_{Ax},(w_{q})_{q\in Q})$ be an SBJ-scheme, defined with
arrangements $w_{Ax}=(V_{Ax},$ $\preceq ,lab_{Ax})$ and $%
w_{q}=(V_{q},\preceq ,lab_{q})$ such that, without loss of generality, the
sets $V_{Ax}$ and $V_{q}$ are pairwise disjoint and the same symbol $\preceq 
$ denotes their orders.
We construct $(N,\leq,\mathcal{U})=\mathit{Unf}(\Delta)$ as follows.

\begin{enumerate}
\item $N$ is the set of finite nonempty sequences $(v_{0},v_{1},\dots,v_{k})$
such that $v_{0}\in V_{Ax},v_{i}\in V_{q_{i}}$ for 
$1\leq i\leq k$, where  $q_{1}=lab_{Ax}(v_{0})$, $q_{2}=lab_{q_{1}}(v_{1})$, \dots, $q_{k}=lab_{q_{k-1}}(v_{k-1})$.

\item  $(v_{0},v_{1},\dots,v_{k})\leq(v_{0}^{\prime},
v_{1}^{\prime},\dots,v_{j}^{%
\prime})$ if and only if $k\geq j$, $(v_{0},v_{1},\dots,v_{j-1})=(v_{0}^{%
\prime},v_{1}^{\prime},\dots,$ $v_{j-1}^{\prime})$ and $v_{j}\preceq
v_{j}^{\prime}.$

\item The axis $A$ is the set of one-element sequences $(v)$ for $v\in V_{Ax}$;
for $x=(v_{0},v_{1},\dots,v_{k}),$ we define $U(x)$ as the set of sequences $%
(v_{0},v_{1},\dots,v_{k-1},v)$ such that $v\in V_{q_{k}}$, hence, we have $%
\widehat{U(x)}=(v_{0},v_{1},\dots,v_{k-1}).$
\end{enumerate}
Note that $(v_{0},\dots,v_{k})<(v_{0},\dots,v_{j})$ if  $j<k$ and that $%
(v_{0},\dots,v_{k-1},v_{k})\leq(v_{0},\dots,v_{k-1},v)$ if and only if $%
v_{k}\preceq v$.
We claim that $\Delta$ describes $(N,\leq,\mathcal{U})$. For proving that, we define a run $r:N\rightarrow Q$ as follows:

\begin{itemize}
\item if $x\in A$, then $x=(v)$ for some $v\in V_{Ax}$ and $r(x):=lab_{Ax}(v);$

\item if $x\in N$ has depth $k\geq 1$, then $x=(v_{0},v_{1},\dots,v_{k})$
for some $v_{0},v_{1},\dots,v_{k}$ as in (a) and $r(x):=lab_{q_{k}}(v_{k})$.
\end{itemize}
It follows that $\overline{r}((A,\leq))\simeq w_{Ax}$ and that, for $%
x=(v_{0},v_{1},\dots,v_{k})$ (of depth $k$), we have $\overline{r}%
((U^{x},\leq))\simeq w_{q_{k}}=w_{r(x)}$, which proves the claim.

We now prove unicity. Assume that $\Delta $ describes $J=(N,\leq ,\mathcal{U}%
)$ with axis $A$ and also $J^{\prime }=(N^{\prime },\leq ^{\prime },\mathcal{%
U}^{\prime })$ with axis $A^{\prime }$, by means of runs $r:N\rightarrow Q$
and $r^{\prime }:N^{\prime }\rightarrow Q$. We construct an isomorphism $%
h:J\rightarrow J^{\prime }$ as the common extension of bijections $%
h_{k}:N_{k}\rightarrow N_{k}^{\prime }$, where $N_{k}$ (resp. $N_{k}^{\prime
}$) is the set of nodes of $J$ (resp. of $J^{\prime }$) of depth at most $k,$
and such that they map $\leq $ to $\leq ^{\prime },$ and the lines of $%
\mathcal{U}$ to those of $\mathcal{U}^{\prime }$ of same depth, and finally, 
$r^{\prime }\circ h_{k}=r\upharpoonright N_{k}.$

\begin{itemize}
\item \emph{Case} $k=0$. We have:
  \[
\overline{r}((A,\leq))=(A,\leq,r)\simeq w_{Ax}\simeq\overline{r^{\prime}}%
((A^{\prime},\leq))=(A^{\prime}, \leq^{\prime},r^{\prime})
\]

which gives the order preserving bijection $h_{0}:N_{0}=A\rightarrow
N_{0}^{\prime}=A^{\prime}$ such that $r^{\prime}\circ h_{0}=r\upharpoonright
N_{0}.$

\item \emph{Case} $k>0$. We assume inductively that $h_{k-1}$ has been
constructed.

Let $U\in\mathcal{U}$ be such that $x=\widehat{U}$ has depth $k-1$; hence, $%
U\cap N_{k-1}=\emptyset.$ Then $(U,\leq,r)\simeq w_{r(x)}.$ Let $x^{\prime
}=h_{k-1}(x)$; we have $r^{\prime}(x^{\prime})=r(x)$. Hence there is $%
U^{\prime}\in\mathcal{U}^{\prime}$ such that $x^{\prime}=\widehat{U^{\prime}}
$, $U^{\prime}\cap N_{k-1}^{\prime}=\emptyset$ and$ (U^{\prime},\leq
^{\prime},r^{\prime})\simeq w_{r^{\prime}(x^{\prime})}=w_{r(x)}$. Hence,
there is an order preserving bijection $h_{U}:U\rightarrow U^{\prime}$ such
that $r^{\prime}\circ h_{U}=r\upharpoonright U.$

We define $h_{k}$ as the common extension of the injective mappings $h_{k-1}$
and $h_{U}$ such that $U\in \mathcal{U}$ and the depth of $\widehat{U}$ is $%
k-1$. These mappings have pairwise disjoint domains whose union is $N_{k}$.
\end{itemize}
The extension to $N$ of all these mappings $h_{k}$ is the desired
isomorphism $h$. 
\qedhere
\end{enumerate}
\end{proof}

\noindent
\textbf{Remarks.}
\begin{enumerate}
\item We call \emph{unfolding} the transformation of $\Delta $
into $\mathit{Unf}(\Delta )$ because it generalizes the unfolding of a
directed graph $G$ into a finite or countable rooted tree. The unfolding is
done from a particular vertex $s$ of $G$, and the nodes of the tree are the
sequences of the form $(x_{0},\dots,x_{k})$ such that $s=x_{0}$ and there is a
directed edge in $G$ from $x_{i}$ to $x_{i+1}$, for each $i<k$. If $\Delta $
is such that the arrangements $w_{Ax}$ and $w_{q}$ are reduced to a single
element, the corresponding directed graph has all its vertices of outdegree
one and the tree resulting from the unfolding consists of one infinite path:
the SBJ-tree $\mathit{Unf}(\Delta )$ is the order type $\omega ^{-}$ of
negative integers and the sets in $\mathcal{U}$ are singletons.

\bigskip

\item An SBJ-scheme $\Delta $ describing an SBJ-tree $J$ can be seen as a 
\emph{quotient of} $\Delta (J).$ We define quotients in terms of surjective
mappings.

Let $\Delta =(Q,w_{Ax},(w_{q})_{q\in Q})$ and $s:Q\rightarrow Q^{\prime }$
be surjective such that, if $s(q)=s(q^{\prime })$, then $\overline{s}%
(w_{q})\simeq \overline{s}(w_{q^{\prime }})$. We define $\Delta ^{\prime
}:=(Q^{\prime },\overline{s}(w_{Ax}),(w_{p}^{\prime })_{p\in Q^{\prime }})$
where $w_{s(q)}^{\prime }\simeq \overline{s}(w_{q})$ for each $q$ in $Q$. If 
$\Delta $ describes $J$ via a run $r:N\rightarrow Q$, then $\Delta ^{\prime
}$ describes $J$ via $s\circ r:N\rightarrow Q^{\prime }$. We say that $%
\Delta ^{\prime }$ is \emph{the quotient of} $\Delta $ by the equivalence $%
\thickapprox $ on $Q$ such that $q\approx q^{\prime }$ if and only if $%
s(q)=s(q^{\prime })$, and we denote it by $\Delta /\thickapprox $.
If $\Delta $ is regular, then $\Delta /\thickapprox $ is regular and its
set of sates is smaller than that of $\Delta $ unless $s$ is a bijection.

Let us now start from an SBJ-tree $J=(N,\leq ,\mathcal{U})$ with axis $A$.
For $x\in N$, let $J_{x}$ be the SBJ-tree $(N_{x},\leq _{x},\mathcal{U}_{x})$
such that $N_{x}:=\mathord\downarrow (U^{x}),\leq _{x}$ is the restriction of $\leq $
to $N_{x}$, and $\mathcal{U}_{x}:=\{U^{x}\}\cup \{U^{y}\mid y\in
N_{x}\}-\{\emptyset \}$. Its axis is $U^{x}$. For the example of Figure 1,
we have $\mathcal{U}_{x_{2}}=\{U_{2},U_{3}\}$ and $U_{2}$ is the axis.

From \emph{J}, we define as follows a canonical SBJ-scheme $\Delta
(J)/\thickapprox $ based on the equivalence $\thickapprox $ on $N$ such that 
$x\approx x^{\prime }$ if and only if $J_{x}\simeq J_{x^{\prime }}$. Let $s$
be the surjective mapping~: $N\rightarrow N^{\prime }:=N/\approx $. If $%
J_{x}\simeq J_{y}$ by an isomorphism $h:N_{x}\rightarrow N_{y}$, then $%
(U^{x},\leq )\simeq (U^{y},\leq )$ by $h\upharpoonright
U^{x}:U^{x}\rightarrow U^{y}$ and furthermore, if $w\in N_{x}$, then $%
J_{w}\simeq J_{h(w)}$ by $h\upharpoonright N_{w}:N_{w}\rightarrow N_{h(w)}$%
. It follows that $\overline{s}((U^{x},\leq ))\simeq \overline{s}%
((U^{y},\leq )),$ hence, the quotient SBJ-scheme $\Delta (J)/\thickapprox
:=(N^{\prime },\overline{s}((A,\leq )),(w_{p}^{\prime })_{p\in N^{\prime }})$
such that $w_{p}^{\prime }\simeq $ $\overline{s}((U^{x},\leq ))$ if $s(x)=p$
is well-defined and describes $J$.

Let $\Delta =(Q,w_{Ax},(w_{q})_{q\in Q})$ describe $J$ via a surjective run $%
r:N\rightarrow Q$ and consider the equivalence relation on $Q$ such that $%
q\equiv q^{\prime }$ if and only if there exist $x,y$ such that $r(x)=q$,
$r(y)=q^{\prime }$ and $J_{x}\simeq J_{y}$. It is well-defined because if $%
r(x)=r(y)=q$, then $\overline{r}((U^{x},\leq ))\simeq \overline{r}%
((U^{y},\leq ))\simeq w_{q},$ and furthermore $J_{x}\simeq J_{y}$: one
constructs an isomorphism $h$ that extends the one between $\overline{r}%
((U^{x},\leq ))$ and $\overline{r}((U^{y},\leq ))$ (this construction uses
an induction on the depth of $u$ in $J_{x}$ for defining $h(u)$). Hence, $%
\Delta /\equiv $ is well-defined and describe $J$ via the run $r^{\prime }$
such that $r^{\prime }(x)$ is the equivalence class of $r(x)$ with respect
to $\equiv .$ It follows that $\Delta /\equiv $ is isomorphic to $\Delta
(J)/\thickapprox .$ If $J$ is regular, then $\Delta (J)/\thickapprox $ is
the unique regular SBJ-scheme describing $J$ that has a minimum number of
states. As usual, unicity is up to isomorphism. This construction is
similar to that of the minimal deterministic automaton of a regular
language, defined from its quotients (see \emph{e.g.,} \cite{Saka}, chapter
I.3.3). 
\end{enumerate}

\begin{prop}\label{P:3.12} A BJ-tree is monadic second-order definable if
it is described by a regular SBJ-scheme.
\end{prop}

\begin{proof}
 That $J=(N,\leq)$ is a BJ-tree is first-order expressible.
Assume that $J=\mathit{fgs}(J^{\prime})$  where $J^{\prime}=(N,\leq ,%
\mathcal{U})\simeq\mathit{Unf}(\Delta)$ for some regular SBJ-scheme $%
\Delta=(Q,w_{Ax},(w_{q})_{q\in Q})$ such that $Q=\{1,\dots,m\}$. Let $r$ be
the corresponding mapping: $N\rightarrow Q$ (cf. Definition \ref{D:3.10}(b)). For
each $q\in Q$, let $\psi_{q}$ be an MS sentence that characterizes $w_{q}$,
up to isomorphism, by the main result of \cite{Tho86}. Similarly, $\psi_{Ax}$
characterizes $w_{Ax}$. We claim that a relational structure $(X,\leq)$ is
isomorphic to $J$ if and only if there exist subsets $%
N_{0},N_{1},M_{1},\dots,M_{m}$ of $X$ such that:

\begin{enumerate}[label=(\roman*)]

\item $(X,\leq)$ is a BJ-tree and $(X,\leq,N_{0},N_{1})=S(J^{\prime\prime}) $
for some SBJ-tree $J^{\prime\prime}=(X,\leq,\mathcal{U}^{\prime})$,

\item $(M_{1},\dots,M_{m})$ is a partition of $X$; we let $r^{\prime}$
maps each $x\in X$ to the unique $q\in Q$ such that $x\in M_{q},$

\item for every $q\in Q$ and $x\in M_{q}$, the arrangement $%
\overline{r^{\prime }}((U^{x},\leq ))$ over $Q$ (where $U^{x}\in \mathcal{U}%
^{\prime })$  is isomorphic to $w_{q},$

\item the arrangement $\overline{r^{\prime}}((A^{\prime},\leq))$
over $Q$ where $A^{\prime}$ is the axis of $J^{\prime\prime}$ is isomorphic
to $w_{Ax}$.

\end{enumerate}
Conditions (ii)-(iv) express that $\Delta$ describes $J^{\prime\prime}$,
hence that $J^{\prime\prime}$ is isomorphic to $J^{\prime},$ and so that $%
(X,\leq)\simeq\mathit{fgs}(J^{\prime})=J$.

By Proposition \ref{P:3.9}, Condition (i) is expressed by an MS formula $\varphi
(N_{0},N_{1})$, and the property $U\in \mathcal{U}\wedge x=\widehat{U}$ is
expressed in terms of $N_{0},N_{1}$ by an MS formula $\theta
(x,U,N_{0},N_{1}).$ Conditions (iii) and (iv) are expressed by means of the
MS sentences $\psi _{Ax}$ and $\psi _{q}$ suitably adapted to take $%
N_{0},N_{1},M_{1},\dots,M_{m}$ as arguments. Hence, $J$ is (up to isomorphism)
the unique model of an MS sentence of the form:
\[
\exists N_{0},N_{1}.\left[\varphi(N_{0},N_{1})\wedge\exists
M_{1},\dots,M_{m}.\varphi^{\prime}(N_{0},N_{1},M_{1},\dots,M_{m})\right]
\]
where $\varphi^{\prime}$ expresses conditions (ii)-(iv). \end{proof}
\noindent
Theorem \ref{T:3.21}  will establish a converse.

\subsection{The algebra of binary join-trees}

We define three operations on structured binary join-trees (SBJ-trees in
short). The finite and infinite terms over these operations will define all
SBJ-trees.

\begin{defi}[Operations on SBJ-trees]\label{D:3.13}
 \leavevmode
\begin{itemize}
\item \emph{Concatenation along axes.}

\noindent Let $J=(N,\leq,\mathcal{U})$ and $J^{\prime}=(N^{\prime},\leq^{\prime },%
\mathcal{U}^{\prime})$ be disjoint SBJ-trees, with respective axes $A$ and $%
A^{\prime}$. We define:
\begin{align*}
J\bullet J^{\prime}&:=(N\uplus N^{\prime},\leq^{\prime\prime},\mathcal{U}%
^{\prime\prime}) \text{ where } \\
x\leq^{\prime\prime}y&:\Longleftrightarrow x\leq y\vee
x\leq^{\prime}y\vee(x\in N\wedge y\in A^{\prime}), \\
\mathcal{U}^{\prime\prime}&:=\{A\uplus A^{\prime}\}\uplus(\mathcal{%
U-}\{A\}\mathcal{)}\uplus(\mathcal{U}^{\prime}\mathcal{-}\{A^{\prime }\}%
\mathcal{)}.
\end{align*}
It follows that $J\bullet J^{\prime}$ is an SBJ-tree with axis $A\uplus A^{\prime}$; the
depth of a node in $J\bullet J^{\prime}$  is the same as in $J$ or $%
J^{\prime}$.
This operation generalizes the concatenation of linear orders: if $(N,\leq )$
and $(N^{\prime },\leq ^{\prime })$ are disjoint linear orders, then the
SBJ-tree $(N,\leq ,\{N\})\bullet\left(N^{\prime},\leq^{\prime},\{N^{\prime
}\}\right)$ corresponds to the concatenation of $(N,\leq )$ and $\left(N^{\prime},\leq
^{\prime}\right)$ usually denoted by $(N,\leq )+\left(N^{\prime},\leq^{\prime}\right)$,
cf. \cite{LS}.
If $K=(M,\leq,\mathcal{V})$ is an SBJ-tree with axis $B$, and $B=A\uplus
A^{\prime}$ such that $A<A^{\prime},$ then $K=J\bullet J^{\prime}$ where:
\begin{itemize}[label=,leftmargin=8mm,topsep=2mm]
\item $N:=\downarrow(A),N^{\prime}:=M-N,$

\item $\mathcal{U}$ is the set of lines of $\mathcal{V}$ included in $%
N-A $ together with $A,$

\item $\mathcal{U}^{\prime}$ is the set of lines of $\mathcal{V}$
included in $N^{\prime}-A^{\prime}$ together with $A^{\prime}$ and

\item the orders of $J$ and $J^{\prime}$ are the restrictions of $\leq$\
to $N$ and $N^{\prime}$.
\end{itemize}

\item \emph{The empty SBJ-tree.}

\noindent The nullary symbol $\Omega$ denotes the empty SBJ-tree.

\item \emph{Extension.}

\noindent Let $J=(N,\leq,\mathcal{U})$ be an SBJ-tree, and $u\notin N$. Then:
\begin{align*}
ext_{u}(J)&:=(N\uplus\{u\},\leq^{\prime},\{\{u\}\}\uplus\mathcal{U}) \text{ where } \\
x\leq^{\prime}y &:\Longleftrightarrow x\leq y\vee y=u,
\end{align*}
the axis is $\{u\}$. Clearly, $ext_{u}(J)$ is an SBJ-tree. The depth of $v\in N$
is its depth in $J$ plus 1. The axis of $J$ is turned into an \textqt{ordinary
  line} of the structuring of $ext_{u}(J)$ with top equal to $u$. When handling
SBJ-trees up to isomorphism, we use the notation $ext(J)$ instead of
$ext_{u}(J).$

\item \emph{Forgetting structuring.}

\noindent If $J$ is an SBJ-tree as above, $\mathit{fgs}(J):=(N,\leq )$ is the
underlying BJ-tree (binary join-tree), where $\mathit{fgs}$ \emph{forgets
the structuring}.
\end{itemize}
\end{defi}

\noindent Anticipating the sequel, we observe that a linear order $a_{1}<\dots<a_{n}$,
identified with the SBJ-tree $(\{a_{1},\dots,a_{n}\},\leq,\{\{a_{1},\dots,a_{n}%
\}\})$ is defined by the term $t:=ext_{a_{1}}(\Omega)\bullet
ext_{a_{2}}(\Omega)\bullet\dots\bullet ext_{a_{n}}(\Omega).$ The binary (it
is even \textqt{unary}) join-tree $(\{a_{1},\dots,a_{n}\},\leq)$ is defined by the
term $\mathit{fgs}(t)$ and also, in a different way, by the term $\mathit{fgs%
}(ext_{a_{n}}(ext_{a_{n-1}}(\dots(ext_{a_{1}}(\Omega)))..)))$.

\bigskip

\begin{defi}[The algebra $\mathbb{SBJT}$]\label{D:3.14}

We let $F$ be the signature $\{\bullet,ext,\Omega\}$. We obtain an algebra $%
\mathbb{SBJT}$ whose domain is the set of isomorphism classes of SBJ-trees.
Concatenation is associative with neutral element $\Omega$.
\end{defi}

\begin{defi}[The value of a term]\label{D:3.15}
  \leavevmode
\begin{enumerate}[label=(\alph*),beginpenalty=99,midpenalty=0,endpenalty=0]
\item In order to define the value of a term $t$ in $T^{\infty}(F)$, we
compare its positions by the following equivalence relation:
\begin{quote}
  $u\approx v$ if and only if every position $w$ such that
  $u<_{t}w\leq_{t}u\sqcup_{t}v$ or $v<_{t}w\leq_{t}u\sqcup_{t}v$ is an
  occurrence of $\bullet$.
\end{quote}
We will also use the lexicographic order $%
\leq_{lex}$ (positions are Dewey words). If $w$ is an occurrence of a binary
symbol, then $s_{1}(w)$ is its first (left) son and $s_{2}(w)$ its second
(right) one.
\item We define the \emph{value} $val(t):=(N,\leq,\mathcal{U})$ of $t$ in $%
T^{\infty}(F)$ as follows:
\begin{itemize}
\item $N:=\mathrm{Occ}(t,ext)$, the set of occurences in $t$ of $ext$,
\item $u\leq v:\Longleftrightarrow u\leq_{t}w\leq_{lex}v$ for some $w\in
N$ such that $w\approx v,$
\item $\mathcal{U}$  is the set of equivalence classes of $\approx.$
\end{itemize}
Equivalently, we have:
\begin{quote}
$u\leq v:\Longleftrightarrow u\leq_{t}v$ or $u\leq_{t}s_{1}(u\sqcup_{t}v), $ 
$v\leq_{t}s_{2}(u\sqcup_{t}v)$ and $v\approx u\sqcup_{t}v$ (the position $%
u\sqcup_{t}v$ is an occurrence of $\bullet$),
\end{quote}
and so (we recall that $\bot$ denotes incomparability):
\begin{quote}
$u\bot v:\Longleftrightarrow u\leq_{t}s_{1}(u\sqcup_{t}v),$ $%
v\leq_{t}s_{2}(u\sqcup_{t}v)$ and there is an occurrence of $ext$ between $%
v $ and $u\sqcup_{t}v$  or vice-versa by exchanging $u$ and $v$.
\end{quote}

\item We now consider terms $t$ written with the operations $ext_{a}$ (such
that $a$ is the node created by applying this operation). For each $a$, the
operation $ext_{a}$ must have at most one occurrence in $t$. Assuming this
condition satisfied, then $val(t):=(N,\leq,\mathcal{U})$, where
\begin{itemize}
\item $N$ is the set of nodes $a$ such that $ext_{a}$ has an occurence in $t$ that
we will denote by $u_{a}$,

\item $a\approx b:\Longleftrightarrow u_{a}\approx u_{b}$, with $\approx$
 as in (a),

\item $a\leq b:\Longleftrightarrow u_{a}\leq u_{b}$, with $\leq$ as in
(b),

\item $\mathcal{U}$ is the set of equivalence classes of $\approx.$
\end{itemize}
Clearly, the mapping $val$ in (b) is a value mapping 
$T^{\infty }(F)\rightarrow\mathbb{SBJT}$.
\end{enumerate}
We say that $t$ \emph{denotes} an SBJ-tree $J$ if $J$ is isomorphic to $%
val(t)$, and, in this case, we also say that $\mathit{fgs}(t)$ denotes the
BJ-tree $\mathit{fgs}(J)$. 
\end{defi}

\begin{figure}
[ptb]
\begin{center}
\includegraphics[
%natheight=4.570500in,
%natwidth=4.766000in,
%height=2.527in,
width=2.6342in
]%
{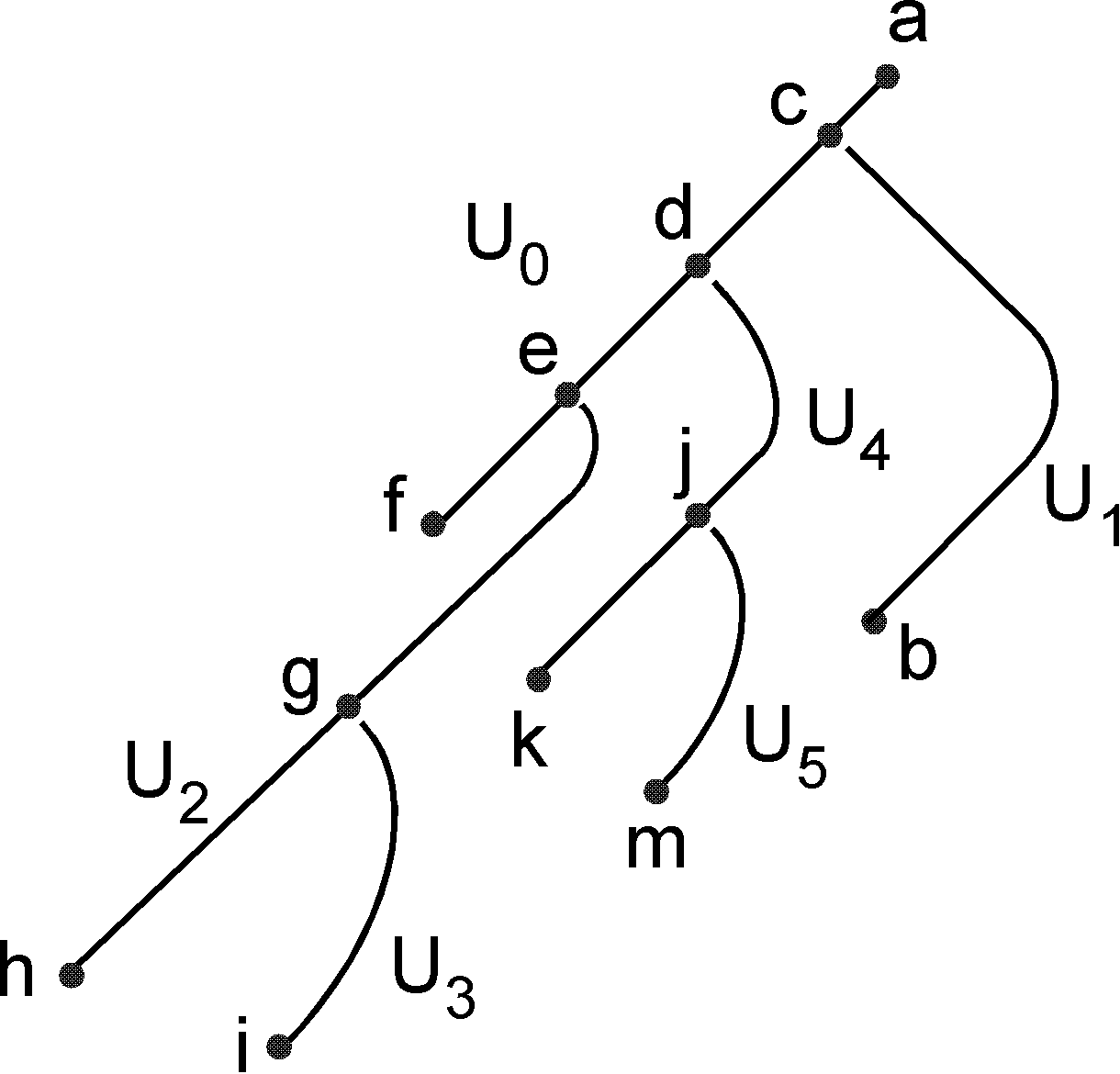}%
\caption{A finite SBJ-tree $J$.}%
\end{center}
\end{figure}

Note that we do not define the value of term as the least upper bound of the
values of its finite subterms. We could use a notion of least upper bound
based on category theory as in \cite{Cou78}, at the cost of heavy
definitions. Our simpler definition shows furthermore that the mapping
associating the join-tree $(N,\leq )$ with $\left\lfloor t\right\rfloor $\
for $t\in T^{\infty }(F)$ is an MS-transduction (cf.~Section 2) defined by $%
\mathcal{D}=\langle \chi ,\delta ,\theta _{\leq }\rangle $ where $\chi $
expresses that the considered input structure $S$ is isomorphic to $%
\left\lfloor t\right\rfloor $ for some $t\in T^{\infty }(F),$ $\delta (x)$
is $lab_{ext}(x)$ (expressing that $x\in N)$ and $\theta _{\leq }(x,y)$
expresses that $x\leq y$, cf.~Definition \ref{D:3.15}(b).

\bigskip

\begin{figure}
[ptb]
\begin{center}
\includegraphics[
%natheight=7.8516in,
%natwidth=5.3566in,
%height=3.883in,
width=2.6584in
]%
{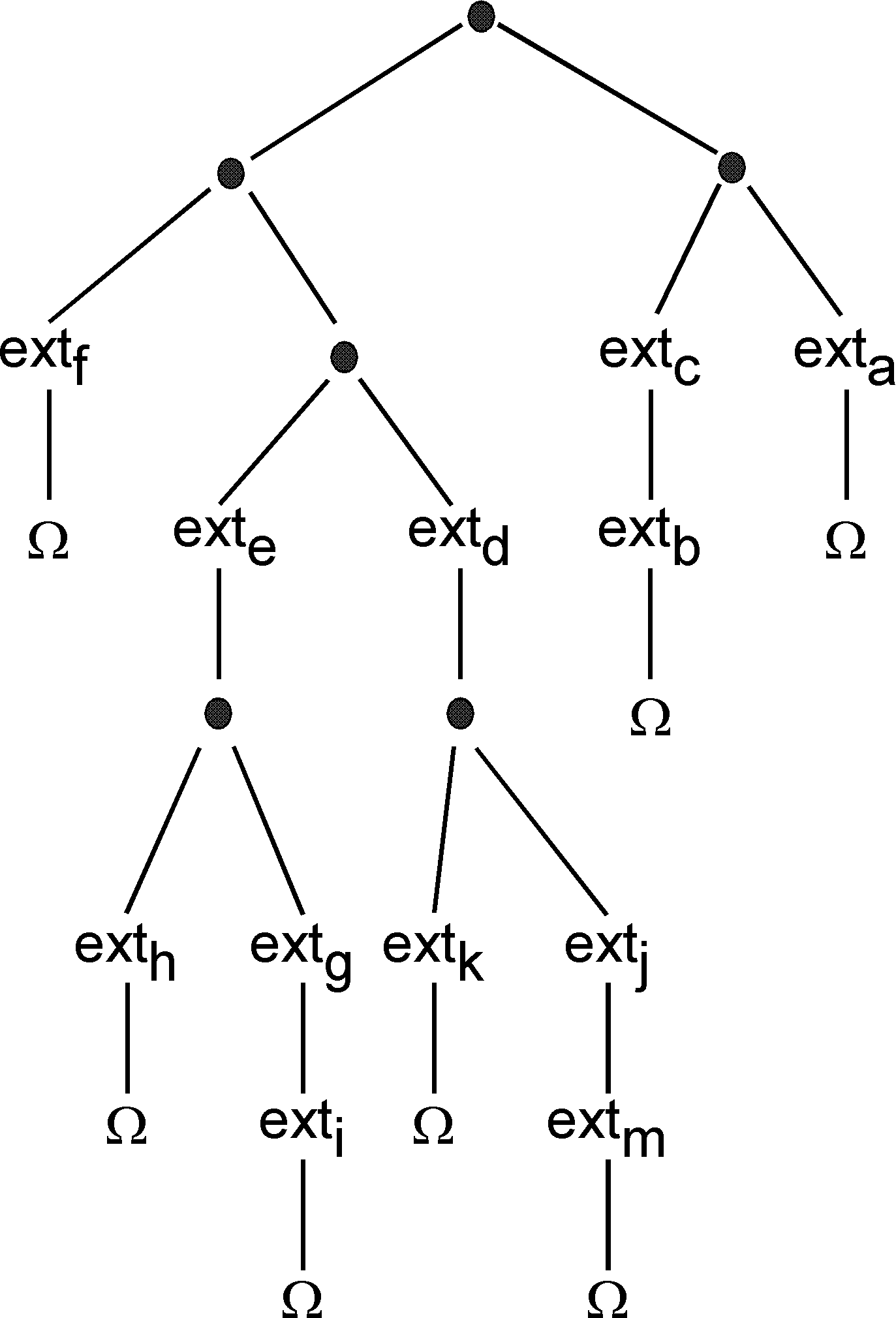}%
\caption{A term $t$ denoting $J$.}%
\end{center}
\end{figure}

\begin{exas}\label{E:3.16}
  \leavevmode
  \begin{enumerate}
\item The term $t_{0}$ that is the unique solution in $T^{\infty}(F)$ of the
equation $t_{0}=t_{0}\bullet t_{0}$ denotes the empty SBJ-tree $\Omega$.

\item Figure 3 shows a finite SBJ-tree $J$ whose structuring consists of $%
U_{0},\dots,U_{5}$, and $U_{0}$ is the axis. The linear order on $U_{0}$ can
be described by the word $fedca$ (with $f<e<d<\dots).$ Similarly, $%
U_{1}=b,U_{2}=hg,U_{3}=i$ , $U_{4}=kj$  and $U_{5}=m$.

Let us examine the term $t$ of Figure 4 that denotes $J$. A function symbol 
$ext_{u}$ specifies the node $u$ of $J$, and we also denote by $u$ its
occurrence, a position of $t$ (hence \emph{b} denotes position 21). The
occurrences of $\bullet $ and $\Omega $ are denoted by Dewey words. For
example, the occurrences of $\bullet $ above the symbols $ext$ are the
words $\varepsilon ,1,2,12$. The set $\{\varepsilon ,1,2,12,f,e,d,c,a\}$ is
an equivalence class of $\approx $. Another one is $\{1221,k,j\}$. Each
line $U_{i}$ is the set of positions of the $ext$ symbols in some
equivalence class of $\approx $. Let us now examine how each line is
ordered.

The case where $u<v$ holds because $u<_{t}v$ is illustrated, to take a few
cases, by $i<g,g<e,m<j$ and $j<d$.
The case where $u<v$ holds because $u\bot_{t}v$, $u\leq_{t}s_{1}(u\sqcup
_{t}v),$ $v\leq_{t}s_{2}(u\sqcup_{t}v)$ and $v\approx u\sqcup_{t}v$ is
illustrated by $f<e,e<d,d<c$ and $i<d$. We have $i<d$ because $i\sqcup
_{t}d=12,$ $i<_{t}s_{1}(12),$ $d\leq_{t}s_{2}(12)$ and $d\approx12$. We do
not have $i<j$ because $j$ is not $\approx$-equivalent to 12, whereas $%
i\sqcup_{t}j=12,$ $i<_{t}s_{1}(12)$ and $ j\leq_{t}s_{2}(12)$. This case
illustrates the characterization of $\bot$ Definition \ref{D:3.15}(c).

\begin{figure}
[ptb]
\begin{center}
\includegraphics[
%natheight=5.8608in,
%natwidth=7.6432in,
height=2.3in,
]%
{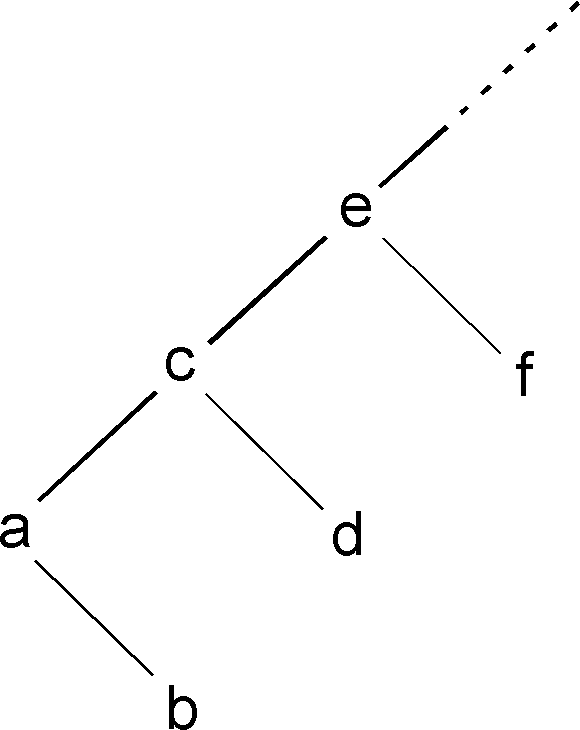}%
\caption{The SBJ-tree $val(t_{1})$.}%
\end{center}
\end{figure}

\item Let $t_{1}$ be the solution in $T^{\infty }(F)$ of the equation $%
t_{1}=ext(ext(\Omega ))\bullet t_{1}$. We write it by naming $a,a^{\prime
},b,b^{\prime },c,c^{\prime }\dots$ the nodes created by the operations $ext$,
hence, \emph{t}$_{1}=ext_{a}(ext_{a^{\prime }}(\Omega ))\bullet
(ext_{b}(ext_{b^{\prime }}(\Omega ))\bullet (ext_{c}(ext_{c^{\prime
}}(\Omega ))\bullet \dots))).$
This term and its value are shown in Figure 5. The bold edges link nodes in
the axis. The nodes $a^{\prime }$ and $c^{\prime }$ are incomparable
because the corresponding occurrences of $ext,$ that are 111 and 2211, have
least common ancestor $\varepsilon $ and 221 is an occurrence of $ext$
between 2211 and $\varepsilon $.

\item The following BJ-tree is defined by Fra\"{\i}ss\'{e} in \cite{Fra} (Section 10.5.3).
We let $W:=(Seq_{+}(\mathbb{Q}),\preceq)$ where
$Seq_{+}(\mathbb{Q})$ is the set of nonempty sequences of rational numbers, partially ordered as follows:$(x_{n},\dots,x_{0})\preceq(y_{m},\dots,y_{0})$ if and only if 
$n\geq m$, $(x_{m-1},\dots,x_{0})=(y_{m-1},\dots,y_{0})$ and $x_{m}\leq y_{m}$. In particular, $\prec$ is the transitive closure of $\prec_{0}\cup\prec_{1}$ where
$(x_{p+1},x_{p},\dots,x_{0})\prec_{0}(x_{p},\dots,x_{0})$ and
$(y,x_{p-1},\dots,x_{0})\prec_{1}(z,x_{p-1},\dots,x_{0})$ if $y<z.$
It is easy to check that $W$ is a BJ-tree.  In particular, two nodes $(x_{n},\dots,x_{0})$ and $(y_{m},\dots,y_{0})$ are
incomparable if and only if 
$(y_{m},\dots,y_{0})=(y_{m},\dots,y_{p+1},x_{p},\dots,x_{0})$ and 
$y_{p+1}\neq x_{p+1}$ for some $p<n,m.$ In this case,
their join is $(\min\{y_{p+1},x_{p+1}\},x_{p},\dots,x_{0})$. 
The two directions
relative to a node $x=(x_{p},\dots,x_{0})$ are:
\begin{quote}
$\partial_{0}(x):=\{(y_{m},\dots,y_{p+1},x_{p},\dots,x_{0})\mid n>p,y_{m},\dots,y_{p+1}\in\mathbb{Q}\}$ and,

\noindent $\partial_{1}(x):=\{(y_{m},\dots,y_{p+1},x_{p}^{\prime},\dots,x_{0})\mid n\geq
p,x_{p}^{\prime},y_{m},\dots,y_{p+1}\in\mathbb{Q}$,  $x_{p}^{\prime}<x_{p}\}$.
\end{quote}
A structuring $\mathcal{U}$ of $W$ consists of the sets $\{(x_{n}%
,\dots,x_{0})\mid x_{n}\in\mathbb{Q}\}$ for each (possibly empty) sequence
$(x_{n-1},\dots,x_{0})$. The set of one element sequences $(r)$ for
$r\in\mathbb{Q}$ is the axis, and $U_{-}(x)\subseteq\partial_{1}(x)$ for all
$x\in Seq_{+}(\mathbb{Q}).$

The proof in \cite{Fra} that every finite or countable generalized tree in
the sense of \cite{LS} (\emph{i.e.}, partial order satisfying Condition 2) of
Definition 3.1(a)) is isomorphic to $(X,\preceq\upharpoonright X)$ for some
subset $X$ of $Seq_{+}(\mathbb{Q})$ uses implicitely the structuring
$\mathcal{U}$. Our description of the two directions of a node shows that
$W=W\bullet(ext(W)\bullet W),$ hence, that $W$ is denoted by the regular term
$t\in T^{\infty}(F)$ such that $t=t\bullet(ext(t)\bullet t)$.
\end{enumerate}\end{exas}

\begin{defi}[The description scheme associated with a term]\label{D:3.17}
\leavevmode\begin{enumerate}
\item Let $t\in T^{\infty}(F)$ and $u\in Pos(t)$. We denote by $\mathrm{Max}%
(t,ext,u)$ the set of maximal occurrences of $ext$ in $t$ that are below $u$
or equal to it. Positions are denoted by Dewey words, hence, these sets are
linearly ordered by $\leq_{lex}$. We denote by $W(t,u)$ the simple
arrangement  $(\mathrm{Max}(t,ext,u),\leq_{lex}).$
Let $J=(N,\leq,\mathcal{U})$ be the value of $t$ (cf. Definition  \ref{D:3.15})
and $x$ be an occurrence of $ext$ with son $u$. We have $(U^{x},\leq )$ $=(
\mathrm{Max}(t,ext,u),\leq_{lex}).$
For the term $t$ in Example \ref{E:3.16}(2), see Figure 4, we have
$W(t,\varepsilon)=fedca, W(t,1)=fed, W(t,1211)=hg.$ For $t_{1}$ in Example
\ref{E:3.16}(3), we have $W(t_{1},\varepsilon )=abc\dots, W(t_{1},1)=a,
  W(t_{1},11)=a^{\prime} \text{ and } W(t_{1},111)=\Omega.$

\item We define $\Delta(t)$ as the SBJ-scheme
$(\mathrm{Occ}(t,ext),W(t,\varepsilon),(W(t,s(x)))_{x\in\mathrm{Occ}(t,ext)}) $    where $s(x)$ is the unique son of an occurrence $x$ of $ext$.
We obtain \[\Delta (t_{1})=(\mathrm{2}^{\ast }1\uplus 2^{\ast
}11,abc\dots,(w_{x})_{x\in \mathrm{Occ}(t_{1},ext)})\] with $w_{1}=a^{\prime }$%
, $w_{21}=b^{\prime }$, \dots, $w_{11}=\Omega $, $w_{211}=\Omega $, \dots for
the term $t_{1}$ of Example \ref{E:3.16}(3).
\end{enumerate}
\end{defi}

\begin{lem}\label{L:3.18} If $t\in T^{\infty}(F),$ then $val(t)$ is described
by $\Delta(t)$.
\end{lem}

\begin{proof}
Let $val(t)=(N,\leq,\mathcal{U}).$ The conditions
of Definition \ref{D:3.10}(b) hold with the identity on $\mathrm{Occ}(t,ext)$ as
mapping $r$ because $(U^{x},\leq)=(\mathrm{Max}(t,ext,$ $s(x)),\leq_{lex})$
as observed in Definition \ref{D:3.17}(a). 
\end{proof}

\begin{prop}\label{P:3.19} Every SBJ-tree is the value of a term.
\end{prop}

\begin{proof}
Let $J=(N,\leq ,\mathcal{U})$ be an SBJ-tree. For each $k$,
we let $J_{k}$ be the SBJ-tree $(N_{k},\leq ,\mathcal{U}_{k})$ where $N_{k}$
is the set of nodes of depth at most $k$ and $\mathcal{U}_{k}$ is the set of
lines $U\in \mathcal{U}$ of depth at most $k$. By induction on $k$, we
define for each $k$ a term $t_{k}$ that defines $J_{k}$ such that $t_{k}$ $%
\underline{\ll }$ $t_{k^{\prime }}$ if $k<k^{\prime }$, and then, the
least upper bound of the terms $t_{k}$ is the desired term $t$ whose value
is $J$.
We define terms using the symbols $ext_{a}$ where $a$ names the node created
by the corresponding occurrence of the extension operation.

If $k=0$, then $J_{0}=(A,\leq ,\{A\})$. There exists a term $t\in T^{\infty
}(\{\bullet \},Ext_{A})$ whose value is $J_{0}$, where $Ext_{A}$ is the set
of terms $ext_{a}(\Omega )$ for $a\in A$ (we use $Ext_{A}$ as a set of
nullary symbols). We use here Theorem 2.3 of \cite{Cou78}, that follows
immediately from the representation of a linear order by the lexicographic
order on a prefix-free language\footnote{Also used in the related paper \cite{Cou14}.} recalled in Section 1.

Let $k\geq1$, where $t_{k-1}$ defines $J_{k-1}$. Then $J_{k}$ is obtained
from $J_{k-1}$ by adding below some nodes $x$ at depth $k-1$ the line $U^{x}$
(if $U^{x}=\emptyset$, there is nothing to add below $x$). Let $t_{x}\in
T^{\infty}(\{\bullet\},Ext_{U^{x}})$ whose value is $(U^{x},\leq)$. We
obtain $t_{k}$ by replacing in $t_{k-1}$ each subterm $ext_{x}(\Omega)$ by $%
ext_{x}(t_{x})$, for $x$ at depth $k-1$ such that $U^{x}\neq\emptyset$.
It is clear that $t_{k-1}$ $\underline{\ll}$ $t_{k}$ and that the least
upperbound of the terms $t_{k}$ defines $J$. 
\end{proof}

\noindent For an example, we apply this construction to the SBJ-tree $J$ of Figure
3. For defining $J_{0}$, we can take:
\[
t_{0}=((ext_{f}(\Omega)\bullet ext_{e}(\Omega))\bullet ext_{d}(\Omega
))\bullet(ext_{c}(\Omega)\bullet ext_{a}(\Omega)).
\]
To obtain $t_{1}$, we replace $ext_{e}(\Omega )$ by $ext_{e}(ext_{h}(\Omega
)\bullet ext_{g}(\Omega ))$, $ext_{d}(\Omega )$ by $ext_{d}(ext_{k}(\Omega
)\bullet ext_{j}(\Omega ))$ and $ext_{c}(\Omega )$ by $ext_{c}(ext_{b}(%
\Omega )),$ which gives:
\[
t_{1}=((ext_{f}(\Omega )\bullet ext_{e}(ext_{h}(\Omega )\bullet
ext_{g}(\Omega )))\bullet ext_{d}(ext_{k}(\Omega )\bullet ext_{j}(\Omega
))\bullet (ext_{c}(ext_{b}(\Omega ))\bullet ext_{a}(\Omega )).
\]
Then, we obtain $t_{2}$ that defines $J$ by replacing $ext_{g}(\Omega )$ by 
$ext_{g}(ext_{i}(\Omega ))$ and $ext_{j}(\Omega )$ by $ext_{j}(ext_{m}(%
\Omega )).$

\subsection{Regular binary join-trees}

As said in the introduction, the regular objects are those defined by
regular terms. We apply this meta-definition to binary join-trees and their
structurings.

\begin{defi}[Regular BJ- and SBJ-trees]\label{D:3.20}

A BJ-tree (resp. an SBJ-tree) $T$ is \emph{regular} if it is denoted by $%
\mathit{fgs}(t)$ (resp. by $t$) where $t$ is a regular term in $T^{\infty}(F)$.

\end{defi}

\begin{thm}\label{T:3.21} The following properties of a BJ-tree $J$ are
equivalent:

\begin{enumerate}
\item   $J$ is regular,

\item $J$ is described by a regular scheme,

\item $J$ is MS definable.
\end{enumerate}
\end{thm}

\begin{proof}
  \leavevmode
  \begin{enumerate}
\item $\!\!\Longrightarrow$(2) Let $J=\mathit{fgs}(J^{\prime})$
with $J^{\prime}$ denoted by a regular term $t$ in $T^{\infty}(F)$. Let $%
h:Pos(t)\rightarrow Q$ and $\tau$ be as in the definition of a regular term
in Section 1. Without loss of generality, we can assume that $h(Pos(t))=Q$.\
If this is not the case, we replace $Q$ by $h(Pos(t))$ and $\tau$ by its
restriction to this set.

\noindent \textbf{Claim.}
\begin{enumerate}
\item For each $u\in Pos(t)$, the arrangement $\overline {h}%
(W(t,u))=(\mathrm{Max}(t,ext,u),$ $\leq_{lex},h)$ over $Q$ is regular.

\item If $u^{\prime}$ is another position in $t$ and $h(u^{\prime})=h(u)$,
then $t/u^{\prime}=t/u$ and furthermore\footnote{Unless $u=u^{\prime}$, the sets $\mathrm{Max}(t,ext,u)$ and $\mathrm{Max}%
(t,ext,u^{\prime})$ are not equal, so that the arrangements $\overline{h}%
(W(t,u))$ and  $\overline{h}(W(t,u^{\prime}))$  are isomorphic but not
equal.} $\overline{h}(W(t,u^{\prime}))\simeq\overline{h}(W(t,u)).$
\end{enumerate}
Leaving its routine proof, we define $\Delta:=(Q,w_{Ax},(w_{q})_{q\in Q})$
as follows:

\begin{enumerate}[label=(\roman*)]
\item $w_{Ax}:=$ $\overline{h}(W(t,\varepsilon))$,
\item if $q\in Q,$ then $w_{q}:=\overline{h}(W(t,s(u)))$ where $s(u)$ is the
unique son of an occurrence $u$ of $ext$ such that $h(u)=q$; if $v$ is
another occurrence of $ext$ such that $h(v)=q$, then $h(s(v))=h(s(u))$ and
so by the claim, $\overline{h}(W(t,s(v)))\simeq\overline{h}(W(t,s(u)))$.\
Hence, $w_{q}$ is well-defined up to isomorphism.
\end{enumerate}

Informally, $\Delta$ is obtained from $\Delta(t)$ by replacing the labelling
mapping $Id$ of the arrangements $W(t,u)$ by $h$, so that these arrangements
are turned into arrangements $\overline{h}(W(t,u))$ over $Q$. Clearly, $%
\Delta$ is a regular scheme. As mapping $r$ showing that it describes $%
J^{\prime}$ (cf.~Definition 3.10), hence also $J$, we take the resriction
of $h$ to $\mathrm{Occ}(t,ext)$ that is the set of nodes of $J^{\prime}=val(t)$.

\item $\!\!\Longrightarrow$(3) is proved in Proposition \ref{P:3.12}.

\item $\!\!\Longrightarrow $(1) By Definition \ref{D:3.15}, the mapping 
$\alpha $ that
transforms the relational structure $\left\lfloor t\right\rfloor $ for $t$
in $T^{\infty }(F)$ into the BJ-tree $J=(N,\leq )=\mathit{fgs}(val(t))$ is
an MS-transduction because an MS\ formula can identify the nodes of $J$
among the positions of $t$ and another one can define 
$\leq $.

Let $J=(N,\leq )$ be an MS definable BJ-tree. It is, up to isomorphism, the
unique model of an MS sentence $\beta $. It follows by a standard argument%
\footnote{If $\alpha $ is an MS-transduction and $\beta $  is an MS sentence, then
the set of structures $S$ such that $\alpha (S)\models \beta $  is
MS-definable (Theorem \ref{T:2.1}).} that the set of terms $t$ in $T^{\infty }(F)$
such that $\alpha (\left\lfloor t\right\rfloor )\models \beta $ is MS
definable and thus, contains a regular term, a result by Rabin \cite%
{Tho90,Rab}. This term denotes $J$, hence $J$ is regular. 
\qedhere
\end{enumerate}
\end{proof}

\begin{cor}\label{C:3.22} The isomorphism problem for regular BJ-trees is
decidable.
\end{cor}

\begin{proof}
 A regular BJ-tree can be given, either by a regular term, a
regular scheme or an MS sentence. The proof of Theorem \ref{T:3.21} is effective~:
algorithms can convert any of these specifications into another one. Hence,
two regular BJ-trees can be given, one by an MS\ sentence $\beta $, the
other by a regular term $t$. They are isomorphic if and only if $\alpha
(\left\lfloor t\right\rfloor )\models \beta $ (cf.~the proof of (3)$%
\Longrightarrow $(1) of Theorem \ref{T:3.21}) if and only if $\left\lfloor
t\right\rfloor \models \beta ^{\prime }$ where $\beta ^{\prime }$ obtained
by applying Theorem \ref{T:2.1} to the sentence $\beta $ and the transduction $%
\alpha $. This is decidable \cite{Tho90,Rab}. 
\end{proof}

\subsection{Logical and algebraic descriptions of join-trees}

We now extend to join-trees the definitions and results of the previous
sections. Structured join-trees are defined in Section 3.1 (Definition
3.3). We extend to them the definitions and results of Sections 3.2-3.4. A
first novelty is that the argument of the extension operation $ext$ will be
an SJ-forest, equivalently a set of SJ-trees, instead of a single SBJ-tree.\
We will need an algebra with two sorts, the sort of SJ-trees and that of
SJ-forests. A second difference consists in the use in monadic second-order
formulas of a finiteness predicate (cf.~Section 2).

\bigskip

\begin{defi}[Description schemes for SJ-trees]\label{D:3.23}
\leavevmode
\begin{enumerate}[label=(\alph*)]
\item A \emph{description scheme for an SJ-tree}, in short an \emph{SJ-scheme,}
is a 5-tuple $\Delta = (Q, D, w_{Ax},$ $(m_{q})_{q\in Q},(w_{d})_{d\in D})$ such
that $Q,D$ are sets, $w_{Ax}\in \mathcal{A}(Q)$, $w_{d}\in \mathcal{A}(Q)$
for each $d\in D$ and $m_{q}=(M_{q},lab_{q})$ is a $D$-labelled set (cf.\
Section 2) for each $q\in Q$. Without loss of generality, we will assume
that the domains $V_{Ax}$ and $V_{d}$ of the arrangements $w_{Ax},w_{d}$ and
the sets $M_{q}$ are pairwise disjoint, because these arrangements and
labelled sets will be used up to isomorphism. Informally, $M_{q}$ encodes
the different lines $U$ such that $\widehat{U}=x$ where $x$ is labelled by $%
q $, and each of these lines is defined, up to isomorphism, by the
arrangement $w_{d}$ where $d$ is its label in $D$, defined by $lab_{q}$.

We say that $\Delta$ is \emph{regular} if $Q\cup D$ is finite and the
arrangements $w_{Ax}$ and $w_{d}$ are regular. The finiteness of $D$ implies
that each $D$-labelled set $m_{q}$ is regular.

\item Let $J=(N,\leq ,\mathcal{U})$ be an SJ-tree with axis $A$; for each $%
x\in N$, we denote by $\mathcal{U}^{x}$ the set of lines $U\in \mathcal{U}$
such that $\widehat{U}=x$. In the example of Figure 3, we have $\mathcal{U}%
^{d}=\{U_{4}\}$. An SBJ-scheme $\Delta $ as in a) \emph{describes} $J$ if
there exist mappings $r:N\rightarrow Q$ and $\widetilde{r}:\mathcal{U}%
-\{A\}\rightarrow D$  such that:

\begin{enumerate}[label=(b.\arabic*)]
\item the arrangement $(A,\leq,r)$ over $Q$ is isomorphic to $w_{Ax}$,

\item for each $x\in N$, the $D$-labelled set\footnote{$\mathcal{U}%
^{x}$ is a set of subsets of $N$ and $\widetilde{r}$ replaces each set in $%
\mathcal{U}^{x}$ by some $d\in D$. Hence, $\widetilde{r}(\mathcal{U}%
^{x}\dots) $ is a multiset of elements of $D$.} $(\mathcal{U}^{x},\widetilde{r}%
)$ is isomorphic to $m_{r(x)}$,

\item for each $U\in\mathcal{U}-\{A\}$, the arrangement $%
(U,\leq,r) $ over $Q$ is isomorphic to $w_{\widetilde{r}(U)}$.
\end{enumerate}
\end{enumerate}
\noindent We will also say that $\Delta $ \emph{describes} the join-tree $\mathit{fgs}%
(J):=(N,\leq ),$ obtained from $J$ by forgetting the structuring.
\end{defi}

\begin{prop}\label{P:3.24}
  \leavevmode
  \begin{enumerate}
\item Every SJ-tree is described by some
SJ-scheme.

\item Every SJ-scheme $\Delta $ describes a unique SJ-tree $\mathit{Unf}(\Delta )$ where unicity is up to isomorphism.
\end{enumerate}
\end{prop}

\begin{proof}
 We extend the proof of Proposition \ref{P:3.11}.

 \begin{enumerate}
\item Each SJ-tree $J=(N,\leq ,\mathcal{U})$ has a \emph{standard} description
scheme $\Delta (J):= (N,\mathcal{U}-\{A\},$ $(A,\leq ),$ $((\mathcal{U}%
^{x},Id))_{x\in N},((U,\leq ))_{U\in \mathcal{U}-\{A\}})$. The identity
mappings $N\rightarrow N$ and $\mathcal{U}-\{A\}\rightarrow \mathcal{U}%
-\{A\}$ show that $\Delta (J)$ describes $J$.

\item Let $\Delta =(Q,D,w_{Ax},(m_{q})_{q\in Q},(w_{d})_{d\in D})$ be an
SJ-scheme, defined with arrangements $w_{Ax}=(V_{Ax},\preceq ,lab_{Ax})$ and 
$w_{d}=(V_{d},\preceq ,lab_{d})$, and labelled sets $m_{q}=(M_{q},lab_{q})$
such that the sets $V_{Ax}$, $V_{d}$ and $M_{q}$ are pairwise disjoint and
the same symbol $\preceq $ denotes the orders of the arrangements $w_{Ax}$
and $w_{d}.$ We construct $\mathit{Unf}(\Delta):=(N,\leq,\mathcal{U})$ as follows.

\begin{enumerate}
\item  $N$ is the set of finite nonempty sequences $%
(v_{0},s_{1},v_{1},s_{2},\dots,s_{k},v_{k})$ such that:

\begin{quote}
$v_{0}\in V_{Ax},v_{i}\in V_{d_{i}}$ and $s_{i}\in M_{q_{i-1}}$ for $1\leq
i\leq k$, where

\noindent $q_{0}=lab_{Ax}(v_{0})$, $d_{1}=lab_{q_{0}}(s_{1})$, $%
q_{1}=lab_{d_{1}}(v_{1})$, $d_{2}=lab_{q_{1}}(s_{2})$, \dots, $%
q_{i}=lab_{d_{i}}(v_{i})$, $d_{i+1}=lab_{q_{i}}(s_{i+1})$ for $1\leq i\leq
k-1$.
\end{quote}

\item $(v_{0},s_{1},v_{1},\dots,s_{k},v_{k})\leq(v_{0}^{\prime},s_{1}^{\prime
},v_{1}^{\prime},\dots,s_{j}^{\prime},v_{j}^{\prime})$ if and only if

\begin{quote}
$k\geq j$, $(v_{0},s_{1},v_{1},\dots,s_{j})=(v_{0}^{\prime},s_{1}^{%
\prime},v_{1}^{\prime},\dots,s_{j}^{\prime})$ and $v_{j}\preceq v_{j}^{\prime}$
($v_{j},v_{j}^{\prime}\in V_{d_{j}}).$
\end{quote}

\item the axis $A$ is the set of one-element sequences $(v)$ for $v\in V_{Ax}$
and, for $x=(v_{0},s_{1},v_{1},\dots,s_{k},v_{k})$, $U(x)$ is the set of
sequences in $N$ of the form $(v_{0},s_{1},v_{1},s_{2},\dots,$ $s_{k},v)$ for $%
v\in V_{d_{k}}$, so that $\widehat{U(x)}%
=(v_{0},s_{1},v_{1},\dots,s_{k-1},v_{k-1})$.
\end{enumerate}
Note that $(v_{0},s_{1},v_{1},\dots,v_{k})<(v_{0},s_{1},v_{1},\dots,v_{j})$ if 
$j<k$ and that $(v_{0},s_{1},v_{1},$ $\dots,s_{k},v_{k})$
$\leq(v_{0},s_{1},v_{1},\dots,s_{k},v)$  if and only if $v_{k}\preceq v$.
In order to prove that $\Delta$ describes $J,$ we define $r:N\rightarrow Q$
and $\widetilde{r}:\mathcal{U}-\{A\}\rightarrow D$ as follows:

\begin{itemize}
\item if $x\in A$, then $x=(v)$ for some $v\in V_{Ax}$ and $r(x):=lab_{Ax}(v);$

\item if $x\in N$ has depth $k\geq1$, then $%
x=(v_{0},s_{1},v_{1},\dots,s_{k},v_{k})$ for some $v_{0},s_{1},\dots,s_{k},v_{k}$
 and $r(x):=lab_{d_{k}}(v_{k})$;

\item if $U\in\mathcal{U}-\{A\},$ then $U=U(x)$ for some $%
x=(v_{0},s_{1},v_{1},\dots,s_{k},v_{k})$, $k\geq1$, and $\widetilde{r}%
(U):=d_{k}$.
\end{itemize}

\noindent We check the three conditions of \ref{D:3.23}(b). We have $(A,\leq
,r)\simeq w_{Ax}$, hence (b.1) holds. For checking (b.2), we consider $%
x=(v_{0},s_{1},v_{1},\dots,s_{k},v_{k})$  $\in N,k\geq1$. The sets $U$ in $%
\mathcal{U}^{x}$ are those of the form $%
\{(v_{0},s_{1},v_{1},\dots,s_{k},v_{k},s,v)\mid v\in V_{d_{k+1}}\}$ for all $%
s\in M_{q_{k}}$ where $q_{k}=lab_{d_{k}}(v_{k})=r(x)$, hence (b.2) holds.
For checking (b.3), we let $U=U(x)$ for some $%
x=(v_{0},s_{1},v_{1},\dots,s_{k},v_{k}),k\geq1$; it is the set of sequences $%
(v_{0},s_{1},v_{1},s_{2},\dots,s_{k},v)$ for $v\in V_{d_{k}}$ ordered by $%
\preceq$ on the last components. Hence, $(U,\leq,lab_{d_{k}})$ is
isomorphic to $w_{d_{k}}$, which proves the property since $\widetilde {r}%
(U):=d_{k}$.

Unicity is proved as in Proposition \ref{P:3.11}. \qedhere
\end{enumerate}
\end{proof}

\noindent As for SBJ-trees, every SJ-tree is described by a canonical SJ-scheme, that
is regular and has a minimum number of states if the SJ-tree is regular.\
The following proposition extends Proposition \ref{P:3.12}.

\begin{prop}\label{P:3.25} A join-tree is MS$_{\mathit{fin}}$-definable if
it is described by a regular SJ-scheme.
\end{prop}

\begin{proof} Let $(N,\leq)$ be a join-tree $J$ (this property is
first-order expressible). Assume that $J=\mathit{fgs}(J^{\prime})$  where $%
J^{\prime}=(N,\leq,\mathcal{U})\simeq\mathit{Unf}(\Delta)$ for some regular
SJ-scheme $\Delta=(Q,D,w_{Ax},(m_{q})_{q\in Q},(w_{d})_{d\in D})$ such that 
$Q=\{1,\dots,m\}$ and $D=\{1,\dots,p\}$. Let $r$,$\widetilde{r}$ be the
corresponding mappings (cf. Definition \ref{D:3.23}(b)). For each $d\in D$, let $%
\psi_{d}$ be an MS sentence that characterizes $w_{d}$ up to isomorphism, by
the main result of \cite{Tho86}. Similarly, $\psi_{Ax}$ characterizes $%
w_{Ax} $.

A $D$-labelled set $m_{q}$ is described up to isomorphism by a $p$-tuple $%
(\smash{m_{q}^{1},\dots,m_{q}^{p}})$ where $\smash{m_{q}^{j}}$ is the number (possibly $%
\omega $) of elements having label $j$.
By Proposition \ref{P:3.7}, there is a bipartition $(N_{0},N_{1})$ of $N$ that
describes the structuring $\mathcal{U}$; from this bipartition, we can
define the axis $A$, the lines forming $\mathcal{U}$ and the node $\widehat{U%
}$ for each $U\in \mathcal{U}-\{A\}$ by MS formulas.
There is a partition $(Y_{1},\dots,Y_{m})$ of $N$ that describes $r$ by $%
Y_{q}:=r^{-1}(q)$. There is a partition $(Z_{1},\dots,Z_{p})$ where $Z_{j}$ is
the union of the lines $U\in \mathcal{U}-\{A\}$ such that $\widetilde{r}%
(U)=j.$

Consider a relational structure $(X,%
\leq,N_{0},N_{1},Y_{1},\dots,Y_{m},Z_{1},\dots,Z_{p}).$  By MS formulas, one
can express the following properties:
\begin{enumerate}[label=(\roman*)]
\item $(X,\leq,N_{0},N_{1})$ is $S(J^{\prime\prime})$ for some SJ-tree $%
J^{\prime\prime}=(X,\leq,\mathcal{U}^{\prime})$; its axis is denoted by $%
A^{\prime}$,

\item $(Y_{1},\dots,Y_{m})$ is a partition of $X$; we let $r(x):=q$ if and
only if $x\in Y_{q}$,

\item $(Z_{1},\dots,Z_{p})$ is a partition of $X$ such that each $Z_{j}$ is a
union of sets $U\in\mathcal{U}^{\prime}-\{A^{\prime}\}$ such that $%
(U,\leq,r)\simeq w_{j},$

\item $(A^{\prime},\leq,r)\simeq w_{Ax}$,

\item for each $q\in Q$ and $x\in Y_{q},$ the number of lines $U\in \mathcal{U}%
^{\prime x}$ that are contained in $Z_{j}$ is $m_{q}^{j}.$
\end{enumerate}

These formulas are constructed as follows: $\varphi (N_{0},N_{1})$ for (i)
is from Proposition \ref{P:3.7}. The formula for (ii) is standard. All other
formulas are constructed so as to express the desired properties when (i)
and (ii) do hold. For (iii), we use a suitable adaptation of $\psi _{i}$
and the fact from Proposition \ref{P:3.7} that, if (i) holds, we can define from $%
(N_{0},N_{1})$, by MS formulas, the axis $A^{\prime }$, the lines forming $%
\mathcal{U}^{\prime }$ and the node $\widehat{U}$ for each $U\in \mathcal{U}%
^{\prime }$. The mapping $r$ is given by $(Y_{1},\dots,Y_{m}).$ For (iv), we
do as for (iii) with $\psi _{Ax}$.

For (v), we do as follows. We write an MS\ formula $\gamma
(x,N_{0},N_{1},Z,W)$ expressing that $W$ consists of one node of each set $%
U\in \mathcal{U}^{\prime }-\{A^{\prime }\}$ that is contained in $Z$ and is
such that $\widehat{U}=x$. For any $x$ and $Z$, all sets $W$ satisfying $%
\gamma (x,N_{0},N_{1},Z,W)$ have same cardinality. Then, Property (v) holds
if and only if, for all $q=1,\dots,m$, $x\in Y_{q}$ and $j=1,\dots,p$, if $%
\gamma (x,N_{0},N_{1},Z_{j},W)$ holds, then $W$ has cardinality $m_{q}^{j}.$
If some number $m_{q}^{j}$ is $\omega $, we need the finiteness predicate $%
Fin(W)$ to express this condition\footnote{If the nodes of $J$ have degree at most $a\in \mathbb{N}$, then $%
m_{i}^{j}\leq a$ for all $i,j$ and the finiteness predicate is not needed,
hence, $J$ is MS definable.}.

Let $\beta(N_{0},N_{1},Y_{1},\dots,Y_{m},Z_{1},\dots,Z_{p})$ express conditions
(ii)-(v) in $(X,\leq)$. If a join-tree $(X,\leq)$ satisfies $\varphi
(N_{0},N_{1})\wedge\beta(N_{0},N_{1},Y_{1},\dots,Y_{m},Z_{1},\dots,Z_{p})$, it
has a structuring $\mathcal{U}^{\prime}$ described by $N_{0},N_{1}$: we let $%
J^{\prime\prime}:=(X,\leq,\mathcal{U}^{\prime})$. The sets $%
Y_{1},\dots,Y_{m},Z_{1},\dots,Z_{p}$ yield a scheme $\Delta$ that describes $%
J^{\prime\prime}$ (by Conditions (iii)-(v)), hence $J^{\prime\prime}$ is
isomorphic to $J^{\prime}$ by the unicity property of Proposition (3.24), \
and so, we have $(X,\leq)\simeq\mathit{fgs}(J^{\prime})=J$.

Hence, $J$ is (up to isomorphism) the unique model of the MS$_{\mathit{fin}}$
sentence:
\begin{equation}
\exists N_{0},N_{1}(\varphi(N_{0},N_{1})\wedge\exists
Y_{1},\dots,Y_{m},Z_{1},\dots,Z_{p}.%
\beta(N_{0},N_{1},Y_{1},\dots,Y_{m},Z_{1},\dots,Z_{p}))).
\tag*{\qEd}
\end{equation}
\def\popQED{}
\end{proof}
\noindent
Theorem \ref{T:3.30} will establish a converse.

\begin{defi}[Operations on SJ-trees and SJ-forests]\label{D:3.26}

We recall from Definition \ref{D:3.1} that a join-forest is the union of disjoint
join-trees. A structured join-forest (an SJ-forest, cf.~Definition \ref{D:3.3})
is the union of disjoint SJ-trees. It has no axis (each of its components
has an axis, but we do not single out any of them). We will use objects of
three types: join-trees, SJ-trees and SJ-forests, but a 2-sorted algebra
will suffice (similarly as above for $\mathbb{SBJT}$, we have not introduced
a separate sort for BJ-trees). The two sorts are $\boldsymbol{t}$ for
SJ-trees and $\boldsymbol{f}$ for SJ-forests.

\begin{itemize}

\item \emph{Concatenation of SJ-trees along axes.}
The \emph{concatenation} $J\bullet J^{\prime }$ of disjoint SJ-trees $J$ and 
$J^{\prime }$ is defined exactly as in Definition \ref{D:3.13} for SBJ-trees.

\item \emph{The empty SJ-tree }is denoted by the nullary symbol $\Omega _{%
\boldsymbol{t}}$.

\item \emph{Extension of an SJ-forest into an SJ-tree.}
Let $J=(N,\leq,\mathcal{U})$ be an SJ-forest and $u\notin N$. Then $%
ext_{u}(J)$ is an SJ-tree defined as in Definition \ref{D:3.13}. When handling
SJ-trees up to isomorphism, we will use the notation $ext(J)$ instead of $%
ext_{u}(J).$

\item \emph{The empty SJ-forest }is denoted by the nullary symbol $\Omega _{%
\boldsymbol{f}}$.

\item \emph{Making an SJ-tree into an SJ-forest. }

This is done by the unary operation $\mathit{mkf}$ that is actually
the identity on the triples that define SJ-trees.

\item \emph{The union of two disjoint SJ-forests} is denoted by $\uplus$.

\end{itemize}
\noindent
The types of these operations are thus:
\begin{align*}
\bullet&:\boldsymbol{t}\times\boldsymbol{t}\rightarrow\boldsymbol{t},
&
\Omega_{\boldsymbol{t}}&:\boldsymbol{t},
&
ext&:\boldsymbol{f}\rightarrow\boldsymbol{t},
\\
\uplus&:\boldsymbol{f}\times\boldsymbol{f}\rightarrow\boldsymbol{f},
&
\Omega_{\boldsymbol{f}}&:\boldsymbol{f},
&
\mathit{mkf}&:\boldsymbol{t}\rightarrow\boldsymbol{f}.
\end{align*}
In addition, we have, as in Definition \ref{D:3.13} the
\emph{Forgetting the structuring: }If $J$ is an SJ-tree, $\mathit{fgs}(J)$
is the underlying join-tree.
\end{defi}

\begin{defi}[The algebra $\mathbb{SJT}$]\label{D:3.27}

We let $F^{\prime}$ be the 2-sorted signature $\{\bullet,\uplus ,ext,\mathit{%
mkf},\Omega_{\boldsymbol{t}},$ $\Omega_{\boldsymbol{f}}\}$ where the types of
these six operations are as above. We obtain an $F^{\prime}$-algebra $%
\mathbb{SJT}$ whose domains are the sets of isomorphism classes of SJ-trees
and of SJ-forests. Concatenation is associative with neutral element $%
\Omega_{\boldsymbol{t}}$ and disjoint union is associative and commutative
with neutral element $\Omega_{\boldsymbol{f}}$.
\end{defi}

\begin{defi}[The value of a term]\label{D:3.28}

The definition is actually identical to that for SBJ-trees (Definition
\ref{D:3.15}). We recall it for the reader's convenience. The equivalence
relation $\approx$ is as in this definition.  The \emph{value} $%
val(t)=(N,\leq ,\mathcal{U})$ of $t\in T^{\infty}(F^{\prime})$ is defined as
follows:
\begin{itemize}
  \item $N:=\mathrm{Occ}(t,ext)$, the set of occurences in $t$ of $ext$,
  \item $u\leq v:\Longleftrightarrow u\leq_{t}w\leq_{lex}v$ for some $w\in
N$ such that $w\approx v,$
  \item $\mathcal{U}$ is the set of equivalence classes of $\approx.$
\end{itemize}
If $t$ has sort $\boldsymbol{t}$ (resp. $\boldsymbol{f}$) then $val(t)$ is
an SJ-tree (resp. an SJ-forest). It is clear that we have a value mapping: $%
T^{\infty}(F^{\prime})\rightarrow\mathbb{SJT}.$

For terms $t$ written with the operations $ext_{a}$, then $val(t):=(N,\leq ,%
\mathcal{U})$ where:

\begin{itemize}
  \item 
$N$ is the set of nodes $a$ such that $ext_{a}$ has an occurence in $t$,
actually a unique one, that we will denote by $u_{a}$,

\item $a\leq b:\Longleftrightarrow u_{a}\leq u_{b}$,

\item $a\approx b:\Longleftrightarrow u_{a}\approx u_{b}$, and

\item $\mathcal{U}$ is the set of equivalence classes of $\approx.$
\end{itemize}
\end{defi}

\begin{defi}[Regular join-trees]\label{D:3.29}

A join-tree (resp.~an SJ-tree) $T$ is \emph{regular} if it is denoted by $%
\mathit{fgs}(t)$ (resp.~by $t$) where $t$ is a regular term in $T^{\infty
}(F^{\prime})$ of sort $\boldsymbol{t}$.
\end{defi}

\begin{thm}\label{T:3.30} The following properties of a join-tree $J$ are
equivalent:
\begin{enumerate}
\item
$J$ is regular,

\item
 $J$ is described by a regular scheme,

\item
$J$ is MS$_{\mathit{fin}}$-definable.
\end{enumerate}
\end{thm}

\begin{proof}
  \begin{enumerate}
 \item $\!\!\Longrightarrow$(2). Similar to that of Theorem \ref{T:3.21}.

\item $\!\!\Longrightarrow(3)$ By Proposition \ref{P:3.25}.

\item $\!\!\Longrightarrow $(1) As in the proof of Theorem \ref{T:3.21}, the mapping 
$\alpha $ that transforms the relational structure $\left\lfloor
t\right\rfloor $ for $t$ in $T^{\infty }(F^{\prime })_{\boldsymbol{t}}$ (the
set of terms in $T^{\infty }(F^{\prime })$ of sort $\boldsymbol{t}$) into
the join-tree $J=(N,\leq )=\mathit{fgs}(val(t))$ is an MS-transduction. Let $%
J=(N,\leq )$ be an MS$_{\mathit{fin}}$-definable join-tree. It is, up to
isomorphism, the unique model of an MS$_{\mathit{fin}}$ sentence $\beta $.
The set $L$ of terms $t$ in $T^{\infty }(F^{\prime })_{\boldsymbol{t}}$
such that $\alpha (\left\lfloor t\right\rfloor )\models \beta $ is thus MS$_{%
\mathit{fin}}$-definable. However, since the relational structures $%
\left\lfloor t\right\rfloor $ have MS definable linear orderings, $L$ is
also MS definable (see Section 2), hence, it contains a regular term. This
term denotes $J$, hence $J$ is regular. 
\qedhere
  \end{enumerate}
\end{proof}
\noindent
The same proof as for Corollary \ref{C:3.22} yields:

\begin{cor}\label{C:3.31} The isomorphism problem for regular join-trees is
decidable.
\end{cor}
The rooted trees of unbounded degree, without order on the sets of sons of
their nodes are the join-trees defined by the terms in $T^{\infty}(F^{%
\prime }-\{\bullet\})_{\boldsymbol{t}}$. Theorem~\ref{T:3.30} and Corollary
\ref{C:3.31} hold for them.

\section{Ordered join-trees}

\begin{defi}[Ordered join-trees and join-hedges]\label{D:4.1}

Let $(N,\leq )$ be a join-forest. A direction relative to a node $x$ is a
maximal subset $C$ of $]-\infty ,x[$ such that $y\sqcup z<x$ for all $%
y,z\in C$ (cf.~Definition \ref{D:3.2}). The set of directions relative to $x$ is
denoted by $Dir(x)$. The notation $x\perp y$ means that $x$ and $y$ are
incomparable with respect to $\leq $, so that $x<x\sqcup y$ and $y<x\sqcup y$
if $x\perp y$ and $x\sqcup y$ is defined.

\bigskip

\begin{enumerate}[label=(\alph*)]
\item We say that a join-tree $J=(N,\leq)$ is \emph{ordered} (is an \emph{%
OJ-tree}) if each set $Dir(x)$ is equipped with a linear order $%
\sqsubseteq_{x}$. (In this way, we generalize the notion of an ordered tree,
cf.~Section 1.) From these orders, we define a single linear order $%
\sqsubseteq$ on $N$ as follows:
\begin{align*}
x\sqsubseteq y\text{ if and only if }&x\leq y\text{ or, }x\perp y\text{  and }\delta
  \sqsubset_{x\sqcup y}\delta^{\prime}
                \\
  &\text{ where }\delta,\delta^{\prime}\in
Dir(x\sqcup y), x\in\delta\text{ and }y\in\delta^{\prime}.
\end{align*}

\item The linear order $\sqsubseteq$ satisfies the following properties, for
all $x,y,x^{\prime},y^{\prime}$:

\begin{enumerate}[label={(\roman*)}]
\item $x\leq y$ implies $x\sqsubseteq y$,

\item if $x\leq y$, $x^{\prime}\leq y^{\prime}$ and $y\perp
y^{\prime}$, then $x\sqsubset x^{\prime}$ if and only if $y\sqsubset
y^{\prime}$.
\end{enumerate}

\begin{claim} If $J=(N,\leq)$ is a join-tree and $\sqsubseteq$\
is a linear order on $N$ satisfying conditions (i) and (ii), then $J$ is
ordered by the family of orders $(\sqsubseteq_{x})_{x\in N}$ such that, for
all $\delta,\delta^{\prime}$ in $Dir(x)$, we have $\delta\sqsubseteq_{x}%
\delta^{\prime}$ if and only if $\delta=\delta^{\prime}$ or $y\sqsubset
y^{\prime}$ for some $y\in\delta$ and $y^{\prime}\in\delta^{\prime}$ (if
and only if $\delta=\delta^{\prime}$ or $y\sqsubset y^{\prime}$ for all $%
y\in\delta$ and $y^{\prime}\in\delta^{\prime}).$
\end{claim}
\begin{proof}[Proof Sketch] Consider different directions $\delta,\delta^{\prime}%
\in Dir(x)$ such that $y\sqsubset y^{\prime}$ for some $y\in\delta$ and $%
y^{\prime}\in\delta^{\prime}$. We have also $y_{1}\sqsubset y_{1}^{\prime}$
for any $y_{1}\in\delta$ and $y_{1}^{\prime}\in\delta^{\prime}$ because $%
(y\sqcup y_{1})<x,(y^{\prime}\sqcup y_{1}^{\prime})<x$ and $(y\sqcup
y_{1})\perp(y^{\prime}\sqcup y_{1}^{\prime})$, hence, Condition (ii) implies
that $y\sqcup y_{1}\sqsubset y^{\prime}\sqcup y_{1}^{\prime}$ and $%
y_{1}\sqsubset y_{1}^{\prime}$.

Hence, each relation $\sqsubseteq_{x}$ is a linear order on $Dir(x)$. It
is clear that $\sqsubseteq$ is derived from the relations $\sqsubseteq_{x}$
by (a).
\end{proof}

It follows that an ordered join-tree can be equivalently defined as a triple 
$(N,\leq,\sqsubseteq)$ such that $(N,\leq)$ is a join-tree and $\sqsubseteq$
is a linear order that satisfies Conditions (i) and (ii). These conditions
are first-order expressible.

\item We define a \emph{join-hedge} as a triple $H=(N,\leq,\sqsubseteq)$ such
that $(N,\leq)$ is a join-forest and $\sqsubseteq$ is a linear order that
satisfies Conditions (i) and (ii). Let $J_{s}$, for $s\in S,$ be the
join-trees composing $(N,\leq)$. Each of them is ordered by $\sqsubseteq$
according to the above claim, and the index set $S$ is linearly ordered by $%
\sqsubseteq_{S}$ such that $s\sqsubset_{S}s^{\prime}$ if and only if $s\neq
s^{\prime}$ and $x\sqsubset y$ for all nodes $x$ of $J_{s}$ and $y$ of $%
J_{s^{\prime}}$. Hence $H$ is also a simple arrangement of pairwise disjoint
join-trees.
\end{enumerate}
\end{defi}

\begin{defi}[Structured join-hedges and structured ordered join-trees]\label{D:4.2}
  \leavevmode
\begin{enumerate}[label=(\alph*)]

\item A \emph{structured join-hedge}, an \emph{SJ-hedge} in short, is a
4-tuple $J=(N,\leq ,\sqsubseteq ,$ $\mathcal{U})$ such that $(N,\leq
,\sqsubseteq )$ is a join-hedge and $\mathcal{U}$ is a structuring of the
join-forest $(N,\leq )$.

A structured ordered join-tree could be defined in the same way, as an
OJ-tree $(N,\leq ,\sqsubseteq )$ equipped with a structuring $\mathcal{U}$.
However, we will need a refinement in order to define the operations that
construct ordered join-trees and join-hedges (cf. Definition  \ref{D:4.8} 
and Remark \ref{R:4.12} below).

\item Let $J$ be an OJ-tree $(N,\leq,\sqsubseteq)$ and $\mathcal{U}$ be a
structuring of $(N,\leq).$ For each node $x$, the set $Dir(x)$ of its
directions consists of the following sets:

\begin{itemize}
\item the sets $\mathord\downarrow (U)$ for each line $U\in \mathcal{U}^{x}$ (we recall
that $\mathord\downarrow (U):=\{y\mid y\leq z$ for some $z\in U\}$),

\item the set $\mathord\downarrow (U_{-}(x))$ (cf.~Definition 3.3) if $%
U_{-}(x)$ is not empty; in this case we call it the \emph{central direction}
of $x$.
\end{itemize}

If $x$ is the smallest element of $U(x)$, it has no central direction but $%
\mathcal{U}^{x}$ may be nonempty. It is clear that $\mathord\downarrow (U)\cap
\mathord\downarrow (U^{\prime })=\emptyset $ if $U$ and $U^{\prime }$ are distinct
lines in $\mathcal{U}^{x}$. We get a linear order on $\mathcal{U}^{x}$
based on that on directions, that we also denote by $\sqsubseteq _{x}$: we
have $U\sqsubset _{x}U^{\prime }$ if and only if $y\sqsubset y^{\prime }$
for all $y\in U$ and $y^{\prime }\in U^{\prime }.$

\item A \emph{structured ordered join-tree} (an \emph{SOJ-tree}) is a tuple $%
(N,\leq ,\sqsubseteq ,A,\mathcal{U}^{-},$ $\mathcal{U}^{+})$ such that $%
(N,\leq ,\sqsubseteq )$ is an OJ-tree and $\mathcal{U}:=\{A\}\uplus \mathcal{%
U}^{-}\uplus \mathcal{U}^{+}$ is a structuring of $(N,\leq )$ with axis $A,$
such that, for each node $x$:
%\begin{quote}
if $U\in\mathcal{U}^{x}\cap\mathcal{U}^{-}$ and $U^{\prime}\in\mathcal{U}%
^{x}\cap\mathcal{U}^{+}$, then $U\sqsubset_{x}U^{\prime}$ and furthermore,
if $x$ has a central direction $\delta$, then $U\sqsubset_{x}\delta\sqsubset
_{x}U^{\prime}$.
%\end{quote}
\end{enumerate}
We define then $Dir^{-}(x)$ as the set of directions $\mathord\downarrow(U)$ for $%
U\in\mathcal{U}^{x}\cap\mathcal{U}^{-}$ and, similarly, $Dir^{+}(x)$ with $%
U\in\mathcal{U}^{x}\cap\mathcal{U}^{+}.$

Let $U\in \mathcal{U}$ and $x\notin U$ be such that $[x,+\infty \mathclose\lbrack \cap
U\neq \emptyset $. By Condition (2) of Definition \ref{D:3.3}(a), there is a node $y_{i}$ in $U$ for some $i>0$ (we use the notation of that definition). We
say that $x$ is \emph{to the left} (resp.~\emph{to the right}) of $U$ if,
for some direction $\delta $ relative to $y_{i}$, we have $x\in \delta \in
Dir^{-}(y_{i})$ (resp. $x\in \delta \in Dir^{+}(y_{i}))$.

\end{defi}
\noindent
As in Propositions \ref{P:3.5} and \ref{P:3.9}, we have:

\begin{prop}\label{P:4.3} Every join-hedge and every ordered join-tree has
a structuring.
\end{prop}

\begin{proof} For a join-hedge $(N,\leq,\sqsubseteq)$, we take any
structuring $\mathcal{U}$ of the join-forest $(N,\leq)$.
Let $(N,\leq ,\sqsubseteq )$ be an OJ-tree and $\mathcal{U}$ be any
structuring of the join-tree $(N,\leq )$. Let $A$ be its axis. In order to
define $\mathcal{U}^{-}$ and $\mathcal{U}^{+}$, we need only partition each
set $\mathcal{U}^{x}$ into two sets $\mathcal{U}^{x}\cap \mathcal{U}^{-}$
and $\mathcal{U}^{x}\cap \mathcal{U}^{+}.$
If $x$ has a central direction $\delta $, we let $\mathcal{U}^{x}\cap 
\mathcal{U}^{-}$ consist of the lines $U$ in $\mathcal{U}^{x}$ such that $%
\orddownarrow (U)\sqsubset _{x}\delta $, and $\mathcal{U}^{x}\cap \mathcal{U}%
^{+}$ consist of those such that $\delta \sqsubset _{x}\orddownarrow (U)$.
Otherwise, we let $\mathcal{U}^{+}$ contain\footnote{We might alternatively
partition $\mathcal{U}^{x}$ into any two sets $%
\mathcal{U}^{x}\cap \mathcal{U}^{-}$ and $\mathcal{U}^{x}\cap \mathcal{U}^{+}$
such that $\mathcal{U}^{x}\cap \mathcal{U}^{-} \sqsubset _{x}
\mathcal{U}^{x}\cap \mathcal{U}^{+}$.} $\mathcal{U}^{x}$ so that $\mathcal{U}%
^{x}\cap \mathcal{U}^{-}=\emptyset $. 
\end{proof}

We now establish the MS definability of these structurings. If $J=(N,\leq
,\sqsubseteq ,A,\mathcal{U}^{-},\mathcal{U}^{+})$  is an SOJ-tree, we
define $S(J)$ as the structure $(N,\leq ,\sqsubseteq
,A,N_{0}^{-},N_{0}^{+},N_{1}^{-},N_{1}^{+})$ such that $A$ is the axis, $%
N_{0}^{-}$ (resp.~$N_{0}^{+}$) is the union of the lines $U\in \mathcal{U}%
^{-}$ (resp.~$U\in \mathcal{U}^{+}$) of even depth and $N_{1}^{-}$ (resp.~$%
N_{1}^{+}$) is the union of the lines $U\in \mathcal{U}^{-}$ (resp.~$U\in 
\mathcal{U}^{+}$) of odd depth.

\begin{prop}\label{P:4.4}
\begin{enumerate}[label=(\arabic*)]
\item There is an MS\ formula $\varphi
(A,N_{0}^{-},N_{0}^{+},N_{1}^{-},N_{1}^{+})$ expressing that a structure
$S=(N,\leq ,\sqsubseteq ,A,N_{0}^{-},N_{0}^{+},N_{1}^{-},N_{1}^{+})$ is $S(J) $
for some SOJ-tree \\
$J=(N,\leq ,\sqsubseteq , A,\mathcal{U}^{-}, \mathcal{U}^{+}).$

\item There exists an MS\ formula $\theta
^{-}(u,U,N_{0}^{-},N_{0}^{+},N_{1}^{-},N_{1}^{+})$ expressing in a structure  \\
${(N,\leq ,\sqsubseteq ,A,N_{0}^{-},N_{0}^{+},N_{1}^{-},N_{1}^{+})} = S(N,\leq
,\sqsubseteq ,A,\mathcal{U}^{-},$  $\mathcal{U}^{+})$  that $U\in 
\mathcal{U}^{-}\wedge u=\widehat{U}$; similarly, there exists an MS\ formula 
$\theta ^{+}(u,U,N_{0}^{-},N_{0}^{+},N_{1}^{-},N_{1}^{+})$ expressing that $%
U\in \mathcal{U}^{+}\wedge u=\widehat{U}$.
\end{enumerate}
\end{prop}

\begin{proof} Easy modification of the proof of Proposition \ref{P:3.7}.
\end{proof}

\begin{defi}[Description schemes for SOJ-trees]\label{D:4.5}
  \begin{enumerate}[label=(\alph*)]

\item A \emph{description scheme for an SOJ-tree}, in short an \emph{%
SOJ-scheme,} is a 6-tuple $\Delta =(Q,D,w_{Ax},(w_{q}^{-})_{q\in
Q},(w_{q}^{+})_{q\in Q},(w_{d})_{d\in D})$ such that $Q,D$ are sets, called
respectively the set of \emph{states }and of \emph{directions}, $w_{Ax}\in 
\mathcal{A}(Q)$, $(w_{d})_{d\in D}$ is a family of arrangements over $Q$ and 
$(w_{q}^{-})_{q\in Q}$ and $(w_{q}^{+})_{q\in Q}$ are families of
arrangements over $D$. Without loss of generality, we will assume that the
domains of these arrangements are pairwise disjoint, and the same symbol $%
\preceq $ denotes their orders. Informally, $(w_{q}^{-})_{q\in Q}$ and $%
(w_{q}^{+})_{q\in Q}$ encodes the sets of lines, ordered by $\sqsubseteq
_{x} $ of the two sets $Dir^{-}(x)$ and $Dir^{+}(x)$ where $x$ is labelled
by $q$.

We say that $\Delta$ is \emph{regular} if $Q\cup D$ is finite and the
arrangements $w_{Ax}$, $w_{d}$, $w_{q}^{-}$ and $w_{q}^{+}$  are regular.

\item Let $J=(N,\leq,\sqsubseteq,A,\mathcal{U}^{-},\mathcal{U}^{+})$ be an
SOJ-tree. An SOJ-scheme $\Delta$ as in (a) \emph{describes} $J$ if there
exist mappings $r:N\rightarrow Q$ and $\widetilde{r}:\mathcal{U}^{-}\cup 
\mathcal{U}^{+}\rightarrow D$  such that:

\begin{enumerate}[label=(b.\arabic*)]
\item $(A,\leq,r)\simeq w_{Ax}$,

\item for each $x\in N$, the arrangement $(\mathcal{U}^{x}\cap%
\mathcal{U}^{-},\sqsubseteq_{x},\widetilde{r})$ over $D$ is isomorphic to $%
w_{r(x)}^{-}$,

\item for each $x\in N$, the arrangement $(\mathcal{U}^{x}\cap%
\mathcal{U}^{+},\sqsubseteq_{x},\widetilde{r})$ over $D$ is isomorphic to $%
w_{r(x)}^{+}$,

\item for each $U\in\mathcal{U}^{-}\cup\mathcal{U}^{+}$, the
arrangement $(U,\leq,r)$ over $Q$ is isomorphic to $w_{\widetilde{r}(U)}$.
\end{enumerate}

We also say that $\Delta $ \emph{describes} the OJ-tree $\mathit{fgs}%
(J):=(N,\leq ,\sqsubseteq )$ where $\mathit{fgs}$ \emph{forgets the
structuring}.

\end{enumerate}
\end{defi}

\begin{prop}\label{P:4.6}
  \begin{enumerate}
  \item Every SOJ-tree is described by some
SOJ-scheme.

\item Every SOJ-scheme describes an SOJ-tree that is unique up to isomorphism.
\end{enumerate}
\end{prop}

\begin{proof}  
  \begin{enumerate}
\item The proof is similar to those of Propositions \ref{P:3.11} and  \ref{P:3.24}.

\item Let $\Delta =(Q,D,w_{Ax},(w_{q}^{-})_{q\in Q},(w_{q}^{+})_{q\in
Q},(w_{d})_{d\in D})$ be an SOJ-scheme, defined with arrangements $%
w_{Ax}=(V_{Ax},\preceq ,lab_{Ax}),$ $w_{d}=(V_{d},\preceq ,lab_{d})$, $%
w_{q}^{-}=(W_{q}^{-},\preceq ,lab_{q})$ and $w_{q}^{+}=(W_{q}^{+},\preceq
,lab_{q})$ such that the sets $V_{Ax},V_{d},W_{q}^{-}$ and $W_{q}^{+}$
are pairwise disjoint. Furthermore, we extend $\prec $ by letting $s\prec
s^{\prime }$ for all $s\in W_{q}^{-},s^{\prime }\in W_{q}^{+}$ and $q\in Q$.
We construct $J=\mathit{Unf}(\Delta )=(N,\leq ,\sqsubseteq ,A,\mathcal{U}%
^{-},\mathcal{U}^{+})$ as follows. Clauses a) to d) are essentially as in
Proposition \ref{P:3.24}.

\begin{enumerate}[label=\alph*)]
\item $N$ is the set of finite nonempty sequences $%
(v_{0},s_{1},v_{1},s_{2},\dots,s_{k},v_{k})$ such that:

\begin{quote}
$v_{0}\in V_{Ax},v_{i}\in V_{d_{i}}$ and $s_{i}\in W_{q_{i-1}}^{-}\cup
W_{q_{i-1}}^{+}$ for $1\leq i\leq k$, where

\noindent $q_{0}=lab_{Ax}(v_{0})$, $d_{1}=lab_{q_{0}}(s_{1})$, $%
q_{1}=lab_{d_{1}}(v_{1})$, $d_{2}=lab_{q_{1}}(s_{2})$, \dots, $%
q_{i}=lab_{d_{i}}(v_{i})$, $d_{i+1}=lab_{q_{i}}(s_{i+1})$ for $1\leq i\leq
k-1$.
\end{quote}

\item $(v_{0},s_{1},v_{1},\dots,s_{k},v_{k})\leq(v_{0}^{\prime},s_{1}^{\prime
},v_{1}^{\prime},\dots,s_{j}^{\prime},v_{j}^{\prime})$ if and only if:
\[
\text{$k\geq j$, $(v_{0},s_{1},v_{1},\dots,s_{j})=(v_{0}^{\prime
},s_{1}^{\prime },v_{1}^{\prime },\dots,s_{j}^{\prime })$ and $v_{j}\preceq
v_{j}^{\prime }$ ($v_{j},v_{j}^{\prime }\in V_{d_{j}}).$}
\]

\item The axis $A$ is the set of one-element sequences $(v)$ for $v\in V_{Ax}$.

\item If $x=(v_{0},s_{1},v_{1},\dots,s_{k},v_{k}),$ the line $U(x)$ is the set of
sequences $(v_{0},s_{1},$ $v_{1},s_{2},\dots,s_{k},$ $v)$  for $v\in V_{d_{k}}$%
; it belongs to $\mathcal{U}^{-}$ if $s_{k}\in W_{q_{k-1}}^{-}$ and to $%
\mathcal{U}^{+}$ if $s_{k}\in W_{q_{k-1}}^{+}$; in both cases, $\widehat {%
U(x)}=(v_{0},s_{1},v_{1},\dots,s_{k-1},v_{k-1})$.

\item $x=(v_{0},s_{1},v_{1},\dots,s_{k},v_{k})\sqsubseteq y=(v_{0}^{\prime
},s_{1}^{\prime },v_{1}^{\prime },\dots,s_{j}^{\prime },v_{j}^{\prime })$ if
and only if

either $x\leq y$ or, for some $\ell<\{j,k\},$ we have

\begin{enumerate}[label=e.\arabic*)]
\item $(v_{0},s_{1},v_{1},\dots,v_{\ell})=(v_{0}^{\prime},s_{1}^{\prime
},v_{1}^{\prime},\dots,v_{\ell}^{\prime})$ and $s_{\ell+1}\prec s_{\ell
+1}^{\prime},$

\item or $(v_{0},s_{1},v_{1},\dots,s_{\ell})=(v_{0}^{\prime},s_{1}^{%
\prime},v_{1}^{\prime},\dots,s_{\ell}^{\prime})$, $s_{\ell+1}\in
W_{q_{\ell}}^{-}$ and $v_{\ell}^{\prime}\prec v_{\ell}$,

\item or $(v_{0},s_{1},v_{1},\dots,s_{\ell})=(v_{0}^{\prime},s_{1}^{%
\prime},v_{1}^{\prime},\dots,s_{\ell}^{\prime})$, $s_{\ell+1}^{\prime}\in
W_{q_{\ell}}^{+}$ and $v_{\ell}\prec v_{\ell}^{\prime}.$
\end{enumerate}

In Case e.1), $x$ and $y$ are in different directions of $%
z:=(v_{0},s_{1},v_{1},\dots,v_{\ell})$ that are not its central direction; in
Case e.2), $x$ is to the left of the central direction $\delta$ of $z$ and $%
y\leq u$ where $u:=(v_{0},s_{1},v_{1},\dots,v_{\ell}^{\prime})$ is here below $%
z$ on $\delta$; in Case e.3), $y$ is to the right of the central direction $%
\delta^{\prime}$ of $u$ and $x\leq z$ where $z$ is below $u$ on $%
\delta^{\prime}.$

\end{enumerate}
In order to prove that $\Delta$ describes $J,$ we define $r:N\rightarrow Q$
and $\widetilde{r}:\mathcal{U}^{-}\cup\mathcal{U}^{+}\rightarrow D$ as
follows:

\begin{itemize}
\item if $x\in A$, then $x=(v)$ for some $v\in V_{Ax}$ and $r(x):=lab_{Ax}(v);$

\item if $x\in N$ has depth $k\geq1$, then $%
x=(v_{0},s_{1},v_{1},\dots,s_{k},v_{k})$ for some $v_{0},s_{1},\dots,s_{k},v_{k}$
and $r(x):=lab_{d_{k}}(v_{k})$;

\item if $U\in\mathcal{U}^{-}\cup\mathcal{U}^{+},$ then $U=U(x)$ for
some $x=(v_{0},s_{1},v_{1},\dots,s_{k},v_{k})$, and $\widetilde{r}(U):=d_{k}$.
\end{itemize}

In the last case, as $d_{k}=lab_{q_{k-1}}(s_{k})$, it depends only on $s_{k}$
and $v_{k-1}$ (via $q_{k-1}$). It follows that $\widetilde{r}(U)$ is the
same if we consider $U$ as $U(y)$ with $y=(v_{0},s_{1},v_{1},\dots,s_{k},v)$ \
hence, is well-defined.

We check the four conditions of Definition \ref{D:4.5}(b). We have $%
(A,\leq,r)\simeq w_{Ax}$, hence (b.1) holds. For (b.2) and (b.3), we
consider $x=(v_{0},s_{1},v_{1},\dots,s_{k},v_{k})\in N$. The sets $U$ in $%
\mathcal{U}^{x} $ are those of the form $%
\{(v_{0},s_{1},v_{1},\dots,s_{k},v_{k},s,v)\mid v\in V_{d_{k+1}}\}$ for all $%
s\in W_{q_{k}}^{-}\cup W_{q_{k}}^{+}$ where $q_{k}=lab_{d_{k}}(v_{k})=r(x)$,
hence (b.2) and (b.3) hold.

For checking (b.4), we let $U=U(x)$ for some $%
x=(v_{0},s_{1},v_{1},\dots,s_{k},v_{k}),k>0$; then $U$ is the set of
sequences $(v_{0},s_{1},v_{1},s_{2},\dots,s_{k},v)$ such that $v\in V_{d_{k}}$%
ordered by $\preceq$ on the last components. Hence, $(U,\leq,lab_{d_{k}})$
is isomorphic to $w_{d_{k}}$, which proves the property since $\widetilde{r}%
(U):=d_{k}$.

Unicity is proved as in Proposition \ref{P:3.11}.
\qedhere
\end{enumerate}
\end{proof}

\begin{prop}\label{P:4.7}
 An SOJ-tree is MS definable if it is described
by a regular SOJ-scheme.
\end{prop}

\begin{proof} Similar to the proofs of Propositions \ref{P:3.12} and \ref{P:3.25}.
\end{proof}

\bigskip

Note that, we need not the finiteness predicate as in Proposition \ref{P:3.25}
because we deal with arrangements that are linearly ordered structures, and
not with labelled sets. Next we define an algebra $\mathbb{SOJT}$ with two
sorts: $\boldsymbol{t}$ for SOJ-trees and $\boldsymbol{h}$ for SJ-hedges.

\bigskip

\begin{defi}[Operations on SOJ-trees and SJ-hedges]\label{D:4.8}
\leavevmode
\begin{itemize}
\item \emph{Concatenation of SOJ-trees along axes.}

\noindent Let $J_{1}=(N_{1},\leq_{1},\sqsubseteq_{1},A_{1},\mathcal{U}_{1}^{-},%
\mathcal{U}_{1}^{+})$ and $J_{2}=(N_{2},\leq_{2},\sqsubseteq_{2},A_{2},%
\mathcal{U}_{2}^{-},\mathcal{U}_{2}^{+})$ be disjoint SOJ-trees. We define
their concatenation as follows:
\begin{align*}
J_{1}\bullet J_{2} &:= (N_{1}\uplus N_{2},\leq,\sqsubseteq,A_{1}\uplus A_{2}, 
\mathcal{U}_{1}^{-}\uplus\mathcal{U}_{2}^{-},\mathcal{U}_{1}^{+}\uplus
\mathcal{U}_{2}^{+}) \text{ where } \\
x\leq y&:\Longleftrightarrow x\leq_{1}y\vee x\leq_{2}y\vee(x\in N_{1}\wedge
y\in A_{2}),\\
x\sqsubseteq y&:\Longleftrightarrow x\leq y\vee x\sqsubseteq_{1}y\vee
x\sqsubseteq_{2}y, \\
&\phantom{:\Longleftrightarrow~}\vee(x\bot y\wedge x\in N_{1}\wedge y\in N_{2}\wedge y\in U\in%
\mathcal{U}_{2}^{+}\cap\mathcal{U}_{2}^{x\sqcup y}) \\
&\phantom{:\Longleftrightarrow~}\vee(x\bot y\wedge x\in N_{2}\wedge y\in N_{1}\wedge x\in U\in%
\mathcal{U}_{2}^{-}\cap\mathcal{U}_{2}^{x\sqcup y}),\text{ for some }U.
\end{align*}
The relations $x\bot y$  and $x\sqcup y$ are relative to $\leq$.
It is clear that $J_{1}\bullet J_{2}$ is an SOJ-tree. Its axis is $%
A_{1}\uplus A_{2},\mathcal{U}^{+}=\mathcal{U}_{1}^{+}\uplus\mathcal{U}%
_{2}^{+}$ and $\mathcal{U}^{-}=\mathcal{U}_{1}^{-}\uplus\mathcal{U}_{2}^{-}. $
\emph{The empty SOJ-tree} is denoted by the nullary symbol $\Omega _{%
\boldsymbol{t}}$.

\item \emph{Extension of two SJ-hedges into a single SOJ-tree.}

Let $H_{1}=(N_{1},\leq_{1},\sqsubseteq_{1},\mathcal{U}_{1})$ and $%
H_{2}=(N_{2},\leq_{2},\sqsubseteq_{2},\mathcal{U}_{2})$ be disjoint
SJ-hedges and $u\notin N_{1}\uplus N_{2}$. Then:
\begin{align*}
ext_{u}(H_{1},H_{2})&:=(N_{1}\uplus N_{2}\uplus\{u\},\leq,\sqsubseteq ,\{u\},%
\mathcal{U}_{1},\mathcal{U}_{2}), \text{ where} \\
x\leq y&:\Longleftrightarrow x\leq_{1}y\vee x\leq_{2}y\vee y=u, \\
x\sqsubseteq y&:\Longleftrightarrow x\leq y\vee x\sqsubseteq_{1}y\vee
x\sqsubseteq_{2}y\vee(x\in N_{1}\wedge y\in N_{2}).
\end{align*}
Clearly, $ext_{u}(H_{1},H_{2})$ is an SOJ-tree, where $u$ has no central
direction. When handling SOJ-trees and SJ-hedges up to isomorphism, we
replace the notation  $ext_{u}(H_{1},H_{2})$ by $ext(H_{1},H_{2}).$

\item \emph{The empty SJ-hedge }is denoted by the nullary symbol $\Omega _{%
\boldsymbol{h}}$.

\item \emph{Making an SOJ-tree into an SJ-hedge.}

\noindent This is done by the unary operation $\mathit{mkh}$ such that, if $%
J=(N,\leq,\sqsubseteq,A,\mathcal{U}^{-},\mathcal{U}^{+})$ is an SOJ-tree,
then
\[
\mathit{mkh}(J):=(N,\leq,\sqsubseteq,\{A\}\uplus\mathcal{U}^{-}\uplus%
\mathcal{U}^{+}).
\]
Similarly as $\mathit{fgs}$, this operation forgets some information, here,
it merges three sets.
Note that in $\mathit{mkh}$($J)$, we distinguish neither $\mathcal{U}^{-}$
from $\mathcal{U}^{+}$ nor the axis $A$ from the other lines.

\item \emph{The concatenation of two disjoint SJ-hedges.}

  \noindent Let $H_{1}=(N_{1},\leq_{1},\sqsubseteq_{1},\mathcal{U}_{1})$ and
  $H_{2}=(N_{2},\leq_{2},\sqsubseteq_{2},\mathcal{U}_{2})$ be disjoint
SJ-hedges. Their \textqt{horizontal} concatenation is:
\begin{align*}
H_{1}\otimes H_{2}&:=(N_{1}\uplus N_{2},\leq_{1}\uplus\leq_{2},\sqsubseteq ,%
\mathcal{U}_{1}\uplus\mathcal{U}_{2})\text{ where }
\\
x\sqsubseteq y&:\Longleftrightarrow x\sqsubseteq_{1}y\vee
x\sqsubseteq_{2}y\vee(x\in N_{1}\wedge y\in N_{2}).
\end{align*}
We let $F^{\prime\prime}$ be the 2-sorted signature $\{\bullet,\otimes ,ext,%
\mathit{mkh},\Omega_{\boldsymbol{t}},\Omega_{\boldsymbol{h}}\}$ whose
operation types are:
\begin{align*}
\bullet&:\boldsymbol{t}\times\boldsymbol{t}\rightarrow\boldsymbol{t},
&
\Omega_{\boldsymbol{t}}&:\boldsymbol{t},
&
ext&:\boldsymbol{h}\times\boldsymbol{h}\rightarrow\boldsymbol{t},
\\
\otimes&:\boldsymbol{h}\times\boldsymbol{h}\rightarrow\boldsymbol{h},
&
\Omega_{\boldsymbol{h}}&:\boldsymbol{h},
&
\mathit{mkh}&:\boldsymbol{t}\rightarrow\boldsymbol{h}.
\end{align*}

\end{itemize}
In addition, we have, as in Definitions \ref{D:3.13} and \ref{D:3.26}:
\begin{itemize}
\item \emph{Forgetting the structuring:}
If $J=(N,\leq,\sqsubseteq,A,\mathcal{U}^{-},\mathcal{U}^{+})$ is an
SOJ-tree, then $\mathit{fgs}(J):=(N,\leq,\sqsubseteq)$ is the underlying\
OJ-tree.
\end{itemize}
\end{defi}

\begin{defi}[The value of a term]\label{D:4.9}

If $u$ is an occurrence of a binary symbol in a term $t$, we denote by $%
s_{1}(u)$ its first son and by $s_{2}(u)$ the second one (cf. Definition \ref{D:3.15}).
The value $val(t):=(N,\leq ,\sqsubseteq ,A,\mathcal{U}^{-},\mathcal{U}^{+})$
of a term $t\in T^{\infty }(F^{^{\prime \prime }})_{\boldsymbol{t}}$ is an
SOJ-tree defined in a similar way\footnote{Example \ref{E:4.10} will illustrate this definition.} as for $t\in T^{\infty
}(F^{\prime })_{\boldsymbol{t}}$, cf. Definitions \ref{D:3.15} and \ref{D:3.28}:

\begin{itemize}
\item $N:=\mathrm{Occ}(t,ext)$,

\item $x\leq y$ $:\Longleftrightarrow x\leq_{t}w\leq_{lex}y$ for some $%
w\in N$ such that $w\approx y,$

\item $A:=\mathrm{Max}(t,ext,\varepsilon)$,
\end{itemize}
where $\approx$ is the equivalence relation on $N$ defined as in Definition \ref{D:3.15}(a):
\begin{itemize}
\item $\mathcal{U}^{-}$ is the set of equivalence classes of $\approx$ of nodes
in $\mathrm{Max}(t,ext,s_{1}(u))$ for some occurrence $u$ of $ext$,

\item $\mathcal{U}^{+}$ is the set of equivalence classes of $\approx$
of nodes in $\mathrm{Max}(t,ext,s_{2}(u))$ for some occurrence $u$ of $ext$.
\end{itemize}
Hence, $U(x)\in\mathcal{U}^{-}$ if $x\leq_{t}s_{1}(\widehat{U(x)})$ and $%
U(x)\in\mathcal{U}^{+}$ if $x\leq_{t}s_{2}(\widehat{U(x)}).$

Next we define $x\sqsubseteq y:\Longleftrightarrow x\leq y$ or $x\bot y$ ($\bot$ is
relative to $\leq$, not to $\leq_{t}$) and we have one of the following
cases:

\begin{enumerate}[label=(\roman*)]
\item $x\sqcup_{t}y$ is an occurrence of $\otimes$ or $ext,$ $%
x\leq_{t}s_{1}(x\sqcup_{t}y)$ and $y\leq_{t}s_{2}(x\sqcup_{t}y),$

\item $x\sqcup_{t}y$ is an occurrence of $\bullet$, $x\leq_{t}s_{1}(x\sqcup
_{t}y)$ and $y\leq_{t}s_{2}(z)$ where $z$ is the unique maximal occurrence
of $ext$ such that $y<_{t}z\leq_{t}s_{2}(x\sqcup_{t}y)$,

\item $x\sqcup_{t}y$ is an occurrence of $\bullet$, $y\leq_{t}s_{1}(x%
\sqcup_{t}y)$ and $x\leq_{t}s_{1}(z)$ where $z$ is the unique maximal
occurrence of $ext$ such that $x<_{t}z\leq_{t}s_{2}(x\sqcup_{t}y)$.
\end{enumerate}
If $t\in T^{\infty}(F^{\prime\prime})_{\boldsymbol{h}}$ its value $val(t)$
is $(N,\leq,\sqsubseteq,\mathcal{U})$ with $(N,\leq,\sqsubseteq)$ defined as
above and $\mathcal{U}$ as in Definition \ref{D:3.28}.
\end{defi}

\begin{claim}
  \begin{enumerate}
 \item The mapping $val$ is a value mapping $T^{\infty
}(F^{^{\prime\prime}}):\rightarrow\mathbb{SOJT}$.

\item The transformation $\alpha$ of $\left\lfloor t\right\rfloor $ into $%
(N,\leq,\sqsubseteq)$ is an MS-transduction.
\end{enumerate}
\end{claim}

\begin{proof} (1) is clear from the definitions and (2) holds because the conditions of Definition \ref{D:4.9} are expressible in $\left\lfloor t\right\rfloor $ by 
MS formulas.
\end{proof}

\begin{figure}
[ptb]
\begin{center}
\includegraphics[
%natheight=8.196700in,
%natwidth=6.431600in,
%height=3.4359in,
width=2.7008in
]%
{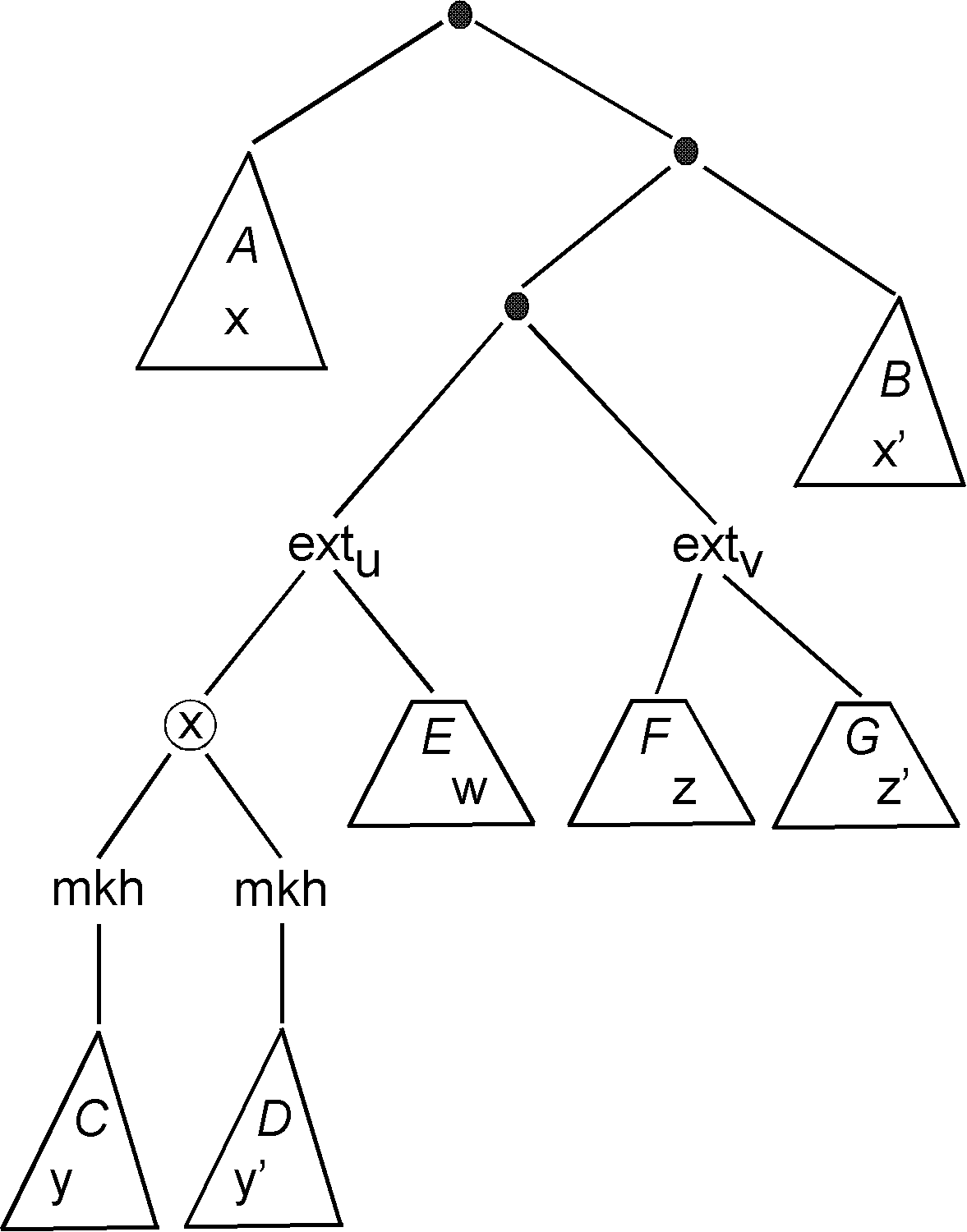}%
\caption{Term $T$ of Example (4.10).}%
\end{center}
\end{figure}

\bigskip

\begin{figure}
[ptb]
\begin{center}
\includegraphics[
%natheight=5.700000in,
%natwidth=7.071600in,
%height=2.8599in,
width=3.5423in
]%
{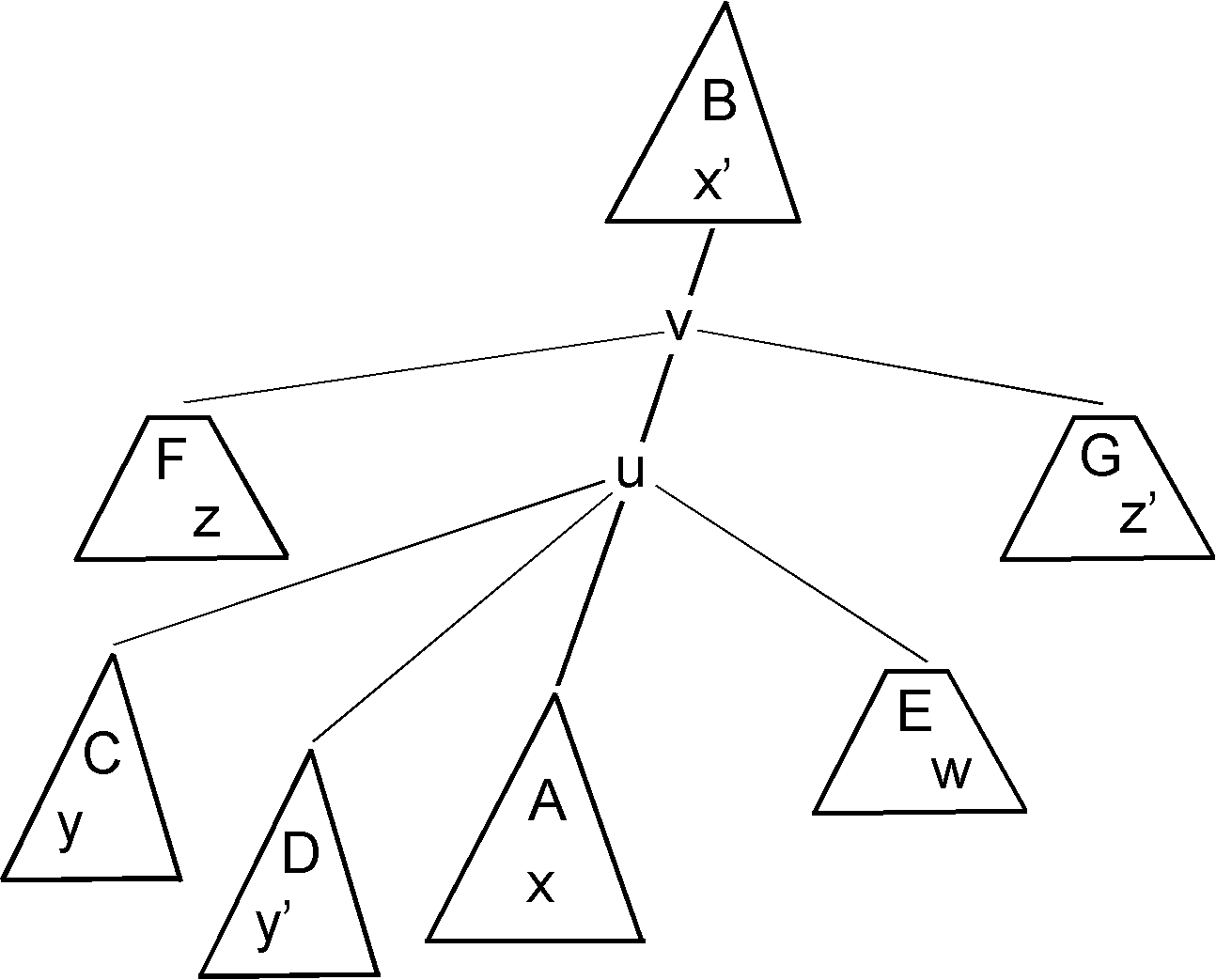}%
\caption{The OJ-tree $val(T)$ of Example (4.10).}%
\end{center}
\end{figure}

\begin{exa}\label{E:4.10} We now illustrate this definition. Figure 6 shows
a term $T$ where $A,B,C$ and $D$ are subterms of type $\boldsymbol{t}$ and $%
E,F$ and $G$ are subterms of type $\boldsymbol{h}$. They contain
occurrences of $ext$ that define nodes $x,x^{\prime},y,y^{\prime},w,z$ and $%
z^{\prime}$ of $val(T)$.

The OJ-tree $val(T)$ is shown on Figure 7, where we designate by $\mathrm{%
A,B,\dots,G}$ the trees and hedges defined by the terms $A,B,\dots,G$.
We have the following comparisons for $<$:
\begin{itemize}
\item $\{z,z^{\prime},u\}<v$, because $\{z,z^{\prime}\}<_{T}v$, $u<_{lex}v$ and $%
u\approx v$,

\item $\{y,y^{\prime},w\}<u$, because $\{y,y^{\prime},w\}<_{T}u$,

\item $x\leq\{u,v\}$ because $x\leq_{T}a<_{lex}\{u,v\}$ and $a\approx
u\approx v$ where $a$ is the root position of $A$,

\item $v<x^{\prime}$ if and only if $x^{\prime}$ is on $X$, the axis of $%
\mathrm{B}$, because in this case, $v\approx x^{\prime}$ and otherwise $v$
and $x^{\prime}$ are incomparable with respect to $\leq$; in all cases we
have $v<_{lex}x^{\prime}$.
\end{itemize}
For $\sqsubset$ we have:
%\begin{quote}
$z\sqsubset y\sqsubset y^{\prime}\sqsubset x\sqsubset w\sqsubset u\sqsubset
z^{\prime}\sqsubset v$ and $x^{\prime}\sqsubset z$ if $x^{\prime}$ is to the left of $X$ ;
 otherwise $v\sqsubset x^{\prime}$.
%\end{quote}
All inequalities for $<$ yield the corresponding inequalities for $\sqsubset 
$. We now compare $z,y,y^{\prime },x,w,z^{\prime }$ that are pairwise
incomparable for $<$.
\begin{itemize}
\item By Case (i) of Definition \ref{D:4.9}, we get $\{y,y^{\prime}\}\sqsubset
w,y\sqsubset y^{\prime}$ and $z\sqsubset z^{\prime}.$

\item By Case (ii), we get $x\sqsubset w,\{x,w\}\prec z^{\prime}$ and $%
\{y,y^{\prime }\}\prec w$.

\item By Case(iii) we get $\{z,y,y^{\prime}\}\prec x$ and $z\prec
\{y,y^{\prime}\}.$
\end{itemize}

Finally, if $x^{\prime}$ is to the left of $X$, then Case (iii) gives $%
x^{\prime}\sqsubset z,$ and if it to its right, then Case (ii) gives $z\sqsubset x^{\prime}$.
\end{exa}

\begin{thm}\label{T:4.11} 
The following properties of an OJ-tree $J$ are
equivalent~:
\begin{enumerate}
\item  $J$ is regular,

\item $J$ is described by a regular SOJ-scheme,

\item  $J$ is MS definable.

\end{enumerate}
\end{thm}

\begin{proof} The proof is similar to that of Theorem \ref{T:3.21}. We only
indicate some differences.
\begin{enumerate}
\item$\!\!\Longrightarrow $(3): Follows from Proposition \ref{P:4.7}.

\item$\!\!\Longrightarrow $(1): As observed in Definition \ref{D:4.9} (cf. the claim), the mapping $\alpha $
that transforms the relational structure $\left\lfloor t\right\rfloor $ for $%
t$ in $T^{\infty }(F^{^{\prime \prime }})_{\boldsymbol{t}}$ into the OJ-tree 
$(N,\leq ,\sqsubseteq )=\mathit{fgs}(val(t))$ is an MS-transduction. Let $%
J=(N,\leq ,\sqsubseteq )$ be an MS definable OJ-tree. It is, up to
isomorphism, the unique model of an MS sentence $\beta $. The set of terms $%
t$ in $T^{\infty }(F^{\prime \prime })_{\boldsymbol{t}}$ such that $\alpha
(\left\lfloor t\right\rfloor )\models \beta $ is thus MS definable, hence,
it contains a regular term. This term denotes $J$, hence $J$ is regular.
\qedhere
\end{enumerate}
\end{proof}

As in Corollaries \ref{C:3.22} and \ref{C:3.31},  we deduce that the isomorphism problem
for regular OJ-trees is decidable.

\begin{rem}[An alternative notion of SOJ-tree]\label{R:4.12} 

We present a variant of Definition \ref{D:4.2}. If $J=(N,\leq ,\sqsubseteq ,A,%
\mathcal{U}^{-},\mathcal{U}^{+})$ is an SOJ-tree, Definition \ref{D:4.2}(c) shows
that, for each $x\in N$, the partition $(\mathcal{U}^{x}\cap \mathcal{U}^{-},%
\mathcal{U}^{x}\cap \mathcal{U}^{+})$ of $\mathcal{U}^{x}$ is defined in a
unique way from $\sqsubseteq $ and the structuring $\mathcal{U}:=\{A\}\uplus 
\mathcal{U}^{-}\uplus \mathcal{U}^{+}$ of $(N,\leq ),$ except if $x$ has no
central direction (cf.~Proposition \ref{P:4.3}). This partition is useful only
when $x$ is the minimal element of $A$, denoted by $\min (A)$ when it
exists. To see that, we consider $J$ and another structuring of the same
OJ-tree, $J^{\prime }=(N,\leq ,\sqsubseteq ,A,\mathcal{U}^{\prime -},%
\mathcal{U}^{\prime +}),$ such that $\mathcal{U}^{\prime x}=\mathcal{U}^{x}$
for each node $x\neq \min (A)$ and $\mathcal{U}^{x}\cap \mathcal{U}^{-}\neq 
\mathcal{U}^{\prime x}\cap \mathcal{U}^{-}$ if $x=\min (A)$. Let $K$ be a
nonempty SOJ-tree. Then $K\bullet J$ is not equal to $K\bullet J^{\prime }$.

We could thus define an SOJ-tree as a tuple $S=(N,\leq ,\sqsubseteq ,%
\mathcal{U}^{\prime },\mathcal{U}_{Ax}^{-},\mathcal{U}_{Ax}^{+})$ such that $%
(N,\leq ,\sqsubseteq )$ is an OJ-tree, $\mathcal{U}:=\mathcal{U}^{\prime
}\uplus \mathcal{U}_{Ax}^{-}\uplus \mathcal{U}_{Ax}^{+}$ is a structuring of 
$(N,\leq )$ with axis $A$ (belonging to $\mathcal{U}^{\prime }$), $\mathcal{U%
}_{Ax}^{-}\uplus \mathcal{U}_{Ax}^{+}=\emptyset $ if $A$ has no minimal
element, and, otherwise, $\mathcal{U}_{Ax}^{-}\uplus \mathcal{U}_{Ax}^{+}=%
\mathcal{U}^{\min (A)}$ and $U\sqsubset U^{\prime }$ for all $U\in \mathcal{%
U}_{Ax}^{-}$ and $U^{\prime }\in \mathcal{U}_{Ax}^{+}.$ Then, the structure $%
S$ corresponding to an SOJ-tree  $(N,\leq ,\sqsubseteq ,A,\mathcal{U}^{-},%
\mathcal{U}^{+})$ is $(N,\leq ,\sqsubseteq ,\mathcal{U}^{\prime },\mathcal{U}%
_{Ax}^{-},\mathcal{U}_{Ax}^{+})$ where:
\begin{itemize}

\item $\mathcal{U}_{Ax}^{-}:=\mathcal{U}^{-}\cap \mathcal{U}^{\min (A)}$ and $%
\mathcal{U}_{Ax}^{+}:=\mathcal{U}^{+}\cap \mathcal{U}^{\min (A)}$ if $A$ has
a minimal element, 

\item $\mathcal{U}_{Ax}^{-}:=\mathcal{U}_{Ax}^{+}:=\emptyset $ otherwise, and, in
both cases,

\item $\mathcal{U}^{\prime }:=\{A\}\cup \mathcal{U}^{-}\cup \mathcal{U}^{+}-(%
\mathcal{U}_{Ax}^{-}\cup \mathcal{U}_{Ax}^{+}).$
\end{itemize}

It is not difficult conversely to construct $(A,\mathcal{U}^{-},\mathcal{U}%
^{+})$ from $(N,\leq ,\sqsubseteq ,\mathcal{U}^{\prime },\mathcal{U}%
_{Ax}^{-},\mathcal{U}_{Ax}^{+})$ and to redefine the operations of
Definition \ref{D:4.8} in terms of the structures $S$ as above. This alternative
definition of SOJ-trees contains no redundant information. However, we found
the initial definition easier to handle in our logical setting.  
\end{rem}

\section{Quasi-trees}
Quasi-trees can be viewed intuitively as ``undirected join-trees''. As in 
\cite{Cou14}, we define them in terms of a ternary betweenness relation.
Their use for defining rank-width is reviewed at the end of the section.

\begin{defi}[Betweenness]\label{D:5.1} \leavevmode
  \begin{enumerate}[label=(\alph*),beginpenalty=99]

\item Let $L=(X,\leq)$ be a linear order. Its \emph{betweenness relation} is
the ternary relation on $X$ such that $B_{L}(x,y,z)$ holds if and only if $%
x<y<z$ or $z<y<x.$ It is empty if $X$ has less than 3 elements.

\item If $T$ is a tree, its \emph{betweenness relation} is the ternary
relation on $N_{T}$, such that $B_{T}(x,y,z)$ holds if and only if $x,y,z$
are pairwise distinct and $y$ is on the unique path between $x$ and $z$. If 
$R$ is a rooted tree and $T=\mathit{Und}(R)$ is the tree obtained from $R$
by forgetting its root and edge directions, then
%
%\begin{quote}
$B_{T}(x,y,z): \Longleftrightarrow x,y,z$ are pairwise distinct and, either $%
x<_{R}y\leq_{R}x\sqcup_{R}z$ or $z<_{R}y\leq_{R}x\sqcup_{R}z$.
%\end{quote}

\item If $B$ is a ternary relation on a set $X$, and $x,y\in X$, then $%
[x,y]_{B}:=\{x,y\}\cup \{z\in X\mid B(x,z,y)\}.$ This set is finite if $%
B=B_{T}$ for some tree $T$.
\end{enumerate}
\end{defi}

\begin{prop}[cf.\ \cite{Cou14}]\label{P:5.2}\leavevmode\begin{enumerate}\item The betweenness relation $B$ of
a linear order $(X,\leq)$ satisfies the following properties for all $x,y,z,u \in X$.

\begin{enumerate}[label=A\arabic*:]
\item $B(x,y,z)\Rightarrow x\neq y\neq z\neq x.$

\item $B(x,y,z)\Rightarrow B(z,y,x).$

\item $B(x,y,z)\Rightarrow\lnot B(x,z,y).$

\item $B(x,y,z)\wedge B(y,z,u)\Rightarrow B(x,y,u)\wedge B(x,z,u).$

\item $B(x,y,z)\wedge B(x,u,y)\Rightarrow B(x,u,z)\wedge B(u,y,z).$

\item $B(x,y,z)\wedge B(x,u,z)\Rightarrow$
$y=u\vee \lbrack B(x,u,y)\wedge B(u,y,z)]\vee
\lbrack B(x,y,u)\wedge B(y,u,z)].$

\item[A7'\rlap{:}] $x\neq y\neq z\neq x\Rightarrow B(x,y,z)\vee B(x,z,y)\vee
B(y,x,z).$
\end{enumerate}

\item The betweenness relation $B$ of a tree $T$ satisfies the properties
A1-A6 for all $x,y,z,u$ in $N_{T}$ together with the following weakening of
A7':
\begin{enumerate}
\item[A7:] $x\neq y\neq z\neq x \Rightarrow
  \begin{array}[t]{l} B(x,y,z)\vee B(x,z,y)\vee B(y,x,z)\,\vee
  \\
  \exists w.(B(x,w,y)\wedge B(y,w,z)\wedge B(x,w,z))
  \end{array}.$
\end{enumerate}
\end{enumerate}

\end{prop}

\begin{prop}\label{P:5.3} Let $B$ be a ternary relation on a set $X$ that
satisfies properties A1-A7' for all $x,y,z,u$ in $X$. Let $a$ and $b$ be
distinct elements of $X$. There is a unique linear order $L=(X,\leq )$ such
that $a<b$ and $B_{L}=B$. It is quantifier-free definable in the logical
structure $(X,B,a,b).$
\end{prop}

\begin{proof}  Let $X,B,a,b$ be as in the statement. Let us enumerate $X$
as $x_{1}=a$, $x_{2}=b$, $x_{3},\dots,x_{n},\dots$ Let $X_{n}:=\{x_{1},\dots,x_{n}\}$\
for $n\geq 3$. Observe that $B\upharpoonright X_{n}$ satisfies properties
A1-A7'. We prove by induction on $n$ the existence and unicity of a linear
order $L_{n}=(X_{n},\leq )$ such that $a<b$ and $B_{L_{n}}=B\upharpoonright
X_{n}$.
\begin{itemize}
\item \emph{Basis}: $n=3.$ The conclusion follows from A7'.

\item  \emph{Induction case}: We assume the conclusion true for $n$.

\begin{claim} If $B(x_{i},x_{n+1},x_{j})$ holds for some $i<j\leq n$,
then, there is a unique pair $k,l$ such that $k<l\leq n$, $%
B(x_{k},x_{n+1},x_{l})$ holds and $[x_{k},x_{l}]_{L}\cap
X_{n+1}=\{x_{k},x_{n+1},x_{l}\}.$
\end{claim}

\begin{proof} Assume that $x_{m}\in \lbrack x_{i},x_{j}]_{L}\cap
X_{n+1}-\{x_{i},x_{n+1},x_{j}\},$ which implies $m\leq n$. Then, by A6, we
have $B(x_{i},x_{n+1},x_{m})$ or $B(x_{m},x_{n+1},x_{j})$ and we can repeat
the argument with $(x_{i},x_{m})$ or $(x_{m},x_{j})$ instead of $(i,j)$.
Furthermore, the considered set, $[x_{i},x_{m}]_{L}\cap X_{n+1}$ or $%
[x_{m},x_{j}]_{L}\cap X_{n+1}$ has less elements than $[x_{i},x_{j}]_{L}\cap
X_{n+1}. $Hence, we must obtain $k,l$ such that $[x_{k},x_{l}]_{L}\cap
X_{n+1}=\{x_{k},x_{n+1},x_{l}\}$ as desired. 
\end{proof}
In this case, there is a unique way to extend $L_{n}$ into $L_{n+1}$: we put 
$x_{n+1}$ between $x_{k}$ and $x_{l}.$ There is another case.

\begin{claim}If $B(x_{i},x_{n+1},x_{j})$ holds for no $i<j\leq n$, then
there is a unique $k$ such that $k\leq n$, $B(x_{l},x_{k},x_{n+1})$ holds
for some $l\leq n$, and $[x_{k},x_{n+1}]_{L}\cap X_{n+1}=\{x_{k},x_{n+1}\}.$
The element $x_{k}$ is \emph{extremal} in $L_{n}$, that is, either maximal
or minimal. \qed
\end{claim}

The proof is similar. In this case, there is a unique way to extend $L_{n}$
into $L_{n+1}$: we put $x_{n+1}$ after $x_{k}$ if it is maximal in $L_{n}$
or before it if it is minimal.
By taking the union of all orders $L_{n}$, we get the desired and unique
linear order on $X$, that we will denote by $\leq _{a,b}.$ We now define it
by a first-order formula.

We first observe a particularly simple case. If there are no $u,v\in X$
such that $B(u,b,v)$ holds, we have $x\leq _{a,b}y\Longleftrightarrow $ $%
x=y\vee y=b\vee B(x,y,b))$. A similar description can be given if there are
no $u,v$ such that $B(u,a,v)$ holds.
From $(X,B,a,b)$ as in the statement, we define the following binary
relation:
\begin{multline*}
Z(x,y):\Longleftrightarrow x\neq y \wedge {}
                             \\
\lbrack (B(x,a,b)\wedge \lnot B(y,x,a))\vee (x=a\wedge \lnot
B(y,a,b))\vee (B(a,x,b)\wedge \lnot B(y,x,b))\vee  {}
  \\
  \qquad(x=b\wedge B(a,b,y))\vee (B(a,b,x)\wedge B(b,x,y))].
\end{multline*}
It is easy to see that $x<_{a,b}y$ implies $Z(x,y)$. In particular, $Z(a,b)$
holds by the clause $x=a\wedge \lnot B(y,a,b)$ with $y=b$.
For the converse, assume that $Z(x,y)$ holds and $x<_{a,b}y$ does not.\
Then, we have $y<_{a,b}x$ because $\leq _{a,b}$ is a linear order. By
looking at the different relative positions of $x,y,a$ and $b$, we get a
contradiction. Hence $x\leq _{a,b}y$ if and only if $x=y\vee Z(x,y)$, which
is expressed by a quantifier-free formula $\xi (a,b,x,y).$
\qedhere
\end{itemize}
\end{proof}

\begin{rem} If there are no $u,v\in X$ such that $B(u,b,v)$ holds, then:
\begin{align*}
 Z(x,y)\Longleftrightarrow x\neq y \wedge \lbrack (B(x,a,b)\wedge
\lnot &B(y,x,a))\vee (x=a\wedge \lnot B(y,a,b))\vee 
\\
&(B(a,x,b)\wedge \lnot B(y,x,b))]
\end{align*}
which is equivalent to $y=b\vee B(x,y,b))$ as one can (painfully)  check by
using axioms A1-A7'.
\end{rem}

\begin{defi}[Quasi-trees]\label{D:5.4}
  \leavevmode
  \begin{enumerate}[label=(\alph*)]

\item A \emph{quasi-tree} is a structure $S=(N,B)$ such that $B$ is a ternary
relation on $N$, called the set of \emph{nodes}, that satisfies conditions
A1-A7 (a definition from \cite{Cou14}). To avoid uninteresting special
cases, we also require that $N$ has at least 3 nodes.

Lemma 11 of \cite{Cou14} shows that in a quasi-tree, the four cases of the
conclusion of A7 are exclusive and that, in the fourth one, there is a
unique node $w$ satisfying $B(x,w,y)\wedge B(y,w,z)\wedge B(x,w,z)$, which
we denote by $M_{S}(x,y,z)$.

A \emph{leaf} (of $S)$ is a node $z$ such that $B(x,z,y)$ holds for no $x,y.$
A \emph{line} is set of nodes $L$ such that $[x,y]_{B}\subseteq L$ if $%
x,y\in L$ and $(L,B\upharpoonright L)$ satisfies A7'. We say that $S$ is 
\emph{discrete} if each set $[x,y]_{B}$ is finite. A quasi-tree $S=(N,B)$ is 
\emph{a subquasi-tree of} a quasi-tree $S^{\prime }=(N^{\prime },B^{\prime
}) $, which we denote by $S\subseteq S^{\prime }$, if $N\subseteq N^{\prime
} $ and $B=B^{\prime }\upharpoonright N$. This condition implies that $%
M_{S}=M_{S^{\prime }}\upharpoonright N$.

\item From a join-tree $J=(N,\leq),$ we define a ternary relation $B_{J}$ on $%
N $ by:
\[
\text{$B_{J}(x,y,z):\Longleftrightarrow x\neq y\neq z\neq x$ and $(x<y\leq x\sqcup
z)\vee(z<y\leq x\sqcup z)$.}
\]
\end{enumerate}
\end{defi}

\begin{prop}\label{P:5.5} \leavevmode
  \begin{enumerate}
    \item The structure $\mathit{qt}(J):=(N,B_{J})$
associated with a join-tree $J=(N,\leq )$ having at least 3 nodes is a
quasi-tree. Every line of $J$ is a line of $\mathit{qt}(J)$. If $J$ is a
subjoin-tree of $J^{\prime }$, then $\mathit{qt}(J)$ is a subquasi-tree of $%
\mathit{qt}(J^{\prime })$.

\item Every quasi-tree $S$ is $\mathit{qt}(J)$ for some join-tree $J$.

\item A quasi-tree is discrete if and only if it is $\mathit{qt}(J)$ for some
rooted tree $J$.
\end{enumerate}
\end{prop}

\begin{proof} \leavevmode
  \begin{enumerate}[beginpenalty=99]
    \item Let $J=(N,\leq)$ be a join-tree with at least 3 nodes.

If it is finite, then it is $(N_{T},\leq _{T})$ for a finite rooted tree $T$, and $\mathit{qt}(J)$ is a finite quasi-tree by Proposition \ref{P:5.2}(b).

Otherwise consider distinct elements $x,y,z,u$ of $N$. We want to prove
that they satisfy A1-A7. There is a set $N^{\prime }\subseteq N$ of
cardinality at most 7 that contains $x,y,z,u$ and is closed under $\sqcup $.
Then $J^{\prime }=(N^{\prime },\leq \upharpoonright N^{\prime })$ is a
finite join-tree, $J^{\prime }\subseteq J$ and $\mathit{qt}(J^{\prime
})=(N^{\prime },B_{J}\upharpoonright N^{\prime })$ is a quasi-tree by the
initial observation, so that $x,y,z,u$ satisfy A1-A7 for $B=B_{J^{\prime }}$
hence, also for $B_{J}$. (The node $w$ that may be necessary to satisfy A7
may have to be chosen in the set $\{x\sqcup y,x\sqcup z,x\sqcup u,\dots\}$).
As $x,y,z,u$ are arbitrary, A1-A7 hold for $B_{J}$ and all $x,y,z,u\in N.$
Hence, $(N,B_{J})$ is a quasi-tree.

That every line of $J$ is a line of $\mathit{qt}(J)$ follows from the
definitions. (The converse does not hold.) The assertion about
subjoin-trees is also easy to prove.

\item Let $S=(N,B)$ be a quasi-tree and $r$ be any element of $N$. We define  $x\leq _{r}y:\Longleftrightarrow y\in \lbrack x,r]_{B}$. Then $(N,\leq _{r})$ is a join-tree $J$ with root $r$ and $S=\mathit{qt}(J)$
by Lemma~14 of \cite{Cou14}.

\item A quasi-tree $S=\mathit{qt}(J)$ is discrete if $J$ is rooted tree.
Conversely, if $S$ is a discrete quasi-tree, then $S=\mathit{qt}(T)$ for
some tree $T$ by Proposition 17 of \cite{Cou14}. By choosing any node as a
root, one makes $T$ into a rooted tree, and its betweenness relation is
that of~$T$. 
\qedhere
\end{enumerate}
\end{proof}

We say that a quasi-tree $S$ is \emph{described} by an
SJ-scheme if this scheme describes a join-tree $J$ such that $\mathit{qt}(J)=S$. It is \emph{regular} if it is $\mathit{qt}(J)$ for some regular
join-tree $J$.

\bigskip

\begin{prop}\label{P:5.6} 
A quasi-tree is MS$_{\mathit{fin}}$-definable if
it is described by a regular SJ-scheme.

\end{prop}

\begin{proof} We first prove a technical result.

\begin{claim} There exists a first-order formula $\mu (L,a,b,u,v)$ such
that, for every join-tree $J=(N,\leq ),$ if $S=\mathit{qt}(J)=(N,B),$ then
there is a subset $L$ of $N$ and elements $a,b$ of $N$ such that, for every $%
u,v\in N$, $(N,B)\models \mu (L,a,b,u,v)$ if and only if $u\leq v$.
\end{claim}

\begin{proof}[Proof of the claim] The formula $\mu (L,a,b,u,v)$ will be defined as $u=v\vee \mu
_{1}(L,a,b,u,v)\vee \mu _{2}(L,a,b,u,v)$ so as to handle two exclusive cases
relative to $J=(N,\leq )$.
\begin{itemize}
\item Case $J=(N,\leq )$ has a root $r$. Then, $u<v$ if and only if $%
v=r$ or $B(u,v,r)$. Hence, we define $\mu _{1}(L,a,b,u,v)$ as $L=\emptyset
\wedge a=b\wedge (v=b\vee B(u,v,b)).$

\item Case $J$ has no root. It has a line $L$ that is \emph{upwards
closed}, i.e., such that $y\in L$ if $x\leq y$ and $x\in L$. This
line has no maximal element (otherwise its maximal element would be a root
of $J$) and is infinite. Moreover, for every $u\in N$, we have $u<x$ for
some $x\in L$ (to prove that, consider $u\sqcup y$ where $y$ is any element
of $L$).
For all $u,v\in N$ we have:
\[
\text{$u<v\Longleftrightarrow \exists x,y\in L[x<y\wedge B(u,v,x)\wedge B(v,x,y)].$}
\]
If $u<v$ we have $u<v<x<y$ for some $x,y$ in $L$. Hence, we have $%
B(u,v,x)\wedge B(v,x,y).$
Assume for the converse that $x<y\wedge B(u,v,x)\wedge B(v,x,y)$ for some $%
x,y$ in $L$. We first prove that $u,v<x$. Since $B=B_{J}$, $%
B(v,x,y)\Longleftrightarrow $ $v<x\leq v\sqcup y\vee y<x\leq y\sqcup v$. As $%
x<y$, we must have $v<x$. Axiom A4 gives $B(u,x,y),$ from which we get
similarly $u<x$. From $B(u,v,x)$ we get $u<v\leq u\sqcup x$ or $x<v\leq
x\sqcup v$. As $v<x$, we have $u<v$.
Let $a,b\in L$ such that $a<b$. Proposition \ref{P:5.3} is applicable to $%
(L,B\upharpoonright L)$ that satisfies Conditions A1-A7'. Hence the
quantifier-free formula $\xi (a,b,x,y)$ defines $x<_{a,b}y$  for $x,y\in L$.
We define $\mu _{2}(L,a,b,u,v)$ as $\exists x,y\in L[a,b\in L\wedge \xi
(a,b,x,y)\wedge B(u,v,x)\wedge B(v,x,y)].$

\end{itemize}

We now complete the proof. If $J=(N,\leq )$ has a root $r$, we choose $%
L=\emptyset \wedge a=b=r$, $\mu _{2}(L,a,b,u,v)$ is false and $\mu
_{1}(L,a,b,u,v)$ is equivalent to $u<v$. If $J$ has no root, we let $L$ be
an upwards closed line, and $a,b\in L$ such that $a<b$. Then $\mu
_{1}(L,a,b,u,v)$ is false and $\mu _{2}(L,a,b,u,v)$ is equivalent to $u<v$.
\end{proof}
We let  $\mu ^{\prime }(L,a,b,u,v)$ be $u=v\vee \mu (L,a,b,u,v)$.
For proving the main assertion, we let $S=\mathit{qt}(J)$ be a quasi-tree
defined from a regular join-tree $J$ and $\psi $ be the MS$_{fin}$ sentence
expressing that a structure $(N,\leq )$ is a join-tree isomorphic to $J$; this
sentence exists by Theorem \ref{T:3.30}. Let $\varphi $ be the MS$_{fin}$
sentence expressing that a structure $(D,B)$ is a quasi-tree that satisfies $%
\exists L,a,b(\psi ^{\prime }\wedge $\textqt{$B$ is the betweenness relation of
  the order relation $\leq $ defined by $\mu ^{\prime }$}$)$ where $\psi
^{\prime } $ is obtained from $\psi $ by replacing each atomic formula $x\leq y$
by $%
\mu ^{\prime }(L,a,b,x,y)$.

Then $\mathit{qt}(J)$ satisfies $\varphi $ by the claim. If conversely, $%
(D,B)$ satisfies $\varphi $ for some $L,a$ and $b$, then $(D,\leq )$ is a
join-tree $J^{\prime }$, where $x\leq y$ is defined in $(D,B)$ by $\mu
^{\prime }(L,a,b,x,y)$ (because $\psi ^{\prime }$ holds), $(D,B)=\mathit{qt}%
(J^{\prime })$ (because $B$ is the betweenness relation of $\leq $) and $%
J^{\prime }\simeq J$ (because of $\psi ^{\prime }$). Hence, $(D,B)\simeq 
\mathit{qt}(J)$. Hence $\mathit{qt}(J)$ is characterized by $\varphi $ up to
isomorphism. 
\end{proof}

The next theorem establishes a converse. As algebra for quasi-trees, we take
the algebra $\mathbb{SJT}$ of join-trees together with the (external) \emph{operation} $\mathit{qt}$ (similar to $\mathit{fgs}$) that makes a join-tree
into a quasi-tree.

\begin{thm}\label{T:5.7} The following properties of a quasi-tree $S$ are
equivalent:
\begin{enumerate}
\item  $S$ is regular,

\item  $S$ is described by a regular SJ-scheme,

\item $S$ is MS$_{\mathit{fin}}$ definable.
\end{enumerate}
Furthermore, the isomorphism problem of regular quasi-trees is decidable.
\end{thm}

\begin{proof} \leavevmode %
\begin{enumerate}[]
\item[(1)]$\!\!\Rightarrow $(2): The proof is similar to that of
Theorem \ref{T:3.21}.

\item[(2)]$\!\!\Rightarrow$(3): By Proposition \ref{P:5.6}.

\item[(3)]$\!\!\Rightarrow $(1): The mapping $\alpha $ that transforms the
relational structure $\left\lfloor t\right\rfloor $ for $t$ in $T^{\infty
}(F^{\prime })_{\boldsymbol{t}}$ into the quasi-tree $S=\mathit{qt}(\mathit{%
fgs}(val(t)))$ is an MS-transduction by Definitions  \ref{D:4.9} (cf. the claim) and \ref{D:5.4}(b). The proof continues as in Theorem \ref{T:3.21}.
\end{enumerate}
The decidability of the isomorphism problem is as in Corollary \ref{C:3.22}. 
\end{proof}

We make these results more precise for subcubic quasi-trees: they are useful
for defining the rank-width of countable graphs, as we will recall.

\begin{defi}[Directions]\label{D:5.8}

Let $S=(N,B)$ be a quasi-tree and $x$ a node of $S$.
\begin{enumerate}[label=(\alph*)]

\item We say that $y,z\in N-\{x\}$ are the \emph{same direction relative to} $x
$ (or \emph{of} $x$) if, either $y=z$ or $B(y,z,x)$ or $B(z,y,x)$ or $%
B(y,u,x)\wedge B(z,u,x)$ for some node $u$ (a definition from \cite%
{Cou14}). Equivalently, $y\sqcup _{x}z<_{x}x$ ($<_{x}$ is as in Proposition  \ref{P:5.5}(2)). Hence, if $B(y,x,z)$ holds, then $y$ and $z$ are in different
directions relative to $x$. This relation is an equivalence, denoted by $%
y\sim _{x}z$, and its classes are the \emph{directions of} $x$.

\item The \emph{degree} of $x$ is the number of classes of $\sim _{x}$. A node
has degree 1 if and only if it is a leaf. We say that $S$ is \emph{subcubic}
if its nodes have degree at most 3. If $S=\mathit{qt}(T)$ for a tree $T$,
then a direction of $x$ is associated with each neighbour $y$ of $x$ and is
the set of nodes of the connected component of $T-\{x\}$ that contains $y$.

\item If $S=\mathit{qt}(J)$ for a join tree $J=(N,\leq )$, then, the
directions of $x$ in $S$ are those of $x$ in $J$ together with the set $%
N- \mathopen]-\infty ,x]$ if it is not empty. It follows that $S$ is subcubic if $J$\
is a BJ-tree.

\end{enumerate}
\end{defi}

\begin{lem}\label{L:5.9} Every subcubic quasi-tree is $\mathit{qt}(\mathit{fgs}%
(J))$ for some SBJ-tree $J$.
\end{lem}

\begin{proof} Let $S$ be a subcubic quasi-tree. Then $S=\mathit{qt}(J)$
for some join-tree $J$. We choose a maximal line $L$ of $J$ and $a,b\in L$
such that $a<b.$ By Proposition \ref{P:5.6}, the partial order $\leq $ of $J$ is
defined by $\mu ^{\prime }(L,a,b,x,y)$. The method of Proposition \ref{P:3.5} with 
$U_{0}:=L$, gives structuring $K$ of $J$, making it into an SBJ-tree as
defined in Definition \ref{D:3.8}.
\end{proof}

\begin{thm}\label{T:5.10} The following properties of a subcubic quasi-tree $%
S $ are equivalent~:
\begin{enumerate}
\item   $S$ is regular,

\item  $S$ is described by a regular SBJ-scheme,

\item  $S$ is MS definable.
\end{enumerate}
\end{thm}

\begin{proof} 
By Lemma \ref{L:5.9} and Proposition \ref{P:3.19}, every subcubic
quasi-tree $S$ is $\mathit{qt}(\mathit{fgs}(val(t)))$ for some term $t\in
T^{\infty}(F)$.

Property (1) means that $S=\mathit{qt}(\mathit{fgs}(val(t)))$ for some
regular term in $T^{\infty}(F^{\prime})_{\boldsymbol{t}}$. Let (1') mean
that $S=\mathit{qt}(\mathit{fgs}(val(t)))$ for some regular term in $%
T^{\infty}(F)$. Then (1')$\Longrightarrow$(2) by the similar implication in
Theorem \ref{T:3.21}.

(2)$\Longrightarrow$(3) by the similar implication in Theorem \ref{T:3.21} and the
observation that, in a quasi-tree $S$, the SBJ-trees $J$ such that $S=%
\mathit{qt}(\mathit{fgs}(J))$ can be specified by MS formulas in terms of a
5-tuple $(A,N_{0},N_{1},a,b)$ satisfying the formula $\varphi^{%
\prime}(A,N_{0},N_{1},a,b)$ of the proof of 
Proposition \ref{P:5.6}.

(3)$\Longrightarrow$(1') by the observation that the mapping $\alpha$ that
transforms the relational structure $\left\lfloor t\right\rfloor $ for $t$
in $T^{\infty}(F)$ into the subcubic quasi-tree $\mathit{qt}(\mathit{fgs}%
(val(t)))$ is an MS-transduction. The proof goes then as in Theorem \ref{T:3.21}.

The implication (1')$\Longrightarrow $(1) is trivial and (1) implies that $S$
is MS$_{\mathit{fin}}$ definable by Proposition \ref{P:5.6}. But a term $t\in
T^{\infty }(F)$ that defines $S$ is MS definable, and the relational
structure representing a term has an MS definable linear order. It follows
that $S$ has an MS definable linear order, hence that $S$ is MS definable by
the facts recalled in Section 2. 
\end{proof}

\bigskip

We now review the use of quasi-trees for rank-width \cite{Cou14}, a width
measure first defined and investigated in \cite{Oum} and \cite{OumSey} for
finite graphs.

\bigskip 

\begin{defi}[Rank-width for countable graphs]\label{D:5.11}

We consider finite or countable, loop-free, undirected graphs without
parallel edges. The \emph{adjacency matrix} of such a graph $G$ is $%
M_{G}:V_{G}\times V_{G}\rightarrow \{0,1\}$ with $M_{G}[x,y]=1$ if and only
if $x$ and $y$ are adjacent. If $U$ and $W$ are disjoint sets of vertices, $%
M_{G}[U,W]$ is the matrix that is the restriction of $M_{G}$ to $U\times W$%
. The \emph{rank } (over $GF(2)$) of $M_{G}[U,W]$ is the maximum
cardinality of an independent set of rows (equivalently, of columns) and is
denoted by $rk(M_{G}[U,W]);$ it belongs to $\mathbb{N}\cup \{\omega \}$. We
take $rk(M_{G}[\emptyset ,W])=rk(M_{G}[U,\emptyset ]):=0.$ If $X\uplus Y$ is
infinite, then $rk(M_{G}[X,Y])=\sup \{rk(M_{G}[U,W])\mid U\subseteq
X,W\subseteq Y$ and, $U$ and $W$ are finite\}.

A \emph{discrete layout} of a graph $G$ is an unrooted tree $T$ of maximal
degree 3 whose set of leaves is $V_{G}$. Its \emph{rank} is the least
upper-bound of the ranks $rk(M_{G}[X\cap V_{G},X^{c}\cap V_{G}])$ such that $%
X$ and $X^{c}:=N_{T}-X$ are the two connected components of $T$ minus one
edge. The \emph{discrete rank-width} of $G,$ denoted by $rwd^{dis}(G),$ is
the smallest rank of a discrete layout. If $G$ is finite, this value is the
rank-width defined in \cite{Oum}. By using for countable graphs $G$
quasi-trees with nodes of maximal degree 3 instead of trees, one obtains
their \emph{rank-with} $rwd(G)$ (cf.~\cite{Cou14} for details). We have $%
rwd(G)\leq rwd^{dis}(G)$. The notation $G\subseteq _{i}H$ means that $G$ is
an induced subgraph of $H$.

\end{defi}

\begin{thm}[cf.\ \cite{Cou14}]\label{T:5.12} For every graph $G$:
\begin{enumerate}[beginpenalty=99]
\item  if $H\subseteq_{i}G$, then $rwd(H)\leq rwd(G)$ and $rwd^{dis}(H)\leq
rwd^{dis}(G)$,

\item    \emph{Compactness}: $rwd(G)=\mathrm{Sup}\{rwd(H)\mid H\subseteq_{i}G$
and $H$ is finite$\},$

\item    \emph{Compactness with gap }: $rwd^{dis}(G)\leq2\cdot\mathrm{Sup}%
\{rwd(H)\mid H\subseteq_{i}G$ and $H$ is finite$\}.$
\end{enumerate}

\end{thm}

The \emph{gap function} in (3) is $n\mapsto 2n$, showing a weak form of
compactness. The proof of (2) uses Koenig's Lemma, and consists in taking $G$
as the union of an increasing sequence of finite induced subgraphs. The
desired layout of $G$ is obtained from an increasing sequence of finite
layouts of finite induced subgraphs where nodes are successively added. The
union of these layouts is in general a quasi-tree and not a tree.

\bigskip

\section{Conclusion}

We have defined regular join-trees of different kinds and regular
quasi-trees from regular terms. These terms have finitary descriptions.
Other infinite terms have finitary descriptions: the algebraic ones \cite{Cou83} and more generally, those of Caucal's hierarchy \cite{Blu+}. Such
terms also yield effective (algorithmically usable) notions of join-trees
and quasi-trees. It is unclear whether the corresponding isomorphism
problems are decidable\footnote{Z. \'{E}sik proved in \cite{E} that the isomorphism of the lexicographic
orderings of two context-free languages is undecidable. As algebraic linear
orders are defined from \emph{deterministic} context-free languages \cite{BE11}, deciding their isomorphism might be nevertheless possible.}.

The article \cite{Car} establishes that a \emph{set of arrangements} is
recognizable if and only if it is MS definable. One might wish to extend
this result to sets of join-trees and quasi-trees. An appropriate notion of
recognizability must be defined.

\end{document}